\documentclass[sigconf]{acmart}
\usepackage{amsmath} 
\usepackage{amsfonts}

\usepackage{amsmath}
\usepackage{amsfonts}
\usepackage{amsthm}
\usepackage{xfrac}

\usepackage{mathtools}
\usepackage{graphicx}
\usepackage{color}
\usepackage{xspace}
\usepackage[font={small}]{caption}
\usepackage{bm}
\usepackage{booktabs}
\usepackage{graphicx}
\usepackage{subcaption}
\usepackage{graphicx}
\usepackage{multirow}

\usepackage{wrapfig}
\usepackage{algorithm}
\usepackage{algpseudocode}
\usepackage{enumitem}
\usepackage{longtable} 
\usepackage{multirow,multicol}
\usepackage{array}
\usepackage{makecell}
\usepackage[table]{xcolor}
\newcommand{\nfeas}[1]{\cellcolor{red!12}{#1}}
\newcommand{\best}[1]{\textbf{#1}}

\newcommand{\techreport}[2]{#1} 

\newcommand{\cmp}{\ensuremath{\mathtt{cmp}}}
\newcommand{\com}{\ensuremath{\mathtt{com}}}

\newcommand{\taskVisionOne}{MLP-MN\xspace}
\newcommand{\taskVisionTwo}{RN56-CF10\xspace}

\newcommand{\taskLangOne}{G1-2B-MMLU\xspace}
\newcommand{\taskLangTwo}{G1-7B-MMLU\xspace}
\newcommand{\taskLangThree}{G1-2B-SGPT\xspace}
\newcommand{\taskLangFour}{G1-7B-SGPT\xspace}
\newcommand{\taskLangFive}{Ll3-8B-MMLU\xspace}
\newcommand{\taskLangSix}{Ll3-8B-WT\xspace}
\newcommand{\taskLangSeven}{FT5-SST2\xspace}

\newcommand{\scenLangSGTwoB}{L-SG1B\xspace}
\newcommand{\scenLangSGTwoBSevenN}{L-SG1B7N\xspace}

\newcommand{\scenLangSGSevenB}{L-SG7B\xspace}

\newcommand{\scenVisionTGH}{V-TGH\xspace}
\newcommand{\scenVisionTRS}{V-TRS\xspace}
\newcommand{\scenVisionLOO}{V-LOO\xspace}

\newcommand{\scenLangFiveShot}{L-5shot\xspace}

\newcommand{\scenLangWTEightB}{L-WT8B\xspace}

\newcommand{\scenVisionJetsonSST}{V-STT\xspace}
\newcommand{\scenVisionJetsonSSQ}{V-STQ\xspace}
\newcommand{\scenVisionJetsonSSL}{V-STL\xspace}
\newcommand{\scenLangJetsonSST}{L-STT\xspace}
\newcommand{\scenLangJetsonSSQ}{L-STQ\xspace}
\newcommand{\scenLangJetsonSSL}{L-STL\xspace}
\newcommand{\scenVisionJetsonS}{V-ST \xspace}
\newcommand{\scenVisionJetsonSS}{V-ST\xspace}
\newcommand{\scenLangJetsonS}{L-ST\xspace}
\newcommand{\scenLangJetsonSS}{L-ST\xspace}

\newcommand{\scenLangJetsonTestbedS}{J-L-S\xspace}

\newcommand{\scenVisionJetsonTestbedS}{J-V-S\xspace}
\newcommand{\scenVisionJetsonTestbedST}{J-V-ST\xspace}

\newcolumntype{V}[1]{>{\raggedright\arraybackslash}p{#1}}

\usepackage{ifthen}

\newcommand{\NewAlg}[3]{%
  \expandafter\newcommand\csname alg#1\endcsname[1][both]{%
    \ifthenelse{\equal{##1}{long}}{\textbf{#2}\xspace}{%
      \ifthenelse{\equal{##1}{short}}{\textsf{#3}\xspace}{%
        \textsf{#2 (#3)}\xspace
      }%
    }%
  }%
}

\NewAlg{MAX}{Max-compression}{MAX}
\NewAlg{NONE}{No-compression}{NONE}

\NewAlg{UNI}{Uniform-compression}{UNI}
\NewAlg{MYO}{Myopic-capacity}{MYO}
\NewAlg{CONS}{Conservative-capacity}{CONS}
\NewAlg{MA}{Moving-average}{MA}

\NewAlg{EQ}{Equal-share}{EQ}
\NewAlg{PROP}{Proportional-share}{PROP}
\NewAlg{PRIO}{Strict-priority}{PRIO}
\NewAlg{DES}{Decoupled-equal-split}{D-EQUAL}
\NewAlg{DLPROP}{Decoupled-$\lambda$-proportional}{D-L-PROP}

\NewAlg{OURSCSI}{CSI-Aware Optimal}{OURS-CSI}
\NewAlg{OURSNOCSIst}{Alg. 1 with LCB}{OURS-NOCSI} 
\NewAlg{OURSNOCSImt}{Alg. 2 with LCB}{OURS-NOCSI} 

\newcommand{\metU}{\overline{U}}            
\newcommand{\metD}{\overline{D}}            
\newcommand{\metED}{\overline{\Delta D}}    

\newcommand{\cmpT}{\texttt{T}\xspace}      
\newcommand{\cmpQ}{\texttt{Q}\xspace}      
\newcommand{\cmpIEight}{\texttt{I8}\xspace}

\newcommand{\scenMTLangMMLUTwo}{L-MMLU2}
\newcommand{\scenMTLangSGPTThree}{L-SGPT3}

\newcommand{\scenMTVisOneSym}{V1-Sym}
\newcommand{\scenMTVisTwoHetero}{V2-H}
\newcommand{\scenMTVisThreeKThreeTight}{V3-K3-T}
\newcommand{\scenMTVisFourWeights}{V4-w}
\newcommand{\scenMTVisFiveLooseNTenTHundred}{V3-LOO}
\newcommand{\scenMTVisFiveLooseNTenTFifty}{V5-ls-T50}
\newcommand{\scenMTVisSixTightRkNTenTHundred}{V6-tg-T100}
\newcommand{\scenMTVisSixTightRkNTenTFifty}{V6-tg-T50}
\newcommand{\scenMTVisCrossROne}{V-C-R1}

\newcommand{\scenMTLVJM}{LV-MT}

\newcommand{\scenMTJetsonTestbedM}{J-LV-M}

\newcommand{\scenTbSSGPT}{R-\scenLangSGTwoB}

\newcommand{\scenPrsntVizWire}{P-V-G}
\newcommand{\scenPrsntVizEdge}{P-V-E}

\newcommand{\scenPrsntLangWikiWire}{P-L1-G}
\newcommand{\scenPrsntLangWikiEdge}{P-L1-E}

\newcommand{\scenPrsntLangMmluEdge}{P-L2-E}

\newcommand{\scenPrsntVizWireEdge}{P-V-[G, E]}
\newcommand{\scenPrsntVizVizWireEdge}{P-VV-[G, E]}

\newcommand{\scenPrsntVizVizWire}{P-VV-G}
\newcommand{\scenPrsntVizVizEdge}{P-VV-E}

\newcommand{\scenPrsntVizLangEdge}{P-VL-E}
\newcommand{\scenPrsntVizLangWlan}{P-VL-W}

\newcommand{\scenPrsntLangLangEdge}{P-LL-E}

\newcommand{\netProfileLan}{1Gbps LAN}
\newcommand{\netProfileEdge}{W-Edge}
\newcommand{\netProfileWlan}{MANET}

\newcommand{\datasetGradFit}{\ensuremath{\mathcal{D}_{\mathrm{grad}}}}
\newcommand{\datasetTest}{\ensuremath{\mathcal{D}_{\mathrm{test}}}}

\newcommand{\CRedit}[1]{#1}

\theoremstyle{plain}
\newtheorem{theorem}{Theorem}
\newtheorem{lemma}{Lemma}
\newtheorem{problem}{Problem}

\theoremstyle{definition}

\theoremstyle{remark}

\AtBeginDocument{%
  \providecommand\BibTeX{{%
    \normalfont B\kern-0.5em{\scshape i\kern-0.25em b}\kern-0.8em\TeX}}}

\setcopyright{acmlicensed}
\copyrightyear{2026}
\acmYear{2026}
\acmDOI{XXXXXXX.XXXXXXX}

\acmConference[MobiHoc 2026]{the 27th International Symposium
on Theory, Algorithmic Foundations, and Protocol Design for Mobile Networks and Mobile Computing}{Nov. 23--26, 2026}{Tokyo, Japan}
\acmISBN{978-1-4503-XXXX-X/18/06}

\acmCodeLink{https://github.com/neu-spiral/communication-aware-inference}

\setcopyright{none}
\renewcommand\footnotetextcopyrightpermission[1]{}

\begin{document}

\techreport{%
  \fancyhead[LE,RO]{}%
  \renewcommand{\headrulewidth}{0pt}%
}{}

\title[Communication-Aware Model Distributed Inference via Latent Representation Compression]{Communication-Aware Model Distributed  Inference\\ via Latent Representation Compression}
\techreport{%
  \titlenote{This is an extended version of the work published at MobiHoc 2026.}%
}{}

\author{Peyman Gholami$^\S$, Theodoros-Thirimachos Davarakis$^\dagger$, Teng Li$^\S$, Miquel Sirera Perell\'{o}$^\dagger$,\\ Salil Reddy$^\ddagger$, Ayberk Yarkın Yıldız$^\dagger$, Anish Arora$^\ddagger$, Atilla Eryilmaz$^\ddagger$, Stratis Ioannidis$^\dagger$,\\ Chengzhang Li$^\ddagger$, Hulya Seferoglu$^\S$, Ness Shroff$^\ddagger$}

\authornote{$^\S$University of Illinois Chicago, $^\dagger$Northeastern University, $^\ddagger$The Ohio State University}

\affiliation{%
  \institution{}
  \country{}
}

\renewcommand{\shortauthors}{Gholami et al.}

\begin{abstract}
We study optimization of  distributed model inference over resource-constrained edge resources.
We propose a framework that optimizes the trade-off between model accuracy and communication costs by controlling latent representation compression to meet strict Quality of Service (QoS) throughput targets. For settings with known channel state information (CSI), we derive a closed-form optimal solution for single tasks and reduce the multi-task problem to a convex optimization program characterized by a per-link water-filling strategy. We extend these to handle unpredictable environments via a stochastic dual descent algorithm that relies only on causal channel estimates. We provide Lyapunov-based proofs demonstrating that our approach strictly satisfies long-term delay constraints while achieving a bounded optimality gap. Our results offer a robust, scalable blueprint for maximizing the performance of pipelined AI tasks in dynamic, resource-constrained distributed systems.
We verify the effectiveness of our proposed framework through simulations and experiments with real edge devices.
\end{abstract}

\maketitle
\section{Introduction}
Deploying machine learning (ML) algorithms over edge devices is critical for applications with stringent latency requirements. It is also appealing when availability, privacy, and rising cloud computing costs matter~\cite{hua2023edge}, and has motivated studies on model compression~\cite{han2015deep,lin2024awq,shen2024agile,zhu2024survey} 
and model distributed inference~\cite{li2024adaptive,li2023model,macario2025model}. 
Nevertheless, successfully harnessing computational resources at the edge gives rise to significant challenges, precisely because model distributed inference entails additional communication overheads. Several recent works have proposed compressing neural network activations transmitted between devices~\cite{li2024efficient,li2025two} precisely to reduce such communication overheads. This, in turn, comes at the cost of reduced inference accuracy.

\begin{figure}[!t]
    \centering
    \includegraphics[width=\techreport{0.37}{0.4}\textwidth]{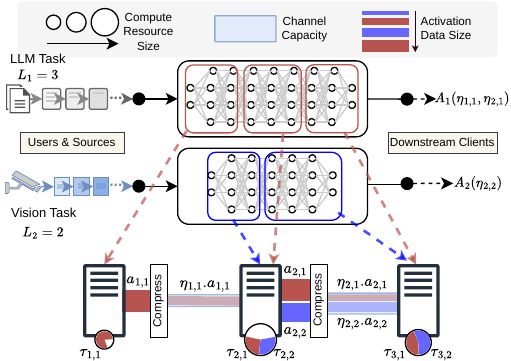}
    \vspace{-8pt}
    \techreport{%
  \vspace{-3pt}
}{}
    \caption{Distributed inference  processing \CRedit{input streams for}  two  tasks (LLM, partitioned into $L_1=3$ stages; Vision, partitioned into $L_2=2$ stages). Computation cycles ($\tau_{i,k}$) are assigned to the computational nodes. Intermediate activations ($a_{i,k}$) are compressed by a factor $\eta_{i,k}$ prior to transmission, yielding effective payloads $\eta_{i,k} a_{i,k}$ and determining the final end-to-end inference accuracy $A_k$. \CRedit{Given  a mapping of partitioned tasks over a topology and a compression scheme (expressed via  $A_k$), we identify optimal compression factors $\eta_{i,k}$ that maximize aggregate accuracy under throughput constraints.}} 
    \label{fig:system_model}
    \vspace{-11pt}
\end{figure}

To characterize this tradeoff, we study a setting in which the layers of multiple, diverse machine learning inference tasks are partitioned and subsequently distributed across a network of resource-constrained edge devices \CRedit{(see Fig.~\ref{fig:system_model})}. Each task  processes \CRedit{\emph{input streams}} generated at the edge: these could be, e.g., data \CRedit{continuously} collected from IoT sensors, cameras, LiDAR scenes, text inputs to an LLM, etc. \CRedit{As inputs arrive continuously, each task must be served at the rate at which its stream is generated, thereby imposing a throughput requirement.}  \CRedit{We target models that cannot fit on a single edge device, or that cannot meet accuracy and throughput requirements, even after static model compression. Instead, to speed up}  inference, devices parallelize execution via \emph{pipelining}:  after  processing a collection of layers, activations pass to a new device, and a new input is processed. 

\CRedit{For example, a vehicle processing a continuous high-rate LiDAR stream can execute early layers locally and send compressed activations to nearby compute, another smart car, a roadside unit, or an edge server, to complete inference.}
To meet throughput requirements in light of dynamic channel conditions, these communications can be \emph{compressed}, at the cost of accuracy degradation. Our contributions are as follows:

\sloppy
\begin{itemize}
\item We formalize the tradeoff between machine learning task accuracy vs.~throughput under pipeline parallelism over an arbitrary edge network topology. To the best of our knowledge, we are the first to introduce  and study this problem. 
\item 
Assuming that channel-state information (CSI) is known, we show that the corresponding optimization problem is tractable when  accuracy  is  monotone on compression rates in the single-task setting, and monotone and concave in the multi-task setting. 
\item Next, we study settings in which CSI is stochastic and a-priori unknown. Under the same assumptions of monotonicity and concavity, we show that  a stochastic dual descent algorithm can asymptotically recover the optimal compression ratios when having access to noisy, near-unbiased estimates of the CSI in the single-task setting.
\item For the multi-task setting with similar CSI estimates, the primal-step is non-convex. We show that a block-coordinate descent solver can be used as a tractable algorithm, and prove that, if it comes with additive approximation guarantees, the latter also translate to convergence guarantees for the end-to-end stochastic dual descent algorithm.  

\item We extensively  evaluate  our proposed CSI-aware and CSI-oblivious algorithms for a broad array of vision and language tasks, under top-$k$~\cite{ming2024adtopk,bian2024does}, quantization~\cite{lin2024awq,shen2024agile}, and LLM-int8~\cite{dettmers2022llmint8} compression schemes. Our  methods outperform competitors by maximizing accuracy while meeting throughput constraints.
We also deploy and evaluate our algorithms over three diverse testbeds, namely (a) a (highly-constrained) Raspberry Pi testbed, (b) a Jetson Orin Nano testbed, and (c) PRESCIENT, a GPU-dense network research testbed, demonstrating the applicability of our compression algorithms for scaling real-time inference deployments.  
\end{itemize}

\fussy
The remainder of this paper is structured as follows. We review related work in Section~\ref{sec:related}. In Section~\ref{sec:meth}, we provide the system model, CSI-aware design, estimated CSI design, and accuracy-compression function \& gradient estimation. Section \ref{sec:experiments} provides extensive evaluations of our proposed algorithms. Finally, we conclude in Section~\ref{sec:conc}.
\section{\label{sec:related} Related Work}



\textbf{Model Distributed Inference.}
The deployment of large-scale neural networks on resource-constrained edge devices has catalyzed the development of distributed inference architectures \cite{chen2019deep,li2023model,inferenceedge}. This distributed paradigm mitigates the strict memory constraints of individual edge devices while simultaneously providing substantial speedups through pipeline parallelism. This architecture has proven highly scalable, recently enabling the distributed inference of Large Language Models (LLMs) over consumer-grade internet connections \cite{macario2025model,borzunov2023distributed}. As a result, current literature is heavily focused on efficiently mapping and distributing computation workloads across available edge nodes \cite{wang2021pipeedge, li2024adaptive, li2025priority}. However, despite successfully utilizing parallelism to reduce computation time, these frameworks generally assume a strict single-task environment. More critically, they rely on transmitting massive, uncompressed intermediate activations between model partitions. In edge environments characterized by dynamic or degraded channel capacities, communicating these high-dimensional tensors becomes a critical bottleneck, severely bounding the end-to-end throughput.

\noindent\textbf{Communication-Efficient Inference and Activation Compression.}
To accelerate execution on resource-constrained edge devices, early efforts heavily focused on static model compression. Techniques such as weight pruning, knowledge distillation, and integer quantization \cite{han2015deep, hinton2015distilling} effectively reduce the memory footprint and computational requirements of neural networks prior to deployment \cite{zhu2024survey}. Extending this paradigm to distributed settings, recent work has explored communication-aware pruning \cite{jian2023communication} to explicitly co-optimize model sparsity and the resulting data transmission overhead. However, while these static methods are highly effective, they do not inherently resolve the dynamic network bottlenecks that arise when models are partitioned across multiple nodes. 
We explore an orthogonal dimension that can be used in addition to static methods such as, e.g., model pruning. We reduce the communication burden during model-distributed inference, by focusing on lightweight compression methods that operate directly on the intermediate outputs (activations) exchanged between workers at a partitioning point between layers (see Fig.~\ref{fig:system_model}).

\noindent\textbf{Activation Compression.} Several intermediate activation compression strategies have been proposed for both training and inference. Specifically, \textit{Top-$k$ sparsification} \cite{ramasinghe2025beyond,ming2024adtopk,bian2024does}
reduces payload sizes by transmitting only the most significant activation values while masking the rest to zero, though this often requires transmitting additional spatial index metadata. Alternatively, \textit{quantization}~\cite{choi2018pact,lin2024awq,shen2024agile} mitigates communication overhead by mapping high-precision floating-point activations to lower bit-width representations. This is typically achieved through fixed-point quantization or dynamically adjusted mixed-precision~\cite{lin2024awq,shen2024agile,dettmers2022llmint8}. Though the majority of these approaches accelerate training, recent literature has demonstrated the efficacy of targeted data compression schemes designed specifically to facilitate efficient and timely communication for network-edge classification tasks~\cite{li2024efficient, li2025two}. 
We balance communication reduction against the accuracy of the downstream inference task. 
This can be achieved by introducing a dynamically adjustable latent compression factor that continuously scales the intermediate payload size in response to fluctuating, real-time channel capacities, ensuring that strict QoS throughput constraints are satisfied without unnecessary degradation of inference accuracy.

\section{\label{sec:meth}Methodology}

\subsection{System Model}

We consider a model distributed inference system designed to process multiple inference tasks on the edge, as illustrated in Fig.~\ref{fig:system_model}. Several tasks are distributed through pipelining and executed continuously  across interconnected computational nodes \CRedit{over an arbitrary network} topology. 
Channel capacity between nodes is dynamic and, to meet throughput constraints, communication between consecutive nodes is compressed, leading to an accuracy–throughput tradeoff. We seek compression policies that maximize aggregate task utility under throughput constraints. \CRedit{We present this system model here in detail and discuss extensions  in \techreport{App.~\ref{app:extensions}.}{App.~L in \cite{techreport}.}}

\noindent\textbf{Computational Nodes \& Inference Tasks.} We model the underlying network topology as a directed graph $G = (\mathcal{V}, \mathcal{E})$, where $\mathcal{V}$ is the set of computational nodes and $\mathcal{E}$ is the set of directed communication links. Time is slotted, and each edge $e \in \mathcal{E}$ has a time-varying \emph{channel capacity}   $c_e(t) \in \mathbb{R}_+$ at time slot $t\in \mathbb{N}$. Assume a set of $K$ distinct \emph{inference tasks}.  The  set of tasks available at time slot $t$ is $\mathcal{K}(t) \subseteq [K]\triangleq \{1,\ldots,K\}$.  For every task $k \in \mathcal{K}(t)$, we assume a continuous stream of inference input samples (e.g., camera frames, LiDAR scenes, sensor measurements, etc.) to be processed for the duration of time slot $t$. To maximize hardware utilization, the system employs pipeline parallelism, overlapping computation and communication phases across all $L_k$ stages.

Formally, each task $k \in \mathcal{K}(t)$ corresponds to a neural network whose layers are partitioned into $L_k$ sequential \emph{stages}. These stages are assigned to a processing path, given by the sequence $P_k = (v_{1,k}, \dots, v_{L_k,k})$: physical node $v_{i,k} \in \mathcal{V}$ processes the $i$-th pipeline stage of the task ($i \in [L_k]$). Let $V_k \subseteq \mathcal{V}$ and $E_k \subseteq \mathcal{E}$ denote the sets of vertices and edges traversed by $P_k$, respectively, where the $i$-th transmission hop occurs over edge $e_{i,k} = (v_{i,k}, v_{i+1,k})$ for $i \in [L_k-1]$.
For notational convenience, we denote the capacity of the physical link traversed by this $i$-th hop as $c_{i,k}(t) \triangleq c_{e_{i,k}}(t)$.

\begin{figure*}[!t]
  \centering
    \begin{tabular}{@{}c@{\hspace{2.0em}}c@{\hspace{2.0em}}c@{}}
    \includegraphics[height=\techreport{3cm}{3.2cm}]{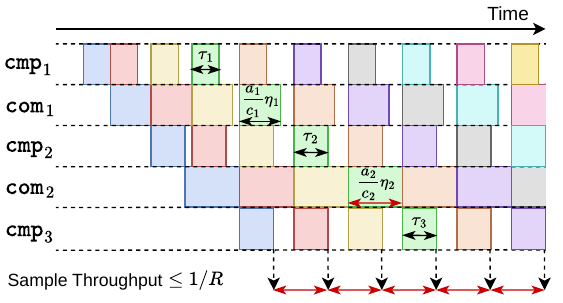} &
    \includegraphics[height=\techreport{3cm}{3.4cm}]{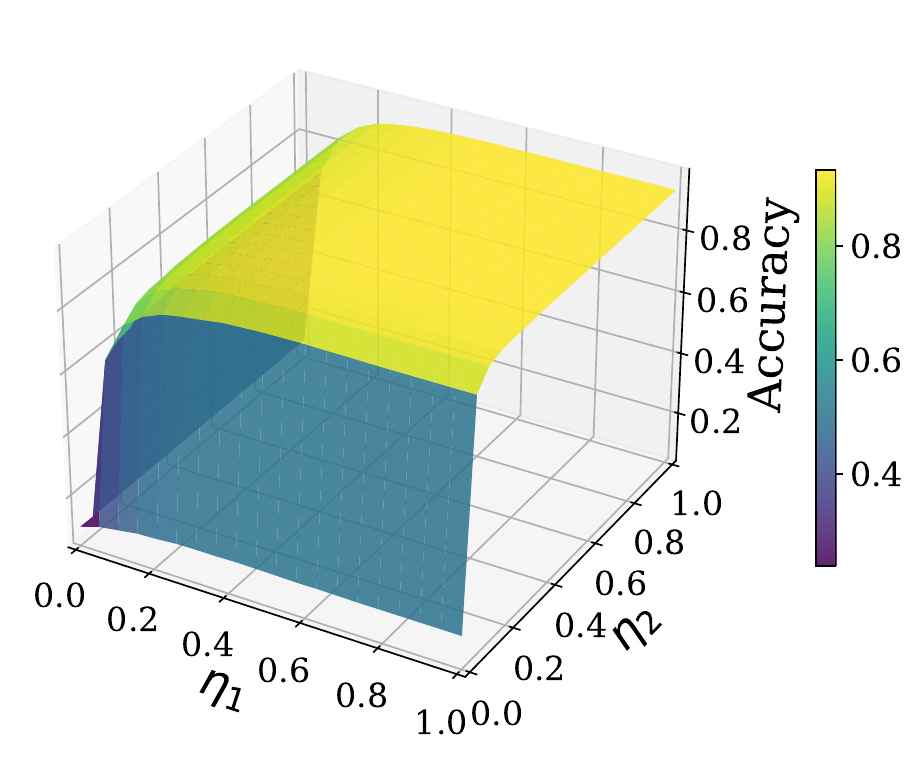} &
    \includegraphics[height=\techreport{3cm}{3.2cm}]{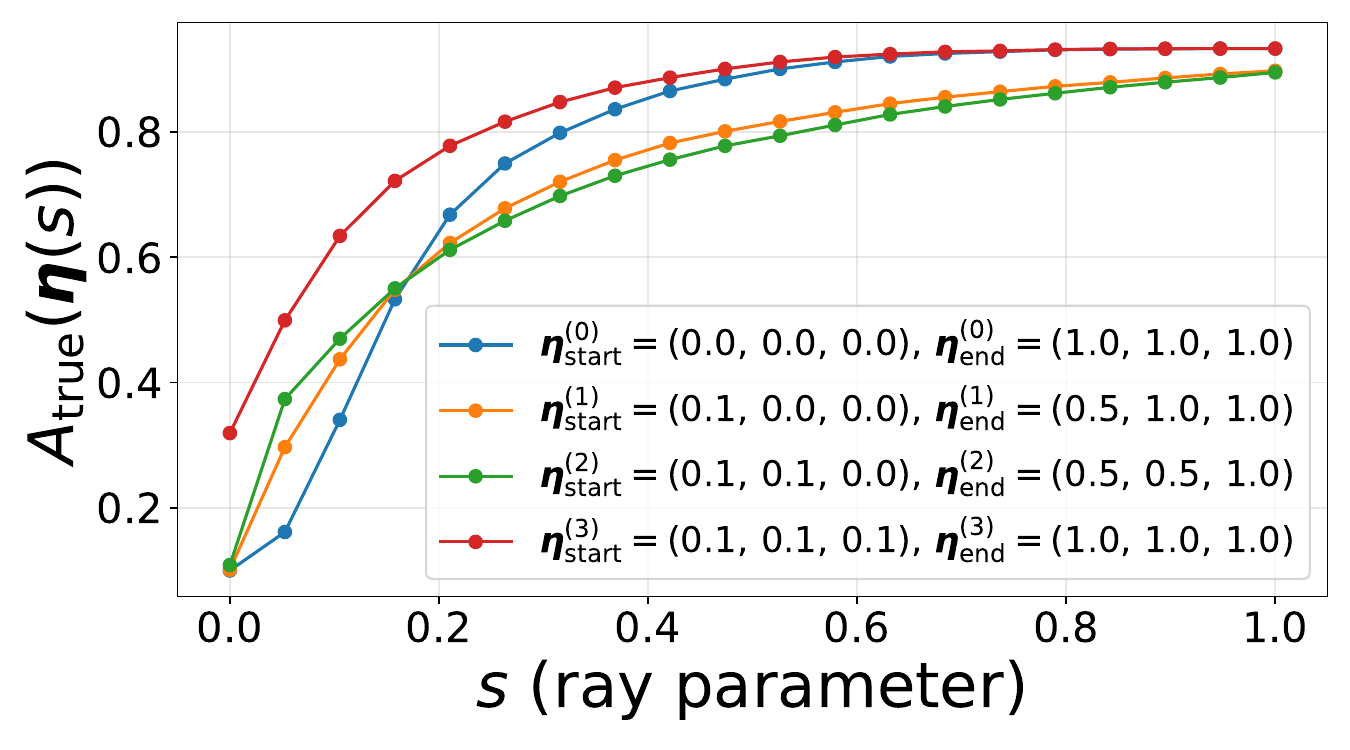} \\\vspace{-2.5ex}
    {\small (a)} & {\small (b)} & {\small (c)}
  \end{tabular}

  \caption{%
    (a) Illustration of pipeline parallelism for distributed inference of a single task.
    A task is partitioned over 3 stages and distributed across 3 nodes.
    Colors indicate different samples.
    Computation (\cmp$_i$) and communication (\com$_i$) phases at each node $i$ overlap.
    This leads to an end-to-end throughput bottlenecked by the maximum of stage processing
    times ($\tau_i$) and compressed transmission delays ($\frac{a_i}{c_i(t)}\eta_i(t)$).
    If compression is sufficiently high, computation dominates communication. If
      computation is also evenly partitioned across nodes 
   the resulting parallel speedup is $3\times$.
    (b)  Accuracy function $A(\boldsymbol{\eta})$ for ResNet-56~\cite{He2015resnet} on CIFAR-10
    with $L=3$ stages (two cut points) of top-$k$ sparsification~\cite{ming2024adtopk,bian2024does}, smoothed via KNN averaging.
    The function is monotone \& concave above $[\boldsymbol{\eta}_1,\boldsymbol{\eta}_2]=[0.1,0.1]$, and collapses below  this threshold.
    (c) To verify that the observed behavior in (b) is present in higher dimensional $\boldsymbol{\eta}$-spaces, we plot the directional profiles of $A(\boldsymbol{\eta})$ along random rays for the same ResNet-56 as in~(b), with $L=4$ stages (three cut points). Pictured is $A(\boldsymbol{\eta}(s))$ evaluated along $4$ linear segments in three dimensional $\boldsymbol{\eta}$-space, parameterized as $\boldsymbol{\eta}(s) = \boldsymbol{\eta}_{\mathrm{start}} + s\,\bigl(\boldsymbol{\eta}_{\mathrm{end}} - \boldsymbol{\eta}_{\mathrm{start}}\bigr)$, for $s \in [0,1].$
    Each  curve traces $A(\boldsymbol{\eta}(s))$ on a segment, as $s$ increases from $0$ to $1$.
    \CRedit{We observe that \emph{(i)} for certain segments, the accuracy $A$ collapses below a threshold, \emph{(ii)} $A$ is concave and monotone in $\boldsymbol{\eta}$ above the threshold.}
  }
  \label{fig:pipeline_and_accuracy}
  \vspace{-0pt}
\end{figure*}

We denote by $\tau_{i,k} \in \mathbb{R}_+$ the compute cycle time required to process an input sample through the $i$-th stage of task $k$ at node $v_{i,k}$. The output of the $i$-th stage (i.e., the activations of its last layer) has an uncompressed size $a_{i,k} \in \mathbb{R}_+$ (e.g., in MB). To ensure Quality of Service (QoS), the system must satisfy a throughput constraint:  sample inferences per second during time slot $t$ meet or exceed a target rate  $R_k(t)\in \mathbb{R}_+$ for task~$k$. Finally, each task is described by an \emph{accuracy-compression} function, which we describe next.

\noindent\textbf{Accuracy-Compression Function.} Since edge capacities fluctuate, to reduce communication overhead, a node may compress a stage's output (i.e., the activations of the last layer) prior to transmitting it to the next node. Several compression schemes are possible, e.g., top-$k$ sparsification \cite{ming2024adtopk,bian2024does}, quantization~\cite{lin2024awq,shen2024agile}, etc.~(see also Sec.~\ref{sec:expcompsch}). Given a compression scheme, we denote by $\eta_{i,k}(t) \in [\eta_{i,k}^{\min}, 1]$ the \emph{compression rate} applied to the activations of the $i$-th stage of task $k$ before transmission over edge $e_{i,k}$, resulting in an effective data size of $\eta_{i,k}(t) a_{i,k}$. Here, $\eta_{i,k}^{\min}$ denotes the minimum allowable compression rate for the $i$-th hop; compressing below this threshold causes the intermediate activations to collapse into noise, thereby degrading the inference accuracy to its absolute minimum (see also Fig.~\ref{fig:pipeline_and_accuracy}b).
Let $\boldsymbol{\eta}_k(t) = (\eta_{i,k}(t))_{1 \le i \le L_k-1}$ denote the vector of compression rates applied across the transmission hops of task $k$. We denote by $A_k(\boldsymbol{\eta}_k(t))\in \mathbb{R}_+$ the resulting \emph{inference accuracy}. We assume that $A_k(\cdot)$ is given for each task $k$; we describe how to estimate it from held-out data in Sec.~\ref{sec:acccompfun}.

\CRedit{Fig.~\ref{fig:pipeline_and_accuracy}(b) shows an example of an estimated function $A$ for  ResNet-56~\cite{He2015resnet} on CIFAR-10
    with $L=3$ stages (two cut points) of top-$k$ sparsification~\cite{ming2024adtopk,bian2024does}, smoothed via KNN averaging. Fig.~\ref{fig:pipeline_and_accuracy}(c) shows cross-sections evaluated over different rays. We observe that \emph{(i)} the accuracy collapses below a certain threshold and \emph{(ii)} is monotone concave above that threshold. We provide extensive  evidence of this behavior across different tasks in \techreport{App.~\ref{sec:concavity-validation}.}{App.~H in \cite{techreport}.}}

\noindent\textbf{Overall Objective. } Given  a mapping of partitioned tasks over an edge topology and a compression scheme (expressed via an accuracy-compression function per task), we wish to identify optimal compression factors that maximize aggregate accuracy, while meeting throughput constraints. 
In the remainder of this section, we formulate and propose solutions for this optimization problem under two distinct assumptions regarding channel state information (CSI), namely, a CSI-Aware  and an Estimated-CSI setting.

\subsection{CSI-Aware Design}

In our channel state information (CSI) aware setting, we assume that channel capacities   $c_e(t) \in \mathbb{R}_+$, $e \in \mathcal{E}$, are provided at the beginning of time slot $t$. As a warmup, we first describe the optimization problem and the solution that arises  under a single task.

\subsubsection{CSI-Aware Case -- Single Task}
When only a single task is present in the system, the optimization reduces to identifying optimal compression rates that meet throughput constraints. For convenience, we drop the task index $k$ from our notation, and assume w.l.o.g. that the nodes $i$ in the path $P$ the task is assigned to are indexed as $1,2,\ldots,L$.
Recall that a stream of inference requests occurs at each timeslot using pipeline parallelism, that overlaps computation and communication phases across different samples.  Under this model, the average  throughput is constrained by its slowest component, either computation or communication (see also Fig.~\ref{fig:pipeline_and_accuracy}a).
This leads to the following optimization problem:
\begin{problem}
[CSI-Aware Single-Task]\label{prob:single-task}
\begin{subequations} \label{prob:max-main}
\begin{align}\mathop{\text{Maximize:}}_{\boldsymbol{\eta}(t)} & \quad A(\boldsymbol{\eta}(t))
 \\[0.5ex]
  \text{subj.~to:}& \quad
     \max\bigl\{
      \max_{i\in [L]}\tau_i,\,
      \max_{i\in [L-1]} \tfrac{a_i}{c_{i}(t)}\eta_i(t)
    \bigr\}
    \;\le\; \tfrac{1}{R(t)} ,
    \label{cons:max}\\
    &\quad \boldsymbol{\eta}(t) \in \textstyle\prod_{i=1}^{L-1} [\eta_i^{\min}, 1].
\end{align}
\end{subequations}
\end{problem}
In short, \eqref{cons:max} states that \emph{all} computation times $\{\tau_i\}_{i=1}^L$ and \emph{all} transmission times $\{\frac{a_i}{c_{i}(t)}\eta_i(t)\}_{i=1}^{L-1}$ per stage must be bounded by the inverse of the target throughput $R(t)$. When the neural network is equipartitioned, reducing communication sufficiently makes execution computation-bound; in turn, this leads to a speedup by a factor $L$ compared to serial execution (see Fig.~\ref{fig:pipeline_and_accuracy}a). More generally, varying compression rates $\boldsymbol{\eta}$ establishes a trade-off between maximum throughput $R(t)$ and accuracy $A(\boldsymbol{\eta}(t))$. 

It is easy to see that the optimal solution per stage is to exhaust available capacity, binding the target throughput rate at all stages:
\begin{theorem}[Closed-Form Optimal Solution]\label{thm:single-convex}
If Problem~\ref{prob:single-task} is feasible, and the accuracy function $A(\boldsymbol{\eta}(t))$ is  non-decreasing with respect to each $\eta_i(t)$, the problem has the unique optimal solution:
 $ \eta_i^\star(t) = \min \{1,\, \tfrac{c_{i}(t)}{R(t)a_i} \},$ for all $  i \in [L-1].$
\end{theorem}
The proof is in \techreport{App.~\ref{proof:thm:single-convex}.}{App.~A in \cite{techreport}.} Note  that an optimal solution can be computed so long as the $A$ is monotone in $\boldsymbol{\eta}(t)$, which is both natural and also observed experimentally (see Figs.~\ref{fig:pipeline_and_accuracy}b and \ref{fig:pipeline_and_accuracy}c). 

\subsubsection{CSI-Aware -- Multi-Task}
In a multi-task setting, different tasks may exhibit varying sensitivities to activation compression and heterogeneous computational demands. For instance, a task that is highly sensitive to compression at stage $i$ but robust at stage $j$ can receive a larger share of its $i$-th transmission hop's bandwidth and a correspondingly smaller share for its $j$-th hop. Similarly, adaptive allocation of computational resources is essential for optimizing pipeline parallelism. Because the end-to-end throughput of a pipelined task is dictated by its slowest stage, optimizing the compute share at each node prevents stragglers and balances the processing delays across the entire pipeline.

Let $s_{i,k}^{\cmp}(t) \in (0,1]$ denote the fraction of computation capacity of node $v_{i,k}$ allocated to process the $i$-th stage of task $k$ during time slot $t$. 
Similarly, let $s_{i,k}^{\com}(t) \in (0,1]$ denote the fraction of communication bandwidth of edge $e_{i,k}$ allocated to task $k$ for its $i$-th transmission hop during time slot $t$.
%
To ensure that physical resource limits are not exceeded, the total fraction of resources allocated to all active pipeline stages on any given node $v \in \mathcal{V}$, or transmissions across any given edge $e \in \mathcal{E}$, must be bounded at each time slot $t$. We express these capacity constraints by summing the assigned allocations directly over the applicable tasks and their respective stages for each physical resource:
\begin{align}
\textstyle  \sum_{k \in \mathcal{K}(t)} \sum_{i: \, v_{i,k} = v} s_{i,k}^{\cmp}(t) &\le 1, \quad \forall v \in \mathcal{V}, \label{eq:comp_cap}\\
 \textstyle    \sum_{k \in \mathcal{K}(t)} \sum_{i: \, e_{i,k} = e} s_{i,k}^{\com}(t) &\le 1, \quad \forall e \in \mathcal{E}. \label{eq:comm_cap}
\end{align}
We are interested in maximizing the aggregate accuracy across tasks $k$, weighted by some $w_k>0$ (e.g., $1/K$).  We thus arrive at the following optimization problem: 
\begin{problem}[CSI-Aware Multi-Task]\label{prob:multi-task}
\begin{subequations} \label{prob:het-main}
\begin{align}
\mathop{\text{Maximize:}}_{\boldsymbol{\eta}(t),\boldsymbol{s}^\cmp(t),\boldsymbol{s}^\com(t)} & \sum_{k=1}^{K} w_k A_k(\boldsymbol{\eta}_k(t))
 \\[0.5ex]
  \text{subj.~to:} ~~
    & \tfrac{\tau_{i,k}}{s_{i,k}^{\cmp}(t)} \le \tfrac{1}{R_k(t)}, \quad \forall k \in \mathcal{K}(t),  i \in [L_k], \label{cons:het-comp-lat} \\
    & \tfrac{a_{i,k} \eta_{i,k}(t)}{s_{i,k}^{\com}(t) c_{i,k}(t)} \le \tfrac{1}{R_k(t)}, \enspace \forall k \!\in\! \mathcal{K}(t), i \!\in\! [L_k\!-\!1], \label{cons:het-comm-lat}\\
    & \text{Constraints \eqref{eq:comp_cap} and \eqref{eq:comm_cap}}, \nonumber\\
    & \boldsymbol{\eta}_k(t) \in \textstyle\prod_{i=1}^{L_k-1} [\eta_{i,k}^{\min}, 1], \quad \forall k \in \mathcal{K}(t), \label{cons:het-eta-domain} \\
    & s_{i,k}^{\cmp}(t), \, s_{i,k}^{\com}(t) \in (0, 1]. \label{cons:het-s-domain}
\end{align}
\end{subequations}
\end{problem}
Akin to~\eqref{cons:max}, constraints \eqref{cons:het-comp-lat} and \eqref{cons:het-comm-lat} enforce the fundamental bottleneck principle of pipeline parallelism by bounding the effective execution delays. The term ${\tau_{i,k}}/{s_{i,k}^{\cmp}(t)}$ in \eqref{cons:het-comp-lat} represents the \emph{dilated computation time}; allocating only a fraction $s_{i,k}^{\cmp}(t)$ of the assigned node's computational capacity proportionally increases the time required to process that stage. Similarly, in \eqref{cons:het-comm-lat}, the term $s_{i,k}^{\com}(t) c_{i,k}(t)$ represents the communication bandwidth allocated  to task $k$ on its $i$-th transmission hop. Dividing the compressed payload $a_{i,k} \eta_{i,k}(t)$ by this bandwidth slice yields the actual transmission delay. To sustain a target throughput of $R_k(t)$, both this effective computation time and the actual transmission delay must be bounded by  ${1}/{R_k(t)}$, the maximum allowable pipeline cycle time.

To solve Problem \ref{prob:multi-task}, we first establish the conditions under which a feasible resource allocation exists.
\begin{lemma}[Multi-Task Feasibility Condition]\label{lem:multi-feasibility}
Problem \ref{prob:multi-task} is feasible if and only if 
    $\sum_{k \in \mathcal{K}(t)} \sum_{i: \, v_{i,k} = v} \tau_{i,k} R_k(t) \le 1, $ for all $ v \in \mathcal{V}, $ and $
    \sum_{k \in \mathcal{K}(t)} \sum_{i: \, e_{i,k} = e} \frac{a_{i,k} \eta_{i,k}^{\min} R_k(t)}{c_{i,k}(t)} \le 1, $ for all $ e \in \mathcal{E}.$

\end{lemma}
The proof is in \techreport{App.~\ref{proof:lem:multi-feasibility}.}{App.~B in \cite{techreport}.} Next, we show Problem~\ref{prob:multi-task} can be reduced to a convex optimization program when $A_k(\cdot)$ are  concave.
\begin{theorem}[Reduction to Convex Optimization]\label{thm:multi-convex}
Assume Problem~\ref{prob:multi-task} is feasible. If the accuracy functions $A_k(\cdot)$ are concave and non-decreasing, Problem~\ref{prob:multi-task} is equivalent to the following convex optimization problem over the communication fraction variables:
\begin{subequations} \label{prob:comm-main}
\begin{align}
 \mathop{\text{Maximize:}}_{\boldsymbol{s}^\com(t)} & \quad  \textstyle\sum_{k=1}^{K} w_k A_k\left( \boldsymbol{\eta}_{k}^*\left(\mathbf{s}^{\com}\left(t\right)\right)  \right) 
 \\[0.5ex]
  \text{subj.~to:} \quad
    & \textstyle \sum_{k \in \mathcal{K}(t)} \sum_{i: \, e_{i,k} = e} s_{i,k}^{\com}(t) \le 1, \quad \forall e \in \mathcal{E}, \label{cons:comm-sub-cap} \\
    & \textstyle s_{i,k}^{\com}(t) \ge \frac{a_{i,k} \eta_{i,k}^{\min} R_k(t)}{c_{i,k}(t)}, \enspace \forall k \in \mathcal{K}(t),  i \in [L_k\!-\!1], \label{cons:comm-sub-lb}
\end{align}
\end{subequations}
where, for all $k \in \mathcal{K}(t)$, the compression rates $\boldsymbol{\eta}_{k}(t)$ are given by:
\begin{align}
{\eta}_{i,k}^*(\mathbf{s}^{\com}(t) ) \triangleq  \textstyle\min\left(1, \frac{c_{i,k}(t)}{R_k(t) a_{i,k}} s_{i,k}^{\com}(t)\right), ~~i\in[L_k-1],
\end{align}
and the computation fractions are $s_{i,k}^{\cmp}(t) = \tau_{i,k} R_k(t)$,  for  $i\in[L_k].$
\end{theorem}
The proof is in \techreport{App.~\ref{proof:thm:multi-convex}.}{App.~C in \cite{techreport}.} Thm.~\ref{thm:multi-convex} implies that, if the accuracy-compression functions are concave and non-decreasing, we can compute the optimal solution to 
Problem~\ref{prob:multi-task} using standard convex optimization techniques. We verify this concavity premise empirically across models, tasks, and compression schemes in \techreport{App.~\ref{sec:concavity-validation}.}{App.~H in \cite{techreport}.} \CRedit{Moreover,} we show in \techreport{App.~\ref{app:water}}{App.~D in \cite{techreport}} that Prob.~\eqref{prob:comm-main} can be solved in polynomial time via a per-link water-filling
with a QoS floor strategy \cite{elgarhy2018waterfilling}.

\subsection{Estimated-CSI Design}
Next, we consider an online setting where the exact channel capacities $\boldsymbol{c}(t)$ are (a) drawn from a stationary random process (e.g., i.i.d.) and (b) not known at the time of decision-making. Instead, the controller has access to an estimate $\hat{\boldsymbol{c}}(t)$ based, e.g., on the transmission history of previous slots. The system must make decisions based on these estimates, while the actual pipeline delay is governed by the true realization $\boldsymbol{c}(t)$.

For each task $k \in \mathcal{K}(t)$, we define a delay function:
\begin{align}
\! \! D_k(\boldsymbol{\eta}_k(t), \boldsymbol{s}_k(t), \boldsymbol{c}(t))\!\triangleq  \! \max\! \big\{ \!\max_{i \in [L_k]} \!\tfrac{\tau_{i,k}}{s_{i,k}^\cmp(t)},\!\! \max_{i \in [L_k\!-\!1]}\!\! \tfrac{a_{i,k} \eta_{i,k}(t)}{c_{i,k}(t) s_{i,k}^\com(t)} \big\},
\end{align}
that measures the maximum delay across all pipeline computation and communication stages under the given parameters. We denote by $D_{\text{act},k}(t)\triangleq D_k(\boldsymbol{\eta}_k(t), \boldsymbol{s}_k(t), \boldsymbol{c}(t))$ the  delay under the actual bandwidth capacities at time $t$, and   $\hat{D}_k(t)\triangleq D_k(\boldsymbol{\eta}_k(t), \boldsymbol{s}_k(t), \hat{\boldsymbol{c}}(t))$ the delay estimate under capacity estimates $\hat{\boldsymbol{c}}(t)$.
%
We introduce the following additional assumptions on our problem setup:

\noindent\textbf{Assumption 1 (Bounded and Concave Objective):}\label{assum:bounded_objective} The accuracy objective function $A_k(\boldsymbol{\eta}_k)$ for any task $k$ is concave w.r.t.~$\boldsymbol{\eta}_k$ and bounded, i.e., that $A_k(\boldsymbol{\eta}_k) \in [A_k^{\min}, A_k^{\max}]$ for all feasible $\boldsymbol{\eta}_k$. Empirically, this holds in the informative regime $\boldsymbol{\eta}_k\ge\boldsymbol{\eta}_k^{\min}$ for both top-$k$ and LLM.int8 compression (\techreport{App.~\ref{sec:concavity-validation}}{App.~H in \cite{techreport}}).

\noindent\textbf{Assumption 2 (Bounded Second Moment of Constraint Violation):}\label{assum:bounded_variance} For any time slot $t$ and task $k$, the expected squared constraint violation is bounded, i.e., there exists a $\xi^2$ s.t.~for all $ k \in [K]$ and $  t \ge 1$ we have
$ \mathbb{E}[\big(D_{\text{act},k}(t) - \frac{1}{R_k(t)}\big)^2] \le \xi^2. $ 

\noindent\textbf{Assumption 3 (Estimation Error Characteristics):} For any time slot $t$ and task $k \in [K]$, let $\Delta_k(t) = D_{\text{act},k}(t) - \hat{D}_k(t)$ denote the estimation error. 
Let $\mathcal{F}(t)$ denote the system history (filtration) up to the beginning of time slot $t$, encompassing all past channel realizations and algorithmic decisions.
Conditioned on the filtration $\mathcal{F}(t)$,  the expected estimation error is bounded; i.e., there exists a constant $\delta>0$ s.t.:
  $  \mathbb{E}[\Delta_k(t) \mid \mathcal{F}(t)] \le \delta, $ for all $k \in [K]$ and   $t \ge 1.$
Additionally,  the expected positive part of the error is also bounded, i.e., there exists a constant $\delta^+>0$ s.t.
   $ \mathbb{E}[\max\{0, \Delta_k(t)\} \mid \mathcal{F}(t)] \le \delta^+, $ for all $ k \in [K] $ and $t \ge 1.$

\noindent\textbf{Assumption 4 (Robust System Feasibility):}\label{assum:robust_feasibility} There exists a valid resource allocation sequence $\{\boldsymbol{s}_k(t)\}_{t\ge1}$ such that the expected baseline bottleneck delay for any task $k$, evaluated at its maximum compression configuration $\boldsymbol{\eta}_k^{\min}$, satisfies its throughput constraint:
 $   \mathbb{E}[D_k(\boldsymbol{\eta}_k^{\min}, \boldsymbol{s}_k(t), \boldsymbol{c}(t))] \le \tfrac{1}{R_k(t)} - \gamma, $ for all $ k \in [K],$ $ t \ge 1,$
for some positive margin $\gamma > 0$. We assume this feasibility margin dominates the worst-case channel estimation error, i.e.,  $\gamma > \delta^+$.

\subsubsection{CSI-Oblivious -- Single-Task}


Dropping the task index $k$ for the single-task setting, let $\hat{\boldsymbol{c}}(t)$ denote the estimated channel capacity available at the beginning of slot $t$; the resource allocation vector  can be trivially fixed to a vector of ones (i.e., $\boldsymbol{s}(t)=\boldsymbol{1}$), as the task consumes all computational and bandwidth resources. 
We aim to maximize the time average of accuracy subject to the time average delay constraint being satisfied under the true channel realizations:
\begin{problem}[No-CSI Single-Task]\label{prob:online-est}
\begin{subequations}
\begin{align}
 \max_{\{\boldsymbol{\eta}(t)\}_{t\ge 1}} & \quad \textstyle \frac{1}{T} \sum_{t=1}^T \mathbb{E}[A(\boldsymbol{\eta}(t))] 
 \label{prob:online-main}\\[0.5ex]
  \text{subj.~to:} \quad
    &  \textstyle\frac{1}{T} \sum_{t=1}^T \mathbb{E}\left[D_{\text{act}}(t) - \frac{1}{R(t)}\right] \le 0, \\
    &\boldsymbol{\eta}(t) \in  \textstyle\prod_{i=1}^{L-1} [\eta_i^{\min}, 1], \quad \forall t \ge 1.
\end{align}
\end{subequations}
\end{problem}

\begin{algorithm}[!t]
\caption{Single-Task Estimated Stochastic Dual Descent}\label{alg:NoCSI-SingleTask}
\footnotesize
\begin{algorithmic}[1]
\State \textbf{Input:} Tradeoff parameter $\mu > 0$, QoS rates $R(t)$, dual lower bound $\epsilon > 0$.
\State \textbf{Initialization:} Set the scalar Lagrange multiplier $\lambda_1 = \epsilon$.
\For{time slot $t = 1, 2, \dots$}
    \State \textbf{Estimation:} Obtain the channel estimate vector $\hat{\mathbf{c}}(t) = [\hat{c}_1(t), \dots, \hat{c}_{L-1}(t)]$.
    \State \textbf{Primal Step (Compression Optimization):} Select $\boldsymbol{\eta}(t)$ to maximize the partial Lagrangian using the \textbf{estimated} delay cost:
    \Statex \quad $\boldsymbol{\eta}(t) = \arg \max_{\boldsymbol{\eta} \in \prod_{i=1}^{L-1} [\eta_i^{\min}, 1]} \left( A(\boldsymbol{\eta}) - \mu \cdot \lambda_t \cdot \hat{D}(\boldsymbol{\eta}, \boldsymbol{1}, \hat{\mathbf{c}}(t)) \right)$.
    \Statex \quad \textit{Note: This is a convex optimization problem.}
    \State \textbf{Execute:} Run the pipeline with configuration $\boldsymbol{\eta}(t)$. The system experiences the \textbf{true} bottleneck delay $D_{\text{act}}(t) = D(\boldsymbol{\eta}(t), \mathbf{c}(t))$.
    \State \textbf{Dual Step (Multiplier Update):} Update the Lagrange multiplier based on the \textbf{actual} realized constraint violation:
    \Statex \quad $\lambda_{t+1} = \max\{\epsilon,  \lambda_t + D_{\text{act}}(t) - 1/R(t) \}$.
\EndFor
\end{algorithmic}
\end{algorithm}

We solve this problem using a primal-dual algorithm (Alg.~\ref{alg:NoCSI-SingleTask}). In short, compression rates can be set by optimizing a Lagrangian using the estimated delay; under Assumption \hyperref[assum:bounded_objective]{1}, the corresponding problem is convex and thereby tractable. Moreover, under assumptions we have made on the quality of the channel estimator, we prove in \techreport{App.~\ref{proof:thm:performance_bounds}}{App.~E in \cite{techreport}} that the following theorem holds:
\sloppy
\begin{theorem}[No-CSI Single Task Performance Bounds]
\label{thm:performance_bounds}
Under Assumptions \hyperref[assum:bounded_objective]{1}--\hyperref[assum:robust_feasibility]{4}, Alg.~\ref{alg:NoCSI-SingleTask} yields the following performance guarantees:
\textbf{Constraint Satisfaction:} The long-term average expected delay strictly satisfies the target threshold, i.e.:
   $     \lim_{T \to \infty} \frac{1}{T} \sum_{t=1}^T \mathbb{E}[D_\text{act}(t) - \tfrac{1}{R(t)}] \le 0.$
 \textbf{Dual Variable Stability Bound:} The long-term average expected dual variable is bounded by
    $    \lim_{T \to \infty} \frac{1}{T} \sum_{t=1}^T \mathbb{E}[\lambda_t] \le \tfrac{A_{\max} - A_{\min}}{\mu(\gamma - \delta^+)} + \tfrac{\xi^2 + \epsilon^2}{2(\gamma - \delta^+)}.$
    \textbf{Optimality Gap:} Let $B_{\text{dual}}$ be the above dual variable stability bound. Then, 
       $ \lim_{T \to \infty} \frac{1}{T} \sum_{t=1}^T \mathbb{E}[A(\boldsymbol{\eta}(t))] \ge A^* - \tfrac{\mu(\xi^2 + \epsilon^2)}{2} - \mu B_{\text{dual}}(\delta - \delta^*),$
where $A^*$ is the optimal  value over  stationary algorithms that can satisfy the  constraints in Prob.~\ref{prob:online-est}, and
$\delta^*$ is the expected estimation error under the  optimal stationary policy.
\end{theorem}
\fussy

The performance bounds established above reveal a fundamental $[\mathcal{O}(1/\mu), \mathcal{O}(\mu)]$ trade-off governed by the control parameter $\mu$. This parameter serves as the primary mechanism for balancing application utility (inference accuracy) against dual variable stability (delay constraint satisfaction). As $\mu$ increases, the algorithm places a heavier penalizing weight on the delay constraint violations, ensuring stricter compliance; mathematically, the upper bound on the time-averaged dual variables---which  track the cumulative delay violations---shrinks proportionally to $O(1/\mu)$. However, this  focus on delay constraints comes at the  expense of application performance: a larger $\mu$ widens the optimality gap proportionally to $O(\mu)$, leading to lower overall inference accuracy. Conversely, a smaller $\mu$ prioritizes maximizing the accuracy objective, but results in a looser upper bound on the accumulated delay violations. We discuss how to select $\mu$ from problem parameters in \techreport{App.~\ref{sec:mu_selection}.}{App.~G in \cite{techreport}.}
\CRedit{Note that constraint satisfaction  here is time-averaged, so transient violations may occur: this is inherent to acting without CSI.}


\subsubsection{Estimated-CSI -- Multi-Task}
Next, we extend the online setting to  multiple concurrent tasks. Recall that, in this environment, the controller faces a coupled optimization challenge: it must jointly set compression rates while dynamically partitioning shared computational and communication resources, 
while relying entirely on the channel estimates $\hat{\boldsymbol{c}}(t)$.
Let $\mathcal{T}_k(T)\triangleq \{t\leq T: k\in\mathcal{K}(t) \}$ be the set of time slots up to horizon $T$ during which task $k$ is active.
 Our objective is again to maximize the weighted sum of time-averaged accuracies across all tasks, strictly subject to individual time-averaged delay constraints: 
\begin{problem}[No-CSI for Multiple Tasks]\label{prob:online-multitask}
\begin{align}
  \!\!\!\! \max_{\substack{\{\boldsymbol{\eta}(t)\}_{t \ge 1}, \\ \{\mathbf{s}^{\com}(t), \mathbf{s}^{\cmp}(t)\}_{t \ge 1}}} \!\!\!& \frac{1}{T} \sum_{t=1}^T \sum_{k \in \mathcal{K}(t)} w_k \mathbb{E}\Big[A_k(\boldsymbol{\eta}_k(t))\Big] \label{eq:online-multi-obj} \\
  \text{subj.~to:} \quad& 
  \tfrac{1}{|\mathcal{T}_k(T)|} \!\!\!\sum_{t \in \mathcal{T}_k(T)} \!\!\!\!\mathbb{E}\big[D_{\text{act},k}(t) \!-\! \tfrac{1}{R_k(t)}\big] \le 0, \: \forall k \in [K], \label{cons:online-multi-delay} \\
&\text{Constraints \eqref{eq:comp_cap}, \eqref{eq:comm_cap}, \eqref{cons:het-eta-domain} and \eqref{cons:het-s-domain}, }~\forall t \ge 1. \nonumber
\end{align}
\end{problem}


\begin{algorithm}[!t]
\caption{Multi-Task Estimated Stochastic Dual Descent} \label{alg:NoCSI-MultiTask}
\footnotesize
\begin{algorithmic}[1]
\State \textbf{Input:} Tradeoff parameter $\mu > 0$, task weights $w_k$, QoS rates $R_k(t)$, max BCD iterations $J$, dual lower bound $\epsilon > 0$.
\State \textbf{Initialization:} Set the dual variables $\lambda_k(1) = \epsilon$ for all $k \in [K]$.
\For{time slot $t = 1, 2, \dots$}
    \State \textbf{Estimation:} Obtain the channel estimate vector $\hat{\mathbf{c}}(t) = [\hat{c}_e(t)]_{e \in \mathcal{E}}$.
    \State \textbf{Primal Step (Resource and Compression Optimization):} Select $\boldsymbol{\eta}(t)$ and $\mathbf{s}(t)$ for active tasks $\mathcal{K}(t)$ to maximize the partial Lagrangian using the \textbf{estimated} delay cost:
    \Statex \quad $\max \sum_{k \in \mathcal{K}(t)} \left[ w_k A_k(\boldsymbol{\eta}_k) - \mu \cdot \lambda_k(t) \cdot \hat{D}_k(\boldsymbol{\eta}_k, \mathbf{s}_k, \hat{\mathbf{c}}(t)) \right]$
    \Statex \quad $\text{subj.~to:} \quad \sum_{k \in \mathcal{K}(t)} \sum_{i: \, e_{i,k} = e} s_{i,k}^{\com} \le 1, \quad \forall e \in \mathcal{E},$
    \Statex \quad $\phantom{\text{subj.~to:}} \quad \sum_{k \in \mathcal{K}(t)} \sum_{i: \, v_{i,k} = v} s_{i,k}^{\cmp} \le 1, \quad \forall v \in \mathcal{V}.$
    \Statex \quad \textit{Optimization Method (Block Coordinate Descent):} For iteration $j = 1, \dots, J$:
    \Statex \quad \quad \textit{Phase A (Resource Optimization):} Fix $\boldsymbol{\eta} = \boldsymbol{\eta}^{(j-1)}$. Update $\mathbf{s}^{(j)}$ by solving the resulting convex sub-problem.
    \Statex \quad \quad \textit{Phase B (Compression Optimization):} Fix $\mathbf{s} = \mathbf{s}^{(j)}$. Update $\boldsymbol{\eta}^{(j)}$ by solving the resulting concave sub-problem.
    \Statex \quad Assign final operational decisions $\boldsymbol{\eta}(t) = \boldsymbol{\eta}^{(J)}$ and $\mathbf{s}(t) = \mathbf{s}^{(J)}$.
    \State \textbf{Execute:} Run active tasks with assigned configurations. The system experiences the \textbf{true} bottleneck delays $D_{\text{act},k}(t) = D_k(\boldsymbol{\eta}_k(t), \mathbf{s}_k(t), \mathbf{c}(t))$ for all $k \in \mathcal{K}(t)$.
    \State \textbf{Dual Step (Multiplier Update):} Update the Lagrange multipliers based on the \textbf{actual} realized constraint violations, freezing inactive tasks:
    \Statex \quad $\lambda_k(t+1) = \begin{cases} 
    \max\left\{\epsilon, \, \lambda_k(t) + D_{\text{act},k}(t) - \frac{1}{R_k(t)} \right\}, & \text{if } k \in \mathcal{K}(t) \\
    \lambda_k(t), & \text{otherwise}
    \end{cases}$
\EndFor
\end{algorithmic}
\end{algorithm}

Our algorithm for this problem is described in Alg.~\ref{alg:NoCSI-MultiTask}. The joint optimization of resource allocation ($s^{\cmp}, s^{\com}$) and compression factors ($\eta$) is inherently non-convex due to the coupling in the delay term $\eta/s^{com}$. To maintain tractability, we employ a Block Coordinate Descent (BCD) strategy, which iteratively fixes one set of variables to solve the resulting convex sub-problem.
Notice that in the primal optimization step (line 5), the dual variables $\lambda_k$ act as the weights for bottleneck delays. If an active task were to have $\lambda_k = 0$, the optimizer would allocate it close to zero resources, effectively starving it. To prevent this, we introduce $\epsilon > 0$ as a mandatory floor for the dual variables. 
To provide guarantees for this algorithm, we need the following additional assumption regarding the stationary point found by the BCD process:

\vspace{1ex}
\noindent\textbf{Assumption 5 (Primal Optimization Gap):}
\label{assum:NonconvexGap}
The expected additive gap between the stationary point solution found by the block coordinate descent and the true global maximum of the partial Lagrangian in Alg. \ref{alg:NoCSI-MultiTask} is bounded by $C_{\text{gap}}$.


\sloppy
\begin{theorem}[No-CSI Multi Task Performance Bounds]
\label{thm:multitask_performance_bounds}
Under Assumptions \hyperref[assum:bounded_objective]{1}--\hyperref[assum:NonconvexGap]{5}, Alg.~\ref{alg:NoCSI-MultiTask} yields the following performance guarantees:
     \textbf{Constraint Satisfaction:} The long-term average expected delay strictly satisfies the target threshold during the time slots when the task is active, for all tasks $k \in [K]$,
     $   \lim_{|\mathcal{T}_k(T)| \to \infty} \textstyle \frac{1}{|\mathcal{T}_k(T)|} \sum_{t \in \mathcal{T}_k(T)} \mathbb{E}[D_{\text{act},k}(t) - \tfrac{1}{R_k(t)}] \le 0.$
    \textbf{Dual Variable Stability Bound:} The sum of the average expected active dual variables across all tasks is bounded as:
        $\lim_{T \to \infty} \frac{1}{T} \sum_{t=1}^T \sum_{k \in \mathcal{K}(t)} \mathbb{E}[\lambda_k(t)] \le 
        \tfrac{\sum_{k=1}^K \rho_k w_k (A_{k,\max} - A_{k,\min}) + C_{\text{gap}}}{\mu(\gamma - \delta^+)} 
        + \tfrac{\sum_{k=1}^K \rho_k (\xi^2 + \epsilon^2)}{2(\gamma - \delta^+)}.$
\textbf{Optimality Gap:} Let $B_{\text{dual}}$ be the above dual variable stability bound. Then,
$\lim_{T\to\infty} \frac{1}{T} \sum_{t=1}^T \sum_{k \in \mathcal{K}(t)} w_k \mathbb{E}[A_k(\boldsymbol{\eta}_k(t))] \geq A^* - \frac{\mu}{2} \sum_{k=1}^K \rho_k (\xi^2 + \epsilon^2) 
- C_{\text{gap}} - \mu B_{\text{dual}} (\delta - \delta^*),$
where $\rho_k \triangleq \lim_{T \to \infty} \frac{|\mathcal{T}_k(T)|}{T}$. $A^*$ is the optimal objective across stationary algorithms that  satisfy the  constraints in Prob.~\ref{prob:online-multitask}, and $\delta^*$ is the expected estimation error under the optimal stationary policy.
\end{theorem}
\fussy
The proof is in \techreport{App.~\ref{proof:thm:multitask_performance_bounds}.}{App.~F in \cite{techreport}.}  
As in the single-task case (Thm.~ \ref{thm:performance_bounds}), we observe the fundamental $[\mathcal{O}(1/\mu), \mathcal{O}(\mu)]$ trade-off, where the parameter $\mu$ balances aggregate accuracy against delay compliance. However, Theorem \ref{thm:multitask_performance_bounds} highlights the critical role of task-specific weights $w_k$ and activity frequencies $\rho_k$ in a shared-resource environment, normalizing the bounds by the duty cycle and priority of each task. The multi-task bounds explicitly capture the impact of the non-convex BCD algorithm; because the solver converges to a stationary point, the primal optimization gap $C_{\text{gap}}$ emerges in both the dual stability bound and the long-term optimality gap.

\subsection{Accuracy-Compression Function}\label{sec:acccompfun}
The accuracy function $A(\boldsymbol{\eta})$
 can be empirically evaluated on a held-out dataset $\mathcal{D}$
 of $N_{\text{eval}}$
 samples. For each candidate compression vector $\boldsymbol{\eta}$
, the intermediate activations of the distributed pipeline are compressed according to $\boldsymbol{\eta}$, passed to the next layer, and the resulting end-to-end inference quality (e.g., classification accuracy or perplexity) is measured on $\mathcal{D}$
. Because this evaluation involves discrete operations such as top-$k$ sparsification \cite{ming2024adtopk,bian2024does} or quantization~\cite{lin2024awq,shen2024agile}, the resulting function $A(\boldsymbol{\eta})$ may be non-smooth or even discontinuous. We explored two approaches to smooth this empirical evaluation. The first is \emph{fitting}: a smooth function such as a polynomial, a multi-layer perceptron, etc., can be regressed from input/output samples of the accuracy function. 
The fitted function can be readily used to produce gradient estimates $\nabla A(\boldsymbol{\eta})$  while the held-out dataset $\mathcal{D}$ is not needed during optimization. A drawback is that different fitted functions may introduce biases. An alternative approach is to use \emph{kernel-smoothing}, e.g., KNN with an inverse distance weighting kernel, a Gaussian kernel, etc. This is less prone to biases, but is  more computationally intensive, and requires repeated evaluations of  $A(\cdot)$  over the  held-out dataset $\mathcal{D}$. 

Interestingly, in the case of Gaussian kernel smoothing, one can produce a sampling estimator of the gradient $\nabla A(\boldsymbol{\eta})$ directly from randomized $A(\boldsymbol{\eta})$ function calls using Stein's lemma~\cite{stein1981estimation}.
Specifically, for a Gaussian perturbation $\mathbf{z} \sim \mathcal{N}(\mathbf{0}, \mathbf{I})$
 with smoothing bandwidth $\sigma > 0$
, the gradient of the Gaussian-convolved function $\tilde{A}(\boldsymbol{\eta}) = \mathbb{E}_{\mathbf{z}}[A(\boldsymbol{\eta} + \sigma \mathbf{z})]$
 satisfies:
$ \nabla_{\boldsymbol{\eta}} \tilde{A}(\boldsymbol{\eta}) = \frac{1}{\sigma} \mathbb{E}_{\mathbf{z}}\left[ A(\boldsymbol{\eta} + \sigma \mathbf{z}) \cdot \mathbf{z} \right],$ and this can be estimated \emph{directly from sampled function evaluations}, without the need for either integration (w.r.t.~the Gaussian kernel) or differentiation (w.r.t.~$\boldsymbol{\eta}$).
We  have implemented and explored both fitting and this kernel-smoothing via Stein's lemma in our experiments; we provide more details in \techreport{App.~\ref{app:acccomp},}{App.~H in \cite{techreport},} including quality of fit experiments for various ML tasks.

\section{Evaluation}
\label{sec:experiments}


\subsection{Experimental Setup}
\label{sec:expcompsch}

\begin{table}[!t]
\footnotesize
\setlength{\tabcolsep}{3pt}
  \caption{Inference tasks used in the evaluation. Each task is specified by a dataset-model pair and its utility metric is $A_k$.}
  \vspace{-10pt}
\label{tab:tasks}
\centering
\begin{tabular*}{\columnwidth}{@{\extracolsep{\fill}}lllll@{}}
\toprule
\textbf{Task ID} & \textbf{Dataset} & \textbf{Model} & \textbf{Metric ($A_k$)} & $\bm{\eta_{\min}}$ \\
\midrule
\multicolumn{5}{c}{\textit{Vision tasks}} \\
\midrule
\taskVisionOne & MNIST~\cite{mnist}        & MLP                              & Accuracy   & 0.125 \\
\taskVisionTwo & CIFAR-10~\cite{cifar10}   & ResNet-56~\cite{He2015resnet}    & Accuracy   & 0.125 \\
\midrule
\multicolumn{5}{c}{\textit{Language tasks}} \\
\midrule
\taskLangThree & \multirow{2}{*}{ShareGPT~\cite{sharegpt}}          & Gemma-1.1 2B~\cite{gemma1.1}        & Perplexity & 0.25 \\
\taskLangFour  &                                                     & Gemma-1.1 7B~\cite{gemma1.1}        & Perplexity & 0.25 \\[2pt]
\taskLangFive  & MMLU~\cite{hendrycks2021mmlu}                      & Llama-3.1 8B~\cite{llama3herd}  & Accuracy   & 0.25 \\
\taskLangSix   & WikiText~\cite{merity2016pointer}                  & Llama-3.1 8B~\cite{llama3herd}  & Perplexity & 0.25 \\[2pt]
\taskLangSeven & SST2~\cite{sst2}                                   & Flan-T5-Base~\cite{FlanT5chung2024scaling} & Accuracy & 0.125 \\
\bottomrule
\end{tabular*}
\vspace{-5pt}
\end{table}


\sloppy
\noindent\textbf{Inference Tasks and Compression Schemes.}
Vision and language tasks we study are summarized in  Table~\ref{tab:tasks}.  
To compute $A_k(\boldsymbol{\eta}_k)$, we use fitting and Stein kernel smoothing (see \techreport{App.~\ref{app:acccomp}}{App.~H in \cite{techreport}}) on 3 compression schemes: \emph{Top-$k$} sparsification (\cmpT{}), whereby only a fraction $\eta$ of the largest-magnitude activation coordinates are transmitted \cite{ming2024adtopk,bian2024does}; \emph{uniform quantization} (\cmpQ{}), which applies symmetric uniform quantization and maps activations to a reduced-precision fixed-point grid~\cite{lin2024awq,shen2024agile}; and \emph{LLM.int8} (\cmpIEight{}), a hybrid scheme that isolates the largest-magnitude values for higher-precision storage while quantizing the remainder at lower precision~\cite{dettmers2022llmint8}. In each scheme, $\eta\in[0,1]$ is this retained fraction and controls how aggressive the transform is. Additional details for tasks, datasets, and compression schemes are in \techreport{App.~\ref{app:reproducibility}.}{App.~I in \cite{techreport}.}\footnote{\CRedit{Our code is at \url{https://github.com/neu-spiral/communication-aware-inference}.}}

\fussy

\begin{table}[t]  
  \centering
  \caption{Offline experiment scenario summary.
Comma-separated entries are per-task.}
 \vspace{-10pt}
  \label{tab:scenarios}
  {\color{black}
  \footnotesize
  \setlength{\tabcolsep}{3pt}
  \renewcommand{\arraystretch}{0.9}
  \begin{tabular*}{\columnwidth}{@{\extracolsep{\fill}}lccccc@{}}
  \toprule
  \textbf{Scenario} & \textbf{Task ID} & $\bm{c_{\min}}$--$\bm{c_{\max}}$ (MB) & $\bm{R_k}$ (Hz) & $\bm{T}$ & $\bm{L_k}$ \\
  \midrule
  \multicolumn{6}{c}{\textit{Single-task}} \\
  \midrule
  \scenLangSGTwoB       & \taskLangThree  & 0.46--2.08 & 8  & 30  & 4 \\
  \scenLangSGTwoBSevenN & \taskLangThree  & 3.4--16    & 6  & 48  & 7 \\
  \scenLangSGSevenB     & \taskLangFour   & 0.575--2.6 & 7  & 36  & 4 \\
  \scenVisionTRS        & \taskVisionTwo  & 0.06--0.3  & 10 & 50  & 4 \\
  \scenVisionLOO        & \taskVisionTwo  & 0.2--0.6   & 10 & 50  & 4 \\
  \scenLangFiveShot     & \taskLangFive   & 10--40     & 10 & 50  & 3 \\
  \scenLangWTEightB     & \taskLangSix    & 5--9       & 4  & 30  & 4 \\
  \scenVisionJetsonSS   & \taskVisionTwo  & 1.08--1.79 & 1  & 100 & 4 \\
  \scenLangJetsonS      & \taskLangSeven  & 1.35--1.95 & 2  & 100 & 4 \\
  \midrule
  \multicolumn{6}{c}{\textit{Multi-task}} \\
  \midrule
  \scenMTLangSGPTThree            & $3\times$\,\taskLangThree & 4.2--7     & 6, 5, 4  & 20  & 4 \\
  \scenMTVisFiveLooseNTenTHundred & $2\times$\,\taskVisionTwo & 0.2--0.6   & 10, 10   & 100 & 4 \\
  \scenMTLVJM                     & \makecell[c]{\taskVisionTwo,\\ \taskLangSeven} & 1.08--1.79 & 0.5, 1.5 & 100 & 4 \\
  \bottomrule
  \end{tabular*}
  }
  \vspace{-8pt}
  \end{table}

\noindent\textbf{Offline Experiments.}
We conduct 12 distinct offline experiments on linear topologies with the tasks from Table~\ref{tab:scenarios}, with corresponding experiment parameters. Simulations run for $T$ timeslots, with the link capacities sampled u.a.r.~from $[c_{\min}, c_{\max}]$. When $\tau$ varies across schemes, it reflects per-scheme profiles characterized on \textit{Jetson} testbed. Additional details are provided in \techreport{App.~\ref{app:reproducibility}.}{App.~I in \cite{techreport}.}  

\noindent\textbf{Testbed Experiments.}
We also conduct 14 experiments on 3 testbeds: (i) a CPU-only Raspberry Pi edge deployment for small LLMs, (ii) a Jetson Edge GPU deployment for LLM and vision workloads, and (iii) a multi-GPU testbed for large-model and multi-task deployments with controllable communication bottlenecks. Experiment settings are summarized in Table~\ref{tab:online-scenarios}; further testbed deployment details are in \techreport{App.~\ref{app:testbed_details}.}{App.~J in \cite{techreport}.} Our focus is on No-CSI methods on testbed experiments, as the channel is volatile and difficult to estimate.

\begin{table}[t]
\centering
\caption{Testbed experiment scenario summary.
Comma-separated entries are per-task, and Jetson $R_k$ is given as \cmpT/\cmpQ/\cmpIEight.}
\vspace{-10pt}
\label{tab:online-scenarios}
{\color{black}
\footnotesize
\setlength{\tabcolsep}{3pt}
\renewcommand{\arraystretch}{0.9}
\begin{tabular*}{\columnwidth}{@{\extracolsep{\fill}}lccccc@{}}
\toprule
\textbf{Scenario} & \textbf{Task ID} & \textbf{Network} & $\bm{R_k}$ (Hz) & $\bm{T}$ & $\bm{L_k}$ \\
\midrule
\multicolumn{6}{c}{\textit{PRESCIENT}} \\
\midrule
\scenPrsntVizWire      & \taskVisionTwo & \netProfileLan  & 5.4 & 50 & 3 \\
\scenPrsntVizEdge      & \taskVisionTwo & \netProfileEdge & 4.3 & 50 & 3 \\
\scenPrsntLangWikiWire & \taskLangSix   & \netProfileLan  & 4.0 & 50 & 3 \\
\scenPrsntLangWikiEdge & \taskLangSix   & \netProfileEdge & 3.0 & 50 & 3 \\
\scenPrsntLangMmluEdge & \taskLangFive  & \netProfileEdge & 2.5 & 50 & 3 \\
\scenPrsntVizVizWire   & $2\times$\,\taskVisionTwo & \netProfileLan  & 2.7, 2.7 & 50 & 3 \\
\scenPrsntVizVizEdge   & $2\times$\,\taskVisionTwo & \netProfileEdge & 1.2, 1.2 & 50 & 3 \\
\scenPrsntVizLangEdge  & \makecell[c]{\taskVisionTwo,\\ \taskLangSix} & \netProfileEdge & 1.1, 1.1 & 50 & 3 \\
\scenPrsntVizLangWlan  & \makecell[c]{\taskVisionTwo,\\ \taskLangSix} & \netProfileWlan & 0.7, 0.5 & 50 & 3 \\
\scenPrsntLangLangEdge & $2\times$\,\taskLangSix   & \netProfileEdge & 0.9, 0.9 & 50 & 3 \\
\midrule
\multicolumn{6}{c}{\textit{Jetson}} \\
\midrule
\scenVisionJetsonTestbedS & \taskVisionTwo & Wi-Fi AP & 0.530/0.527/0.525 & 50 & 4 \\
\scenLangJetsonTestbedS   & \taskLangSeven & Wi-Fi AP & 5.17/5.51/3.97    & 50 & 4 \\
\scenMTJetsonTestbedM     & \makecell[c]{\taskVisionTwo,\\ \taskLangSeven} & Wi-Fi AP
  & \makecell[c]{0.274/0.273/0.273,\\ 0.820/0.821/0.817} & 20 & 4 \\
\midrule
\multicolumn{6}{c}{\textit{Raspberry Pi}} \\
\midrule
\scenTbSSGPT & \taskLangThree & Wi-Fi 802.11ac & 5.0 & 100 & 7 \\
\bottomrule
\end{tabular*}
}
\vspace{-8pt}
\end{table}

\newcommand{\OfflineArrayStretch}{0.45}
\newcommand{\OfflineDataSkip}{-1.0pt}      
\newcommand{\OfflineTabColSepST}{1.5pt}  
\newcommand{\OfflineTabColSepMT}{1.0pt}  
\newcommand{\offrow}{\\[\OfflineDataSkip]}

\begin{table*}[t]
  \centering
  \caption{Offline results for single (upper) and multi-task (lower). We report agg. utility $\metU$, avg. delay $\metD$, and excess delay $\metED$, across methods, grouped into reference, CSI-aware, and No-CSI methods. $\cmpT$, $\cmpQ$, and $\cmpIEight$ denote Top-$k$, quantization, and LLM.int8, respectively. Infeasible entries ($\metED/\metD > 0.05$) are shaded light red. Bold and underlined values denote the best and second-best accuracies, respectively, among feasible entries.}
  \vspace{-15pt}
   \techreport{%
  \vspace{11pt}
}{\vspace{2pt}}
  \label{tab:results-offline-all}
  
  \medskip
  \hrule
  \smallskip
  {\bfseries Single-task Scenarios}
  \smallskip
  \hrule
  \medskip
  
  {\footnotesize
  \setlength{\tabcolsep}{\OfflineTabColSepST}
  \renewcommand{\arraystretch}{\OfflineArrayStretch}
  \begin{tabular*}{\textwidth}{@{\extracolsep{\fill}}ll c ccc !{\vrule width 0.8pt} ccc ccc !{\vrule width 0.8pt} ccc ccc ccc ccc@{}}
  \toprule
  \multirow{3}{*}{\textbf{Scenario}} & \multirow{3}{*}{\textbf{Metric}}
   & \multicolumn{4}{c!{\vrule width 0.8pt}}{\textbf{Reference}}
   & \multicolumn{6}{c!{\vrule width 0.8pt}}{\textbf{CSI-aware}}
   & \multicolumn{12}{c}{\textbf{No-CSI}} \\
  \cmidrule(l{0.3pt}r{0.3pt}){3-6}\cmidrule(l{0.3pt}r{0.3pt}){7-12}\cmidrule(l{0.3pt}r{0.3pt}){13-24}
   & & \multirow{2}{*}{\algNONE[short]} & \multicolumn{3}{c!{\vrule width 0.8pt}}{\algMAX[short]}
     & \multicolumn{3}{c}{\algUNI[short]} & \multicolumn{3}{c!{\vrule width 0.8pt}}{\algOURSCSI[short]}
     & \multicolumn{3}{c}{\algMYO[short]} & \multicolumn{3}{c}{\algCONS[short]} & \multicolumn{3}{c}{\algMA[short]} & \multicolumn{3}{c}{\algOURSNOCSIst[short]} \\
  \cmidrule(l{0.3pt}r{0.3pt}){4-6}\cmidrule(l{0.3pt}r{0.3pt}){7-9}\cmidrule(l{0.3pt}r{0.3pt}){10-12}\cmidrule(l{0.3pt}r{0.3pt}){13-15}\cmidrule(l{0.3pt}r{0.3pt}){16-18}\cmidrule(l{0.3pt}r{0.3pt}){19-21}\cmidrule(l{0.3pt}r{0.3pt}){22-24}
   & & & \cmpT & \cmpQ & \cmpIEight & \cmpT & \cmpQ & \cmpIEight & \cmpT & \cmpQ & \cmpIEight & \cmpT & \cmpQ & \cmpIEight & \cmpT & \cmpQ & \cmpIEight & \cmpT & \cmpQ & \cmpIEight & \cmpT & \cmpQ & \cmpIEight \\
  \midrule
  \multirow{3}{*}{\scenLangSGTwoB}
   & $\metU$   & \nfeas{1.00} & 0.802 & 0.365 & 0.889 & 0.968 & 0.989 & 0.994 & 0.985 & \underline{0.995} & \best{0.997} & \nfeas{0.982} & \nfeas{0.992} & \nfeas{0.996} & 0.925 & 0.961 & \underline{0.984} & \nfeas{0.992} & \nfeas{0.996} & \nfeas{0.998} & 0.957 & \best{0.991} & 0.980 \offrow
   & $\metD$   & \nfeas{222}  & 55.5  & 55.5  & 55.5  & 124   & 124   & 124   & 124   & \underline{124} & \best{124} & \nfeas{183} & \nfeas{183} & \nfeas{183} & 96.8  & 96.8  & \underline{96.8} & \nfeas{203} & \nfeas{203} & \nfeas{203} & 121   & \best{119} & 99.1 \offrow
   & $\metED$  & \nfeas{97.1} & -69.5 & -69.5 & -69.5 & -1.46 & -1.46 & -1.46 & -1.46 & \underline{-1.46} & \best{-1.46} & \nfeas{58.0} & \nfeas{58.0} & \nfeas{58.0} & -28.2 & -28.2 & \underline{-28.2} & \nfeas{77.8} & \nfeas{77.8} & \nfeas{77.8} & -4.15 & \best{-5.67} & -25.9 \offrow
  \midrule
  
  \multirow{3}{*}{\scenLangSGTwoBSevenN}
   & $\metU$   & \nfeas{1.00} & 0.444 & 0.155 & 0.344 & 0.824 & 0.805 & 0.949 & \underline{0.951} & 0.909 & \best{0.985} & \nfeas{0.944} & \nfeas{0.891} & \nfeas{0.981} & 0.750 & 0.378 & \underline{0.876} & \nfeas{0.979} & \nfeas{0.980} & \nfeas{0.994} & 0.825 & \nfeas{0.797} & \best{0.949} \offrow
   & $\metD$   & \nfeas{440}  & 96.9  & 96.9  & 96.9  & 167   & 167   & 167   & \underline{167} & 167   & \best{167} & \nfeas{344} & \nfeas{344} & \nfeas{344} & 151   & 151   & \underline{151} & \nfeas{337} & \nfeas{337} & \nfeas{337} & 175 & \nfeas{182} & \best{173} \offrow
   & $\metED$  & \nfeas{274}  & -69.8 & -69.8 & -69.8 & 0.00  & 0.00  & 0.00  & \underline{0.00} & 0.00  & \best{0.00} & \nfeas{178} & \nfeas{178} & \nfeas{178} & -15.7 & -15.7 & \underline{-15.7} & \nfeas{171} & \nfeas{171} & \nfeas{171} & 7.86 & \nfeas{15.3} & \best{5.94} \offrow
  \midrule
  
  \multirow{3}{*}{\scenLangSGSevenB}
   & $\metU$   & \nfeas{1.00} & 0.537 & 0.882 & 0.851 & 0.930 & \underline{0.994} & 0.992 & 0.966 & \best{0.996} & \underline{0.994} & \nfeas{0.963} & \nfeas{0.995} & \nfeas{0.994} & 0.869 & \underline{0.989} & 0.986 & \nfeas{0.988} & \nfeas{0.996} & \nfeas{0.997} & 0.968 & \best{0.989} & 0.986 \offrow
   & $\metD$   & \nfeas{269}  & 67.3  & 67.3  & 67.3  & 141   & \underline{141} & 141   & 141   & \best{141} & \underline{141} & \nfeas{230} & \nfeas{230} & \nfeas{230} & 115   & \underline{115} & 115   & \nfeas{234} & \nfeas{234} & \nfeas{234} & 146 & \best{113} & 113 \offrow
   & $\metED$  & \nfeas{126}  & -75.5 & -75.5 & -75.5 & -1.37 & \underline{-1.37} & -1.37 & -1.37 & \best{-1.37} & \underline{-1.37} & \nfeas{87.2} & \nfeas{87.2} & \nfeas{87.2} & -27.7 & \underline{-27.7} & -27.7 & \nfeas{91.3} & \nfeas{91.3} & \nfeas{91.3} & 3.53 & \best{-29.5} & -30.3 \offrow
  \midrule
  
  \multirow{3}{*}{\scenVisionTRS}
   & $\metU$   & \nfeas{0.933} & 0.415 & 0.769 & 0.873 & 0.677 & 0.827 & 0.893 & 0.708 & \underline{0.926} & \best{0.933} & \nfeas{0.703} & \nfeas{0.921} & \nfeas{0.932} & 0.786 & 0.899 & \underline{0.903} & \nfeas{0.757} & \nfeas{0.912} & \nfeas{0.913} & 0.889 & \best{0.923} & \nfeas{0.933} \offrow
   & $\metD$   & \nfeas{476}   & 59.5  & 59.5  & 59.5  & 101 & 101 & 101 & 101 & \underline{101} & \best{101} & \nfeas{175} & \nfeas{175} & \nfeas{175} & 75.5 & 75.5 & \underline{75.5} & \nfeas{164} & \nfeas{164} & \nfeas{164} & 85.3 & \best{104} & \nfeas{119} \offrow
   & $\metED$  & \nfeas{376}   & -40.5 & -40.5 & -40.5 & 1.23 & 1.23 & 1.23 & 1.23 & \underline{1.23} & \best{1.23} & \nfeas{75.4} & \nfeas{75.4} & \nfeas{75.4} & -24.5 & -24.5 & \underline{-24.5} & \nfeas{63.7} & \nfeas{63.7} & \nfeas{63.7} & -14.7 & \best{3.99} & \nfeas{19.0} \offrow
  \midrule
  
  \multirow{3}{*}{\scenVisionLOO}
   & $\metU$   & \nfeas{0.933} & 0.415 & 0.769 & 0.873 & 0.894 & 0.913 & \underline{0.923} & 0.900 & \best{0.933} & \best{0.933} & \nfeas{0.898} & \nfeas{0.923} & \nfeas{0.933} & 0.831 & 0.903 & \underline{0.913} & \nfeas{0.910} & \nfeas{0.930} & \nfeas{0.933} & \nfeas{0.901} & 0.913 & \best{0.933} \offrow
   & $\metD$   & \nfeas{189}   & 23.6  & 23.6  & 23.6  & 100 & 100 & \underline{100} & 100 & \best{100} & \best{100} & \nfeas{122} & \nfeas{122} & \nfeas{122} & 72.5 & 72.5 & \underline{72.5} & \nfeas{118} & \nfeas{118} & \nfeas{118} & \nfeas{106} & 89.0 & \best{98.5} \offrow
   & $\metED$  & \nfeas{88.7}  & -76.4 & -76.4 & -76.4 & 0.00 & 0.00 & \underline{0.00} & 0.00 & \best{0.00} & \best{0.00} & \nfeas{21.5} & \nfeas{21.5} & \nfeas{21.5} & -27.5 & -27.5 & \underline{-27.5} & \nfeas{18.3} & \nfeas{18.3} & \nfeas{18.3} & \nfeas{6.06} & -11.0 & \best{-1.50} \offrow
  \midrule
  
  \multirow{3}{*}{\scenLangFiveShot}
   & $\metU$   & \nfeas{0.595} & 0.420 & 0.255 & 0.543 & 0.531 & 0.479 & \underline{0.583} & 0.545 & 0.499 & \best{0.590} & \nfeas{0.575} & \nfeas{0.559} & \nfeas{0.592} & 0.467 & 0.355 & \underline{0.543} & \nfeas{0.493} & \nfeas{0.417} & \nfeas{0.555} & \nfeas{0.487} & \nfeas{0.425} & \best{0.555} \offrow
   & $\metD$   & \nfeas{233}   & 23.9  & 23.9  & 23.9  & 100 & 100 & \underline{100} & 100 & 100 & \best{100} & \nfeas{155} & \nfeas{155} & \nfeas{155} & 68.7 & 68.7 & \underline{68.7} & \nfeas{145} & \nfeas{145} & \nfeas{145} & \nfeas{106} & \nfeas{115} & \best{91.1} \offrow
   & $\metED$  & \nfeas{133}   & -76.1 & -76.1 & -76.1 & 0.00 & 0.00 & \underline{0.00} & 0.00 & 0.00 & \best{0.00} & \nfeas{55.0} & \nfeas{55.0} & \nfeas{55.0} & -31.3 & -31.3 & \underline{-31.3} & \nfeas{44.9} & \nfeas{44.9} & \nfeas{44.9} & \nfeas{5.74} & \nfeas{15.1} & \best{-8.89} \offrow
  \midrule
  
  \multirow{3}{*}{\scenLangWTEightB}
   & $\metU$   & \nfeas{1.00} & 0.727 & 0.206 & 0.818 & 0.866 & 0.894 & \underline{0.972} & 0.909 & 0.922 & \best{0.988} & \nfeas{0.904} & \nfeas{0.917} & \nfeas{0.984} & 0.833 & 0.843 & \underline{0.943} & \nfeas{0.910} & \nfeas{0.942} & \nfeas{0.990} & 0.856 & 0.904 & \best{0.962} \offrow
   & $\metD$   & \nfeas{709}  & 177   & 177   & 177   & 250 & 250 & \underline{250} & 250 & 250 & \best{250} & \nfeas{316} & \nfeas{316} & \nfeas{316} & 230 & 230 & \underline{230} & \nfeas{295} & \nfeas{295} & \nfeas{295} & 252 & 261 & \best{242} \offrow
   & $\metED$  & \nfeas{459}  & -72.6 & -72.6 & -72.6 & 0.00 & 0.00 & \underline{0.00} & 0.00 & 0.00 & \best{0.00} & \nfeas{65.5} & \nfeas{65.5} & \nfeas{65.5} & -20.2 & -20.2 & \underline{-20.2} & \nfeas{44.8} & \nfeas{44.8} & \nfeas{44.8} & 1.79 & 10.8 & \best{-7.82} \offrow
  \midrule

  
  \multirow{3}{*}{\shortstack[l]{\scenVisionJetsonS}}
   & $\metU$   & \nfeas{0.933} & 0.414 & 0.869 & 0.931 & 0.749 & \underline{0.932} & \underline{0.932} & 0.781 & \best{0.933} & \best{0.933} & \nfeas{0.780} & \underline{0.932} & \nfeas{0.932} & 0.738 & 0.917 & \best{0.932} & \nfeas{0.780} & 0.932 & \nfeas{0.932} & 0.786 & 0.925 & 0.931 \offrow
   & $\metD$   & \nfeas{3810}  & 476  & 473  & 477  & 1000 & \underline{1000} & \underline{1000} & 1000 & \best{1000} & \best{1000} & \nfeas{1091} & \underline{1034} & \nfeas{1091} & 902  & 885  & \best{902} & \nfeas{1084} & 1034 & \nfeas{1084} & 1034 & 1017 & 477 \offrow
   & $\metED$  & \nfeas{2810}  & -524 & -527 & -523 & 0.00 & \underline{0.00} & \underline{0.00} & 0.00 & \best{0.00} & \best{0.00} & \nfeas{91.4} & \underline{33.7} & \nfeas{91.4} & -97.9 & -115 & \best{-97.9} & \nfeas{83.5} & 34.3 & \nfeas{83.5} & 34.3 & 16.5 & -523 \offrow
  \midrule

  
  \multirow{3}{*}{\shortstack[l]{\scenLangJetsonS}}
   & $\metU$   & \nfeas{0.930} & 0.850 & 0.506 & 0.930 & \best{0.930} & \underline{0.480} & \best{0.930} & \best{0.930} & \underline{0.480} & \best{0.930} & 0.930 & 0.480 & 0.930 & 0.930 & 0.480 & 0.930 & 0.930 & 0.480 & 0.930 & \underline{0.930} & 0.501 & \best{0.930} \offrow
   & $\metD$   & \nfeas{683}   & 105  & 104  & 105  & \best{500} & \underline{500} & \best{500} & \best{500} & \underline{500} & \best{500} & 500 & 500 & 500 & 434 & 436 & 434 & 500 & 498 & 500 & \underline{304} & 104 & \best{105} \offrow
   & $\metED$  & \nfeas{183}   & -395 & -396 & -395 & \best{0.00} & \underline{0.00} & \best{0.00} & \best{0.00} & \underline{0.00} & \best{0.00} & 0.310 & -0.004 & 0.0003 & -65.7 & -63.6 & -65.7 & -0.150 & -1.52 & -0.0002 & \underline{-196} & -396 & \best{-395} \offrow
  \bottomrule
  \end{tabular*}}
  
  \medskip
  \hrule
  \smallskip
  {\bfseries Multi-task Scenarios}
  \smallskip
  \hrule
  \medskip
  
  {\footnotesize
  \setlength{\tabcolsep}{\OfflineTabColSepMT}
  \renewcommand{\arraystretch}{\OfflineArrayStretch}
  \begin{tabular*}{\textwidth}{@{\extracolsep{\fill}}ll c ccc !{\vrule width 0.8pt} ccc ccc ccc ccc !{\vrule width 0.8pt} ccc ccc ccc@{}}
  \toprule
  \multirow{3}{*}{\textbf{Scenario}} & \multirow{3}{*}{\textbf{Metric}}
   & \multicolumn{4}{c!{\vrule width 0.8pt}}{\textbf{Reference}}
   & \multicolumn{12}{c!{\vrule width 0.8pt}}{\textbf{CSI-aware}}
   & \multicolumn{9}{c}{\textbf{No-CSI}} \\
  \cmidrule(l{0.3pt}r{0.3pt}){3-6}\cmidrule(l{0.3pt}r{0.3pt}){7-18}\cmidrule(l{0.3pt}r{0.3pt}){19-27}
   & & \multirow{2}{*}{\algNONE[short]} & \multicolumn{3}{c!{\vrule width 0.8pt}}{\algMAX[short]}
     & \multicolumn{3}{c}{\algEQ[short]} & \multicolumn{3}{c}{\algPROP[short]} & \multicolumn{3}{c}{\algPRIO[short]} & \multicolumn{3}{c!{\vrule width 0.8pt}}{\algOURSCSI[short]}
     & \multicolumn{3}{c}{\algDES[short]} & \multicolumn{3}{c}{\algDLPROP[short]} & \multicolumn{3}{c}{\algOURSNOCSImt[short]} \\
  \cmidrule(l{0.3pt}r{0.3pt}){4-6}\cmidrule(l{0.3pt}r{0.3pt}){7-9}\cmidrule(l{0.3pt}r{0.3pt}){10-12}\cmidrule(l{0.3pt}r{0.3pt}){13-15}\cmidrule(l{0.3pt}r{0.3pt}){16-18}\cmidrule(l{0.3pt}r{0.3pt}){19-21}\cmidrule(l{0.3pt}r{0.3pt}){22-24}\cmidrule(l{0.3pt}r{0.3pt}){25-27}
   & & & \cmpT & \cmpQ & \cmpIEight & \cmpT & \cmpQ & \cmpIEight & \cmpT & \cmpQ & \cmpIEight & \cmpT & \cmpQ & \cmpIEight & \cmpT & \cmpQ & \cmpIEight & \cmpT & \cmpQ & \cmpIEight & \cmpT & \cmpQ & \cmpIEight & \cmpT & \cmpQ & \cmpIEight \\
  \midrule
  \multirow{3}{*}{\scenMTLangSGPTThree}
   & $\metU$  & \nfeas{1.00} & 0.888 & $0.0001$ & 0.364 & \best{0.967} & 0.581 & 0.379 & \underline{0.966} & 0.626 & 0.375 & 0.951 & 0.491 & 0.543 & 0.964 & $0.0001$ & 0.364 & \best{0.935} & $0.0001$ & 0.364 & \underline{0.934} & $0.0001$ & 0.364 & \nfeas{0.955} & $0.0001$ & 0.364 \offrow
   & $\metD$  & \nfeas{352} & 87.9 & 87.9 & 87.9 & \best{131} & 131 & 131 & \underline{131} & 131 & 131 & 131 & 131 & 131 & 131 & 92.1 & 92.1 & \best{104} & 87.9 & 87.9 & \underline{103} & 87.9 & 87.9 & \nfeas{133} & 87.8 & 87.8 \offrow
   & $\metED$ & \nfeas{663} & -128 & -128 & -128 & \best{0.880} & 0.880 & 0.880 & \underline{1.38} & 1.38 & 1.38 & 0.00 & 0.00 & 0.00 & 0.00 & -115 & -115 & \best{-80.3} & -128 & -128 & \underline{-83.1} & -128 & -128 & \nfeas{7.50} & -128 & -128 \offrow
  \midrule
  
  \multirow{3}{*}{\scenMTVisFiveLooseNTenTHundred}
   & $\metU$  & \nfeas{0.933} & 0.415 & 0.869 & 0.933 & 0.738 & 0.931 & 0.933 & 0.781 & \underline{0.933} & \underline{0.933} & 0.704 & 0.925 & \underline{0.933} & 0.767 & 0.932 & \best{0.933} & \nfeas{0.707} & 0.931 & 0.933 & 0.711 & 0.932 & \underline{0.933} & \nfeas{0.737} & 0.932 & \best{0.933} \offrow
   & $\metD$  & \nfeas{416} & 52.0 & 52.0 & 52.0 & 101 & 101 & 101 & 100 & \underline{100} & \underline{100} & 100 & 100 & \underline{100} & 99.8 & 57.9 & \best{46.2} & \nfeas{106} & 66.1 & 52.0 & 105 & 67.0 & \underline{51.9} & \nfeas{114} & 71.7 & \best{43.9} \offrow
   & $\metED$ & \nfeas{316} & -48.0 & -48.0 & -48.0 & 0.620 & 0.620 & 0.620 & 0.00 & \underline{0.00} & \underline{0.00} & 0.00 & 0.00 & \underline{0.00} & -0.240 & -42.1 & \best{-53.8} & \nfeas{5.58} & -34.0 & -48.0 & 4.99 & -33.0 & \underline{-48.1} & \nfeas{13.8} & -28.4 & \best{-56.1} \offrow
  \midrule

  \multirow{3}{*}{\shortstack[l]{\scenMTLVJM}}
   & $\metU$  & \nfeas{0.931} & 0.632 & 0.684 & 0.930 & 0.859 & 0.706 & 0.931 & 0.880 & 0.721 & 0.931 & 0.877 & 0.724 & \underline{0.931} & 0.892 & 0.714 & \best{0.931} & \nfeas{0.919} & 0.718 & \underline{0.930} & \nfeas{0.919} & 0.719 & \underline{0.930} & 0.883 & 0.690 & \best{0.930} \offrow
   & $\metD$  & \nfeas{4498} & 582 & 577 & 581 & 1335 & 1335 & 1335 & 1334 & 1334 & 1335 & 1333 & 1333 & \underline{1333} & 1332 & 1331 & \best{797} & \nfeas{2427} & 1281 & \underline{584} & \nfeas{2427} & 1281 & \underline{584} & 1345 & 738 & \best{488} \offrow
   & $\metED$ & \nfeas{3164} & -752 & -757 & -752 & 1.46 & 1.45 & 1.45 & 0.665 & 0.608 & 1.41 & 0.00 & 0.00 & \underline{0.00} & -1.61 & -1.99 & \best{-536} & \nfeas{1094} & -52.7 & \underline{-749} & \nfeas{1094} & -52.7 & \underline{-749} & 11.2 & -645 & \best{-845} \offrow
  \bottomrule
  \end{tabular*}}
  \vspace{-4pt}
     \techreport{%
}{\vspace{7pt}}
  \end{table*}

\noindent\textbf{Compression Rate Selection Algorithms.}
For single-task experiments, we evaluate the following policies: 
\algMAX{} is a baseline forcing $\eta_{i}=\eta_{i}^{min}$, a strict lower bound on accuracy. 
\algNONE{} is a baseline transmitting uncompressed data ($\eta_{i}=1$) to establish an upper bound on accuracy, ignoring constraints. 
\algUNI{} is a CSI-aware policy enforcing  uniform compression across stages, bottlenecking the system to the weakest link. 
\algMYO{} is a No-CSI policy assuming static networks by treating the most recent channel observation as the current capacity. 
\algCONS{} is a No-CSI pessimistic baseline  assuming future capacities will match the lowest historically observed value. 
\algMA{} is a No-CSI baseline assuming a moving average of past capacities to smooth transient network noise. 
Finally, we evaluate our proposed methods: \algOURSCSI{} is our CSI-aware optimizer computing the closed-form optimal solution of Theorem \ref{thm:single-convex}. \algOURSNOCSIst is our No-CSI optimizer using Alg.~\ref{alg:NoCSI-SingleTask} with a 90\% Lower Confidence Bound (LCB) as channel estimate~$\hat{\mathbf{c}}$.

In multi-task environments, beyond 
 \algMAX{} and \algNONE{}, we implement policies that select both compression and resource allocations.
\algEQ{} is a CSI-aware baseline splitting resources equally, then applies optimal single-task compression. 
\algPROP{} is a CSI-aware baseline allocating resources proportionally based on inherent task computational loads and data sizes, before applying optimal single-task compression. 
\algPRIO{} is a CSI-aware baseline meeting minimum requirements for all tasks, then greedily funneling leftover resources to the highest-priority tasks. 
\algDES{} is a No-CSI baseline splitting resources equally and running Alg. \ref{alg:NoCSI-SingleTask} for each task in isolation.
\algDLPROP{} is a No-CSI baseline dynamically shifting resources proportional to dual variables ($\lambda_{k}(t)$) to explicitly aid struggling tasks. 
\algMA{} is a No-CSI baseline feeding smoothed, moving-average channel estimates into the optimal CSI-aware multi-task solver. 
Lastly, we evaluate our proposed multi-task frameworks: \algOURSCSI{} is our CSI-aware optimizer solving the convex optimization in Theorem \ref{thm:multi-convex}. \algOURSNOCSImt{} is our No-CSI optimizer using Algorithm \ref{alg:NoCSI-MultiTask} with a 90\% LCB as channel estimate~$\hat{\mathbf{c}}$.

\noindent\textbf{Evaluation Metrics.}
\sloppy
We evaluate system performance over the experiment horizon using three time-averaged metrics. \emph{Aggregate Utility ($\metU$)} measures the overall normalized weighted accuracy across active tasks: $\metU = \frac{1}{T} \sum_{t=1}^T \frac{\sum_{k \in \mathcal{K}(t)} w_k A_k(\boldsymbol{\eta}_k(t))}{\sum_{k \in \mathcal{K}(t)} w_k}$. \emph{Average Delay ($\metD$)} represents the average bottleneck delay: $\metD = \frac{1}{KT} \sum_{t=1}^T \sum_{k \in \mathcal{K}(t)} D_{act,k}(t)$. Finally, \emph{Excess Delay ($\metED$)} evaluates QoS compliance by tracking the difference between actual bottleneck delay and the target threshold: $\metED = \frac{1}{KT} \sum_{t=1}^T \sum_{k \in \mathcal{K}(t)} \big( D_{act,k}(t) - 1/R_k(t) \big)$. 
We consider an excess delay higher than 5\% of the average target delay as violating feasibility.

\newcommand{\TestbedHalfWidth}{0.495\textwidth} 
\newcommand{\TestbedHGap}{0.01\textwidth}       
\newcommand{\TestbedTabColSepL}{.5pt}          
\newcommand{\TestbedTabColSepR}{.5pt}          
\newcommand{\TestbedDataSkipL}{-1.5pt}             
\newcommand{\TestbedDataSkipR}{1.44pt}           
\newcommand{\tbrowL}{\\[\TestbedDataSkipL]}
\newcommand{\tbrowR}{\\[\TestbedDataSkipR]}

\begin{table*}[t]
  \centering
  \caption{Testbed results for single-task (left) and multi-task (right), excluding the CSI-aware methods. We report aggregate utility $\metU$, average delay $\metD$, and excess delay $\metED$, across methods, grouped into reference baselines and No-CSI methods. Compression schemes $\cmpT$, $\cmpQ$, and $\cmpIEight$ denote Top-$k$, quantization, and LLM.int8, respectively. Infeasible entries ($\metED/\metD > 0.05$) are shaded light red. Bold and underlined values denote the best and second-best results, respectively, among feasible entries within the No-CSI group.}
  \vspace{-15pt}
     \techreport{%
  \vspace{12pt}
}{\vspace{2pt}}
  \label{tab:results-testbed-all-nocsi}
  \footnotesize
  \begin{minipage}[t]{\TestbedHalfWidth}
  \centering
  \setlength{\tabcolsep}{\TestbedTabColSepL}
  \renewcommand{\arraystretch}{1}
  \medskip
  \hrule
  \smallskip
  {\bfseries Single-task Scenarios}
  \smallskip
  \hrule
  \medskip
  \resizebox{\linewidth}{!}{%
  \begin{tabular}{@{}ll c ccc !{\vrule width 0.8pt} ccc ccc ccc ccc@{}}
  \toprule
  \multirow{3}{*}{\textbf{Scenario}} & \multirow{3}{*}{\textbf{Metric}}
   & \multicolumn{4}{c!{\vrule width 0.8pt}}{\textbf{Reference}}
   & \multicolumn{12}{c}{\textbf{No-CSI}} \\
  \cmidrule(l{0.3pt}r{0.3pt}){3-6}\cmidrule(l{0.3pt}r{0.3pt}){7-18}
   & & \multirow{2}{*}{\algNONE[short]} & \multicolumn{3}{c!{\vrule width 0.8pt}}{\algMAX[short]}
     & \multicolumn{3}{c}{\algMYO[short]} & \multicolumn{3}{c}{\algCONS[short]} & \multicolumn{3}{c}{\algMA[short]} & \multicolumn{3}{c}{\algOURSNOCSIst[short]} \\
  \cmidrule(l{0.3pt}r{0.3pt}){4-6}\cmidrule(l{0.3pt}r{0.3pt}){7-9}\cmidrule(l{0.3pt}r{0.3pt}){10-12}\cmidrule(l{0.3pt}r{0.3pt}){13-15}\cmidrule(l{0.3pt}r{0.3pt}){16-18}
   & & & \cmpT & \cmpQ & \cmpIEight & \cmpT & \cmpQ & \cmpIEight & \cmpT & \cmpQ & \cmpIEight & \cmpT & \cmpQ & \cmpIEight & \cmpT & \cmpQ & \cmpIEight \\
  \midrule
  \multirow{3}{*}{\scenTbSSGPT}
   & $\metU$  & \nfeas{1.000} & 0.331 & 0.242 & 0.552 & \nfeas{0.686} & \nfeas{0.470} & \nfeas{0.742} & 0.607 & 0.443 & \underline{0.702} & \nfeas{0.718} & \nfeas{0.506} & \nfeas{0.798} & \nfeas{0.676} & 0.469 & \best{0.713} \tbrowL
   & $\metD$  & \nfeas{320} & 102 & 91.1 & 97.0 & \nfeas{307} & \nfeas{285} & \nfeas{293} & 154 & 143 & \underline{162} & \nfeas{288} & \nfeas{265} & \nfeas{284} & \nfeas{86.0} & 81.1 & \best{128} \tbrowL
   & $\metED$ & \nfeas{120} & -97.9 & -109 & -103 & \nfeas{107} & \nfeas{84.9} & \nfeas{92.7} & -45.7 & -57.3 & \underline{-38.0} & \nfeas{88.2} & \nfeas{65.1} & \nfeas{84.3} & \nfeas{16.5} & -119 & \best{-71.8} \tbrowL
  \midrule
  \multirow{3}{*}{\shortstack[l]{\scenPrsntVizWire}}
   & $\metU$  & 0.927 & 0.120 & 0.114 & 0.751 & \nfeas{0.900} & \best{0.927} & \nfeas{0.934} & \nfeas{0.880} & \underline{0.927} & \nfeas{0.935} & \nfeas{0.907} & 0.927 & \nfeas{0.934} & 0.855 & 0.648 & 0.892 \tbrowL
   & $\metD$  & 236 & 226 & 169 & 221 & \nfeas{286} & \best{175} & \nfeas{279} & \nfeas{288} & \underline{179} & \nfeas{279} & \nfeas{287} & 180 & \nfeas{279} & 183 & 184 & 187 \tbrowL
   & $\metED$ & 9.06 & -1.25 & -58.1 & -5.80 & \nfeas{58.7} & \best{-52.2} & \nfeas{51.4} & \nfeas{60.7} & \underline{-48.5} & \nfeas{51.6} & \nfeas{59.3} & -47.6 & \nfeas{51.7} & -44.1 & -43.6 & -39.9 \tbrowL
  \midrule
  \multirow{3}{*}{\shortstack[l]{\scenPrsntVizEdge}}
   & $\metU$  & \nfeas{0.927} & 0.120 & 0.114 & 0.751 & \nfeas{0.551} & \nfeas{0.933} & \nfeas{0.908} & \nfeas{0.543} & \nfeas{0.933} & \best{0.909} & \nfeas{0.546} & \nfeas{0.933} & \nfeas{0.909} & 0.533 & 0.684 & \underline{0.845} \tbrowL
   & $\metD$  & \nfeas{676} & 197 & 197 & 209 & \nfeas{286} & \nfeas{264} & \nfeas{255} & \nfeas{284} & \nfeas{268} & \best{249} & \nfeas{286} & \nfeas{264} & \nfeas{252} & 235 & 223 & \underline{206} \tbrowL
   & $\metED$ & \nfeas{438} & -40.9 & -40.6 & -28.6 & \nfeas{47.7} & \nfeas{26.2} & \nfeas{17.1} & \nfeas{46.4} & \nfeas{30.4} & \best{10.7} & \nfeas{47.7} & \nfeas{25.5} & \nfeas{14.1} & -2.65 & -15.4 & \underline{-32.2} \tbrowL
  \midrule
  \multirow{3}{*}{\shortstack[l]{\scenPrsntLangWikiWire}}
   & $\metU$  & \nfeas{1.000} & \nfeas{0.192} & 0.028 & 0.682 & \nfeas{1.000} & \nfeas{1.000} & \nfeas{1.000} & \nfeas{1.000} & \nfeas{1.000} & \nfeas{1.000} & \nfeas{1.000} & \nfeas{1.000} & \nfeas{1.000} & \underline{0.639} & 0.066 & \best{0.942} \tbrowL
   & $\metD$  & \nfeas{297} & \nfeas{282} & 236 & 274 & \nfeas{339} & \nfeas{295} & \nfeas{308} & \nfeas{339} & \nfeas{287} & \nfeas{317} & \nfeas{343} & \nfeas{294} & \nfeas{332} & \underline{262} & 261 & \best{236} \tbrowL
   & $\metED$ & \nfeas{33.4} & \nfeas{18.5} & -26.7 & 10.7 & \nfeas{75.4} & \nfeas{31.7} & \nfeas{44.8} & \nfeas{75.6} & \nfeas{24.3} & \nfeas{53.4} & \nfeas{80.0} & \nfeas{31.2} & \nfeas{68.5} & \underline{-1.06} & -2.35 & \best{-26.9} \tbrowL
  \midrule
  \multirow{3}{*}{\shortstack[l]{\scenPrsntLangWikiEdge}}
   & $\metU$  & \nfeas{1.000} & 0.192 & 0.028 & 0.677 & \nfeas{0.957} & \nfeas{0.013} & \nfeas{0.989} & \nfeas{0.956} & \nfeas{0.013} & \nfeas{0.987} & \nfeas{0.957} & \nfeas{0.013} & \nfeas{0.989} & \underline{0.691} & 0.087 & \best{0.857} \tbrowL
   & $\metD$  & \nfeas{790} & 287 & 335 & 292 & \nfeas{463} & \nfeas{457} & \nfeas{434} & \nfeas{454} & \nfeas{453} & \nfeas{438} & \nfeas{458} & \nfeas{449} & \nfeas{436} & \underline{338} & 314 & \best{344} \tbrowL
   & $\metED$ & \nfeas{420} & -83.8 & -35.0 & -78.7 & \nfeas{92.3} & \nfeas{87.1} & \nfeas{63.9} & \nfeas{83.4} & \nfeas{82.8} & \nfeas{67.6} & \nfeas{87.9} & \nfeas{78.3} & \nfeas{65.3} & \underline{-32.6} & -56.0 & \best{-26.6} \tbrowL
  \midrule
  \multirow{3}{*}{\shortstack[l]{\scenPrsntLangMmluEdge}}
   & $\metU$  & \nfeas{0.467} & 0.267 & 0.283 & 0.283 & 0.408 & \nfeas{0.300} & \nfeas{0.477} & \underline{0.417} & \nfeas{0.300} & \nfeas{0.481} & \best{0.417} & \nfeas{0.300} & \nfeas{0.478} & 0.269 & 0.282 & 0.397 \tbrowL
   & $\metD$  & \nfeas{901} & 274 & 280 & 297 & 385 & \nfeas{646} & \nfeas{874} & \underline{391} & \nfeas{651} & \nfeas{862} & \best{386} & \nfeas{656} & \nfeas{863} & 290 & 310 & 358 \tbrowL
   & $\metED$ & \nfeas{401} & -226 & -220 & -203 & -115 & \nfeas{146} & \nfeas{374} & \underline{-109} & \nfeas{151} & \nfeas{362} & \best{-114} & \nfeas{156} & \nfeas{363} & -210 & -190 & -142 \tbrowL
  \midrule
  \multirow{3}{*}{\shortstack[l]{\scenVisionJetsonTestbedS}}
   & $\metU$  & \nfeas{0.933} & 0.414 & 0.870 & 0.930 & 0.894 & 0.933 & 0.932 & 0.856 & \best{0.933} & 0.932 & 0.893 & \underline{0.933} & 0.932 & 0.894 & 0.925 & 0.930 \tbrowL
   & $\metD$  & \nfeas{3770} & 476 & 469 & 510 & 1960 & 1950 & 1960 & 1620 & \best{1780} & 1170 & 1940 & \underline{1940} & 1990 & 1900 & 1270 & 972 \tbrowL
   & $\metED$ & \nfeas{1880} & -1410 & -1430 & -1400 & 68.3 & 60.1 & 54.5 & -267 & \best{-114} & -733 & 53.4 & \underline{47.5} & 86.9 & 12.5 & -624 & -934 \tbrowL
  \midrule
  \multirow{3}{*}{\shortstack[l]{\scenLangJetsonTestbedS}}
   & $\metU$  & \nfeas{0.934} & 0.863 & 0.484 & 0.939 & \nfeas{0.939} & 0.486 & 0.939 & \best{0.942} & 0.486 & 0.939 & \nfeas{0.939} & 0.486 & 0.939 & 0.939 & 0.540 & \underline{0.939} \tbrowL
   & $\metD$  & \nfeas{237} & 98.1 & 95.5 & 96.2 & \nfeas{235} & 179 & 249 & \best{148} & 153 & 129 & \nfeas{230} & 178 & 243 & 149 & 95.9 & \underline{98.9} \tbrowL
   & $\metED$ & \nfeas{65.2} & -95.2 & -76.1 & -156 & \nfeas{41.5} & 7.86 & -3.39 & \best{-45.0} & -18.1 & -123 & \nfeas{36.2} & 6.60 & -9.04 & -44.0 & -75.7 & \underline{-153} \tbrowL
  \bottomrule
  \end{tabular}}
  \end{minipage}%
  \hspace{\TestbedHGap}%
  \begin{minipage}[t]{\TestbedHalfWidth}
  \centering
  \setlength{\tabcolsep}{\TestbedTabColSepR}
  \renewcommand{\arraystretch}{1}
  \medskip
  \hrule
  \smallskip
  {\bfseries Multi-task Scenarios}
  \vspace{1.5pt}
  \smallskip
  \hrule
  \medskip
  \resizebox{\linewidth}{!}{%
  \begin{tabular}{@{}ll c ccc !{\vrule width 0.8pt} ccc ccc ccc ccc@{}}
  \toprule
  \multirow{3}{*}{\textbf{Scenario}} & \multirow{3}{*}{\textbf{Metric}}
   & \multicolumn{4}{c!{\vrule width 0.8pt}}{\textbf{Reference}}
   & \multicolumn{12}{c}{\textbf{No-CSI}} \\
  \cmidrule(l{0.3pt}r{0.3pt}){3-6}\cmidrule(l{0.3pt}r{0.3pt}){7-18}
   & & \multirow{2}{*}{\algNONE[short]} & \multicolumn{3}{c!{\vrule width 0.8pt}}{\algMAX[short]}
     & \multicolumn{3}{c}{\algDES[short]} & \multicolumn{3}{c}{\algDLPROP[short]} & \multicolumn{3}{c}{\algMA[short]} & \multicolumn{3}{c}{\algOURSNOCSImt[short]} \\
  \cmidrule(l{0.3pt}r{0.3pt}){4-6}\cmidrule(l{0.3pt}r{0.3pt}){7-9}\cmidrule(l{0.3pt}r{0.3pt}){10-12}\cmidrule(l{0.3pt}r{0.3pt}){13-15}\cmidrule(l{0.3pt}r{0.3pt}){16-18}
   & & & \cmpT & \cmpQ & \cmpIEight & \cmpT & \cmpQ & \cmpIEight & \cmpT & \cmpQ & \cmpIEight & \cmpT & \cmpQ & \cmpIEight & \cmpT & \cmpQ & \cmpIEight \\
  \midrule
  \multirow{3}{*}{\scenPrsntVizVizWire}
   & $\metU$  & \nfeas{0.924} & \nfeas{0.094} & \nfeas{0.086} & \nfeas{0.758} & \nfeas{0.094} & \nfeas{0.086} & \nfeas{0.758} & \nfeas{0.094} & \nfeas{0.086} & \nfeas{0.758} & \nfeas{0.094} & \nfeas{0.086} & \nfeas{0.758} & 0.109 & \underline{0.109} & \best{0.759} \tbrowR
   & $\metD$  & \nfeas{430} & \nfeas{415} & \nfeas{376} & \nfeas{383} & \nfeas{390} & \nfeas{351} & \nfeas{394} & \nfeas{388} & \nfeas{343} & \nfeas{393} & \nfeas{393} & \nfeas{343} & \nfeas{396} & 286 & \underline{283} & \best{294} \tbrowR
   & $\metED$ & \nfeas{193} & \nfeas{163} & \nfeas{85.6} & \nfeas{98.9} & \nfeas{113} & \nfeas{34.9} & \nfeas{120} & \nfeas{109} & \nfeas{18.4} & \nfeas{120} & \nfeas{119} & \nfeas{20.2} & \nfeas{125} & -95.1 & \underline{-101} & \best{-79.5} \tbrowR
  \midrule
  \multirow{3}{*}{\shortstack[l]{\scenPrsntVizVizEdge}}
   & $\metU$  & \nfeas{0.924} & 0.094 & 0.086 & 0.758 & 0.094 & 0.086 & 0.763 & 0.094 & 0.086 & \underline{0.763} & 0.094 & 0.086 & \nfeas{0.763} & 0.389 & 0.325 & \best{0.765} \tbrowR
   & $\metD$  & \nfeas{1400} & 465 & 412 & 469 & 424 & 432 & 487 & 436 & 433 & \underline{483} & 426 & 428 & \nfeas{493} & 358 & 344 & \best{369} \tbrowR
   & $\metED$ & \nfeas{1840} & -21.8 & -128 & -13.6 & -104 & -89.1 & 21.1 & -80.6 & -86.8 & \underline{13.2} & -101 & -95.8 & \nfeas{34.0} & -236 & -264 & \best{-215} \tbrowR
  \midrule
  \multirow{3}{*}{\shortstack[l]{\scenPrsntVizLangEdge}}
   & $\metU$  & \nfeas{0.962} & 0.153 & 0.059 & 0.737 & 0.153 & 0.059 & \underline{0.737} & 0.153 & 0.059 & 0.737 & 0.153 & 0.059 & 0.737 & 0.450 & 0.232 & \best{0.844} \tbrowR
   & $\metD$  & \nfeas{1350} & 492 & 431 & 487 & 475 & 480 & \underline{471} & 471 & 449 & 483 & 474 & 452 & 487 & 388 & 380 & \best{446} \tbrowR
   & $\metED$ & \nfeas{1370} & -349 & -471 & -360 & -384 & -374 & \underline{-392} & -391 & -436 & -368 & -384 & -429 & -360 & -557 & -574 & \best{-441} \tbrowR
  \midrule
  \multirow{3}{*}{\shortstack[l]{\scenPrsntVizLangWlan}}
   & $\metU$  & \nfeas{0.947} & 0.147 & 0.059 & 0.707 & 0.147 & 0.059 & 0.707 & 0.147 & 0.059 & 0.707 & 0.147 & 0.059 & \underline{0.707} & 0.184 & 0.072 & \best{0.716} \tbrowR
   & $\metD$  & \nfeas{6110} & 803 & 962 & 1250 & 766 & 940 & 1270 & 767 & 951 & 1250 & 769 & 968 & \underline{1230} & 865 & 996 & \best{1220} \tbrowR
   & $\metED$ & \nfeas{9360} & -1250 & -934 & -358 & -1320 & -976 & -325 & -1320 & -956 & -352 & -1320 & -922 & \underline{-401} & -1130 & -864 & \best{-405} \tbrowR
  \midrule
  \multirow{3}{*}{\shortstack[l]{\scenPrsntLangLangEdge}}
   & $\metU$  & \nfeas{1.000} & 0.211 & 0.033 & 0.716 & 0.211 & 0.033 & 0.716 & 0.211 & 0.033 & \underline{0.716} & 0.211 & 0.033 & 0.716 & 0.594 & 0.033 & \best{0.898} \tbrowR
   & $\metD$  & \nfeas{1490} & 554 & 519 & 518 & 573 & 488 & 512 & 577 & 491 & \underline{510} & 576 & 477 & 511 & 491 & 495 & \best{518} \tbrowR
   & $\metED$ & \nfeas{1560} & -321 & -390 & -393 & -283 & -453 & -405 & -274 & -446 & \underline{-409} & -277 & -475 & -407 & -446 & -438 & \best{-393} \tbrowR
  \midrule
  \multirow{3}{*}{\shortstack[l]{\scenMTJetsonTestbedM}}
   & $\metU$  & \nfeas{0.926} & 0.418 & 0.860 & 0.929 & 0.707 & \underline{0.915} & \nfeas{0.932} & \nfeas{0.836} & \best{0.915} & \nfeas{0.932} & \nfeas{0.928} & \nfeas{0.927} & \nfeas{0.933} & 0.747 & 0.913 & \nfeas{0.932} \tbrowR
   & $\metD$  & \nfeas{4330} & 1090 & 957 & 1680 & 2440 & \underline{1230} & \nfeas{2880} & \nfeas{3590} & \best{1210} & \nfeas{2670} & \nfeas{4880} & \nfeas{3990} & \nfeas{13800} & 2400 & 1450 & \nfeas{2730} \tbrowR
   & $\metED$ & \nfeas{1900} & -1350 & -1480 & -760 & 9.55 & \underline{-1210} & \nfeas{443} & \nfeas{1160} & \best{-1230} & \nfeas{234} & \nfeas{2450} & \nfeas{1550} & \nfeas{11400} & -30.3 & -986 & \nfeas{285} \tbrowR
  \bottomrule
  \end{tabular}}
  \end{minipage}
  \vspace{-8pt}
     \techreport{%
  \vspace{10pt}
}{}
\end{table*}

\begin{figure*}[!t]
    \centering
    \includegraphics[width=1.0\textwidth]{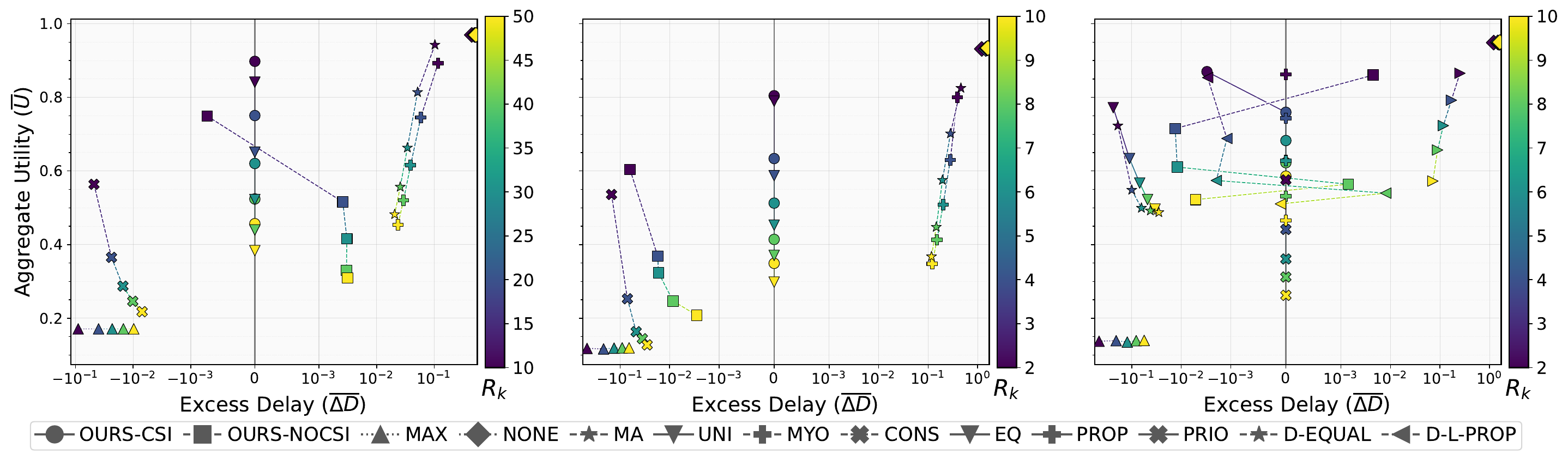}
    \vspace{-15pt}
    \techreport{%
\vspace{-1pt}
}{\vspace{-7pt}}
    \caption{Aggregate utility versus excess delay for single-task and multi-task scenarios under varying target throughputs ($R_k$) and Top-$k$ compression. Each line represents the trajectory of a specific algorithm. The left panel displays single-task \taskVisionOne, the middle panel shows single-task \taskVisionTwo, and the right panel illustrates a multi-task scenario combining both workloads with equal $R_k$.}
    \label{fig:topk-varyR}
    \vspace{-5pt}
       \techreport{%
\vspace{3pt}}{\vspace{-5pt}}
\end{figure*}

\fussy

\vspace{-2pt}
\subsection{Experimental Results}
\label{sec:experimental_results}

\techreport{%
\vspace{4pt}
}{}
 \noindent\textbf{Offline Experiments.}
 Table~\ref{tab:results-offline-all} shows our offline experiment results. Across all scenarios, no-compression (\algNONE[short]) is infeasible. In single-task experiments, among CSI-aware methods, \algOURSCSI[short] consistently achieves the highest feasible utility, whereas \algUNI[short] bottlenecks all links to the worst-case channel, sacrificing up to 0.03 utility under \cmpIEight. In the No-CSI group, \algMYO[short] and \algMA[short] are infeasible in nearly every scenario, while \algCONS[short] is feasible through pessimistic estimation at a substantial accuracy cost. \algOURSNOCSIst[short] balances both: it meets the delay constraint in most scenarios while maintaining utility competitive with the CSI-aware policy, trailing \algCONS[short] in excess delay but exceeding it by up to 0.09 in utility.

For multi-task scenarios, resource contention makes the allocation strategy as important as the compression policy. \algOURSCSI[short] achieves the best feasible utility in every scenario. On \scenMTLVJM{} with \cmpIEight, it  attains the highest utility and lowest delay, outperforming other CSI-aware baselines by a wide margin. \algEQ[short]{} and \algPROP[short]{} are feasible but cannot exploit asymmetric accuracy sensitivities across tasks. In the No-CSI group, \algOURSNOCSImt[short] has the best feasible utility in three of four scenarios under \cmpIEight, while \algDES[short] and \algDLPROP[short] are infeasible under high task heterogeneity. 


\techreport{%
\vspace{2pt}
}{}
\noindent\textbf{Testbed Experiments.}
The results in Table~\ref{tab:results-testbed-all-nocsi} confirm the offline trends under realistic conditions. We report on No-CSI algorithms as the channel is hard to estimate (we discuss CSI-aware testbed deployment challenges in \techreport{App.~\ref{app:additional_exp}}{App.~K in \cite{techreport}}). \algMYO[short] and \algMA[short] are infeasible across most scenarios, consistent with our offline findings, while \algCONS[short] is feasible through pessimistic estimation at reduced utility. \algOURSNOCSIst[short] with \cmpIEight{} is the only No-CSI method achieving both feasibility and competitive utility on the challenging LLM scenarios (\scenPrsntLangWikiWire: $0.942$, \scenPrsntLangWikiEdge: $0.857$), where all other No-CSI methods are infeasible. On the Jetson testbed, \algCONS[short] and \algOURSNOCSIst[short] are the only methods always meeting the delay constraint, with \algOURSNOCSIst[short] matching utility at lower excess delay.
On the hardest scenario (\scenPrsntVizVizWire), all methods  are infeasible except \algOURSNOCSImt[short], which finds a feasible operating point ($0.759$ utility with \cmpIEight). On \scenPrsntVizVizEdge, \algOURSNOCSImt[short] with \cmpIEight{} achieves the best feasible utility ($0.765$) at the lowest delay, demonstrating efficient resource usage under imperfect channel knowledge. The cross-modal \scenPrsntVizLangEdge{}  confirms this: \algOURSNOCSImt[short] with \cmpIEight{} has the highest utility ($0.844$) with the most negative excess delay ($-440$\,ms). On the Jetson multi-task scenario, \algDLPROP[short] and \algOURSNOCSImt[short] with \cmpQ{} are the only feasible entries, both above $0.91$ utility.


\techreport{%
\vspace{3pt}
}{}
\noindent\textbf{Varying Target Throughputs.}
Fig.~\ref{fig:topk-varyR} (left) isolates a single-task setup deploying an MLP on the MNIST dataset across a 3-node routing path, whereas Fig.~\ref{fig:topk-varyR} (middle) examines a ResNet-56 model processing CIFAR-10 data over a 4-node path. Fig.~\ref{fig:topk-varyR} (right) illustrates a multi-task scenario across a 4-node path, where the MLP and ResNet models share the first three nodes and two connecting links. Both tasks share an identical target throughput, with utility weights set to 0.7 for the MLP and 1.0 for the ResNet. The accuracy-compression functions in these setups are based on Stein's approach for the MLP task and the fitting approach for the ResNet task.
In Fig.~\ref{fig:topk-varyR} (left) and (middle),  \algNONE[short] and \algMAX[short] establish performance extremes. With perfect channel knowledge, \algOURSCSI[short] strictly enforces the target rate while maximizing utility, outperforming \algUNI[short]. In the no-CSI case, \algOURSNOCSIst[short] tightly tracks the zero excess delay axis, avoiding severe constraint violations of \algMYO[short] and \algMA[short], as well as unnecessary utility degradation of over-pessimistic \algCONS[short]. 
Fig.~\ref{fig:topk-varyR} (right) demonstrates multi-task resource contention between MLP and ResNet-56 pipelines. \algOURSCSI[short] has non-positive excess delay while experiencing optimal utility via multi-dimensional water-filling, whereas static CSI-aware heuristics (\algEQ[short], \algPROP[short]) and greedy policies (\algPRIO[short]) fail to catch up. Crucially, \algOURSNOCSImt[short] achieves a highly competitive delay-utility profile via block coordinate descent. By jointly optimizing the tasks, it outperforms No-CSI baselines like \algDES[short] and \algDLPROP[short], which optimize tasks in isolation. Furthermore, the \algMA[short] baseline incurs high constraint violations because its weak estimation fails to track the highly variable channel.

\vspace{-1pt}
\section{\label{sec:conc} Conclusion}

    We attempt to balance communication reduction against the accuracy of the downstream inference task. 
We achieve this by introducing a dynamically adjustable latent compression factor that continuously scales the intermediate payload size in response to fluctuating, real-time channel capacities, ensuring that strict QoS throughput constraints are satisfied without unnecessary degradation of inference accuracy.  \CRedit{Compression cannot, however, overcome a link outage: once a capacity falls below the demand of the maximally compressed payload, no compression level is feasible, and the system must buffer, shed load, or exit inference early. Early-exit policies for such outages, along with autoregressive workloads whose prefill and decode phases impose different compute and activation loads, are natural directions for future work.}

\techreport{%
}{\vspace{5pt}}

\begin{acks}
\CRedit{This work was supported by the National Science Foundation through the AI-EDGE Institute (Award~No.~2112471), by the Army Research Laboratory under Grant~No.~W911NF-24-2-0172, and by the Army Research Office under Grant~No.~W911NF-24-1-0103.}

\end{acks}

\newcommand{\showDOI}[1]{\unskip}
\newcommand{\showURL}[1]{\unskip}
\renewcommand{\showeprint}[2][]{\unskip}
\bibliographystyle{ACM-Reference-Format}
\bibliography{sample-base}

\techreport{%
\newpage
\appendix
\onecolumn
\fancyhead[LE,RO]{}%
\renewcommand{\headrulewidth}{0pt}%
\section{Proof of Theorem~\ref{thm:single-convex}} \label{proof:thm:single-convex}
\begin{proof}
We begin by determining the feasibility conditions for Prob.~\ref{prob:single-task}: 
\begin{lemma}[Single-Task Feasibility Condition]\label{lem:feasibility}
Problem \ref{prob:single-task} is feasible if and only if: 
 \begin{align}
     \tau_i \le \frac{1}{R(t)}, & \quad \forall i \in [L],
    & &\text{and} & 
      \frac{a_i \eta_i^{\min}}{c_{i}(t)} \le \frac{1}{R(t)} ,& \quad \forall i \in [L-1].
       \label{cond:feas-tau}
 \end{align}
\end{lemma}
\begin{proof}
We prove this in both directions.

\noindent\textbf{Necessity:} Assume the problem is feasible. Let $\boldsymbol{\eta}(t)$ be a feasible solution. The constraint in \eqref{cons:max} dictates that the maximum of a set is bounded by $\frac{1}{R(t)}$, meaning every individual element in the set must be bounded by $\frac{1}{R(t)}$. 
Consequently, for the computation delays, it directly establishes \eqref{cond:feas-tau}. 
For the communication delays, we must have $\frac{a_i}{c_{e_i}(t)}\eta_i(t) \le \frac{1}{R(t)}$ for all $i \in [L-1]$. Because the solution is drawn from the valid domain, $\eta_i(t) \ge \eta_i^{\min}$. Substituting this lower bound preserves the inequality:
$$ \frac{a_i \eta_i^{\min}}{c_{e_i}(t)} \le \frac{a_i \eta_i(t)}{c_{e_i}(t)} \le \frac{1}{R(t)}, \quad \forall i \in [L-1]. $$

\noindent\textbf{Sufficiency:} Assume condition \eqref{cond:feas-tau} holds. We demonstrate feasibility by constructing a valid solution. Let $\eta_i(t) = \eta_i^{\min}$ for all $i \in [L-1]$. By definition, this vector satisfies the domain constraint $\boldsymbol{\eta}(t) \in \prod_{i=1}^{L-1} [\eta_i^{\min}, 1]$.
Condition \eqref{cond:feas-tau} ensures that $\tau_i \le \frac{1}{R(t)}$ for all $i \in [L]$ and  $\frac{a_i}{c_{e_i}(t)}\eta_i^{\min} \le \frac{1}{R(t)}$ for all $i \in [L-1]$. 
Since every individual computation and communication delay term under this configuration is bounded by $\frac{1}{R(t)}$, their maximum is also bounded by $\frac{1}{R(t)}$. Thus, $\boldsymbol{\eta}(t) = \boldsymbol{\eta}^{\min}$ satisfies \eqref{cons:max}, proving the feasible set is non-empty.
\end{proof}


To  prove Theorem~\ref{thm:single-convex}, observe that, since the problem is feasible, the inequalities regarding $\tau_i$ are inherently satisfied and independent of the choice of $\boldsymbol{\eta}(t)$. Thus, we can rewrite the constraints entirely in terms of $\eta_i(t)$. The upper bound imposed by the target delay $\frac{1}{R(t)}$ is $\eta_i(t) \le \frac{c_{e_i}(t)}{R(t)a_i}$. Combined with the natural domain constraint $\eta_i(t) \le 1$, the feasible region for each compression factor is:
$$ \eta_i(t) \in \left[ \eta_i^{\min},\, \min\!\left\{ 1,\, \frac{c_{e_i}(t)}{R(t)a_i} \right\} \right] $$
Because $A(\boldsymbol{\eta}(t))$ is monotonically non-decreasing with respect to $\eta_i(t)$, maximizing the objective function strictly equates to maximizing each $\eta_i(t)$ independently. Therefore, the optimal value $\eta_i^\star(t)$ is the upper bound of its feasible interval.
\end{proof}


\section{Proof of Lemma \ref{lem:multi-feasibility}}\label{proof:lem:multi-feasibility}
\begin{proof}
We prove this by demonstrating both necessity and sufficiency.

\noindent\textbf{Necessity:} Assume Problem \ref{prob:multi-task} is feasible. Thus, there exist sets of variables $\mathbf{s}^{\text{comp}}(t)$, $\mathbf{s}^{\text{comm}}(t)$, and $\boldsymbol{\eta}(t)$ satisfying all constraints. From \eqref{cons:het-comp-lat}, we rearrange to find $s_{i,k}^{\text{comp}}(t) \ge \tau_{i,k} R_k(t)$. Summing this inequality over all active tasks $k \in \mathcal{K}(t)$ and their corresponding stages $i$ assigned to any physical node $v \in \mathcal{V}$ (where $v_{i,k} = v$), and applying the capacity constraint \eqref{eq:comp_cap} yields:
$$ 1 \ge \sum_{k \in \mathcal{K}(t)} \sum_{i: \, v_{i,k} = v} s_{i,k}^{\text{comp}}(t) \ge \sum_{k \in \mathcal{K}(t)} \sum_{i: \, v_{i,k} = v} \tau_{i,k} R_k(t), $$
which establishes the first condition.
Similarly, rearranging \eqref{cons:het-comm-lat} yields $s_{i,k}^{\text{comm}}(t) \ge \frac{a_{i,k} \eta_{i,k}(t) R_k(t)}{c_{i,k}(t)}$. By the domain constraint \eqref{cons:het-eta-domain}, $\eta_{i,k}(t) \ge \eta_{i,k}^{\min}$. Substituting this lower bound and summing over all active tasks $k \in \mathcal{K}(t)$ and their corresponding hops $i$ assigned to any physical edge $e \in \mathcal{E}$ (where $e_{i,k} = e$) gives:
$$ 1 \ge \sum_{k \in \mathcal{K}(t)} \sum_{i: \, e_{i,k} = e} s_{i,k}^{\text{comm}}(t) \ge \sum_{k \in \mathcal{K}(t)} \sum_{i: \, e_{i,k} = e} \frac{a_{i,k} \eta_{i,k}^{\min} R_k(t)}{c_{i,k}(t)}. $$

\noindent\textbf{Sufficiency:} Assume the conditions in the lemma
hold. We construct a valid solution by setting the variables to their minimum required limits: let $s_{i,k}^{\text{comp}}(t) = \tau_{i,k} R_k(t)$, $\eta_{i,k}(t) = \eta_{i,k}^{\min}$, and $s_{i,k}^{\text{comm}}(t) = \frac{a_{i,k} \eta_{i,k}^{\min} R_k(t)}{c_{i,k}(t)}$ for all corresponding $i, k$. 
By definition, $\eta_{i,k}(t)$ satisfies \eqref{cons:het-eta-domain}. These allocations exactly satisfy the latency constraints \eqref{cons:het-comp-lat} and \eqref{cons:het-comm-lat} as equalities. Furthermore, by the assumed conditions in the lemma statement,
the sums of these allocated shares do not exceed $1$, thus satisfying the node and edge capacity constraints \eqref{eq:comp_cap} and \eqref{eq:comm_cap}. Therefore, a feasible solution exists.
\end{proof}

\section{Proof of Theorem \ref{thm:multi-convex}} \label{proof:thm:multi-convex}

\begin{proof}
Because the objective function \eqref{prob:het-main} is a monotonically non-decreasing function of $\boldsymbol{\eta}_k(t)$, maximizing the objective equates to maximizing each $\eta_{i,k}(t)$ independently, constrained only by the allocated bandwidth $s_{i,k}^{\text{comm}}(t)$.
From \eqref{cons:het-comm-lat} and \eqref{cons:het-eta-domain}, the feasible upper bound for compression is:
$$ \eta_{i,k}(t) \le \min\left(1, \frac{c_{i,k}(t)}{R_k(t) a_{i,k}} s_{i,k}^{\text{comm}}(t) \right). $$
Since we aim to maximize accuracy, the optimal compression factor is exactly this upper bound. We can define this relationship as a function of the allocated bandwidth:
$$ \eta_{i,k}^*(s_{i,k}^{\text{comm}}(t)) \triangleq\min\left(1, \frac{c_{i,k}(t)}{R_k(t) a_{i,k}} s_{i,k}^{\text{comm}}(t) \right). $$
The function $\eta^*_{i,k}$ is the pointwise minimum of two affine functions, and is therefore concave. 
Substituting $\eta_{i,k}(t) = \eta^*_{i,k}(s_{i,k}^{\text{comm}}(t))$ into the objective function yields $\sum_{k=1}^{K} w_k A_k(\eta^*_{i,k}(\mathbf{s}^{\text{comm}}(t)))$. 
Because $A_k$ is concave and non-decreasing, the composition of $A_k$ with the concave function $\eta^*_{i,k}$ is also concave. A positively weighted sum of concave functions remains concave. 

Furthermore, the computation allocation variables $\mathbf{s}^{\text{comp}}(t)$ are decoupled from the objective and can be trivially satisfied by assigning $s_{i,k}^{\text{comp}}(t) = \tau_{i,k} R_k(t)$ (guaranteed valid by Lemma \ref{lem:multi-feasibility}).
Thus, the optimization simplifies to maximizing a concave objective function subject to linear inequality constraints on the communication allocations:
\begin{align}
  \text{Maximize:} \quad & \sum_{k=1}^{K} w_k A_k\left( \dots, \min\left(1, \frac{c_{i,k}(t)}{R_k(t) a_{i,k}} s_{i,k}^{\text{comm}}(t) \right), \dots \right) \nonumber \\
  \text{subj.~to:} \quad & \sum_{k \in \mathcal{K}(t)} \sum_{i: \, e_{i,k} = e} s_{i,k}^{\text{comm}}(t) \le 1, \quad \forall e \in \mathcal{E}, \nonumber \\
  & s_{i,k}^{\text{comm}}(t) \ge \frac{a_{i,k} \eta_{i,k}^{\min} R_k(t)}{c_{i,k}(t)}, \quad \forall k \in \mathcal{K}(t), \; \forall i \in [L_k-1]. \nonumber
\end{align}
Maximizing a concave function over a convex set defined by linear constraints constitutes a standard convex optimization problem, completing the proof.
\end{proof}

\section{Water-Filling Algorithm}\label{app:water}

\paragraph{Optimality Conditions (KKT)} 
To gain insight into the convex optimization problem formulated in Theorem \ref{thm:multi-convex}, we apply the Karush-Kuhn-Tucker (KKT) conditions. We attach a Lagrange multiplier $\lambda_e \ge 0$ to the communication bandwidth capacity constraint of each physical edge $e \in \mathcal{E}$, and a multiplier $\mu_{i,k} \ge 0$ to the minimum QoS bandwidth constraint for the $i$-th transmission hop of task $k$.

Let $s_{i,k}^{\min}(t) = \frac{a_{i,k} \eta_{i,k}^{\min} R_k(t)}{c_{i,k}(t)}$ denote the minimum bandwidth fraction required for task $k$ at its $i$-th hop during time slot $t$. The Lagrangian is:

\begin{align}
    \mathcal{L} (\boldsymbol{\eta}(t),\boldsymbol{s}^\cmp(t),\boldsymbol{s}^\com(t),\boldsymbol{\lambda},\boldsymbol{\mu})&= \sum_{k=1}^{K} w_k A_k(\boldsymbol{\eta}_k(t))  - \sum_{e \in \mathcal{E}} \lambda_e \left( \sum_{k \in \mathcal{K}(t)} \sum_{i: \, e_{i,k} = e} s_{i,k}^{\com}(t) - 1 \right)  - \sum_{k \in \mathcal{K}(t)} \sum_{i=1}^{L_k-1} \mu_{i,k} \left( s_{i,k}^{\min}(t) - s_{i,k}^{\com}(t) \right).
\end{align}

Assuming the solution lies in the active region where $\eta_{i,k}(t) < 1$, we take the partial derivative of the Lagrangian with respect to the bandwidth allocation $s_{i,k}^{\com}(t)$ and apply the chain rule. Note that differentiating the capacity penalty for a specific $s_{i,k}^{\com}(t)$ isolates the shadow price $\lambda_{e_{i,k}}$ of the specific edge traversed by that hop:

\begin{equation}
    \frac{\partial \mathcal{L}}{\partial s_{i,k}^{\com}(t)} = w_k \frac{\partial A_k}{\partial \eta_{i,k}(t)} \cdot \frac{\partial \eta_{i,k}(t)}{\partial s_{i,k}^{\com}(t)} - \lambda_{e_{i,k}} + \mu_{i,k} = 0.
\end{equation}

Substituting the partial derivative $\frac{\partial \eta_{i,k}(t)}{\partial s_{i,k}^{\com}(t)} = \frac{c_{i,k}(t)}{R_k(t) a_{i,k}}$, we obtain the stationarity condition:

\begin{equation}
    \frac{w_k c_{i,k}(t)}{R_k(t) a_{i,k}} \frac{\partial A_k}{\partial \eta_{i,k}(t)} + \mu_{i,k} = \lambda_{e_{i,k}}, \quad \forall k \in \mathcal{K}(t), \; \forall i \in [L_k-1].
    \label{eq:kkt-stationarity}
\end{equation}

Furthermore, the KKT complementary slackness conditions require that:
\begin{align}
    \lambda_e \left( \sum_{k \in \mathcal{K}(t)} \sum_{i: \, e_{i,k} = e} s_{i,k}^{\com}(t) - 1 \right) &= 0, \quad \forall e \in \mathcal{E}, \label{eq:kkt-slack-cap}\\
    \mu_{i,k} \left( s_{i,k}^{\min}(t) - s_{i,k}^{\com}(t) \right) &= 0, \quad \forall k \in \mathcal{K}(t), \; \forall i \in [L_k-1]. \label{eq:kkt-slack-min}
\end{align}

\paragraph{Algorithm.} Equations \eqref{eq:kkt-stationarity} through \eqref{eq:kkt-slack-min} reveal a \textbf{Per-Link Water-Filling with a QoS Floor} strategy. 
\begin{itemize}
    \item Each physical edge $e$ possesses a shadow price $\lambda_e$, representing the marginal utility of additional bandwidth at that link.
    \item If a task $k$ is allocated strictly more than its minimum required bandwidth ($s_{i,k}^{\com}(t) > s_{i,k}^{\min}(t)$) for hop $i$, the complementary slackness condition \eqref{eq:kkt-slack-min} dictates that $\mu_{i,k} = 0$. For these tasks, the solver equalizes their effective marginal accuracy gains to the traversed edge's shadow price $\lambda_{e_{i,k}}$.
    \item If a task's marginal utility is too low to compete for additional bandwidth, the optimization algorithm will attempt to starve it. However, the constraint halts the reduction at $s_{i,k}^{\min}(t)$. At this boundary, $\mu_{i,k} > 0$ compensates for the task's lack of competitiveness, artificially satisfying the stationarity condition \eqref{eq:kkt-stationarity} and pegging the task to its minimum viable bandwidth.
\end{itemize}

\section{Proof of Theorem \ref{thm:performance_bounds}} \label{proof:thm:performance_bounds}

\begin{proof}
\textbf{Part A: Proof of Constraint Satisfaction}\\
From the dual update rule $\lambda_{t+1} = \max\left\{\epsilon, \lambda_t + D(\boldsymbol{\eta}(t), \mathbf{c}(t)) - \frac{1}{R(t)}\right\}$, we establish the fundamental inequality:
\begin{equation*}
    \lambda_{t+1} \ge \lambda_t + D(\boldsymbol{\eta}(t), \mathbf{c}(t)) - \frac{1}{R(t)}
\end{equation*}

Rearranging to isolate the instantaneous constraint violation yields:
\begin{equation*}
    D(\boldsymbol{\eta}(t), \mathbf{c}(t)) - \frac{1}{R(t)} \le \lambda_{t+1} - \lambda_t
\end{equation*}

Taking the expectation and averaging over $t = 1, \dots, T$ produces a telescoping sum:
\begin{equation*}
    \frac{1}{T} \sum_{t=1}^T \mathbb{E}\left[D(\boldsymbol{\eta}(t), \mathbf{c}(t)) - \frac{1}{R(t)}\right] \le \frac{\mathbb{E}[\lambda_{T+1}] - \mathbb{E}[\lambda_1]}{T}
\end{equation*}

As rigorously established in Part B, the expected dual variable $\mathbb{E}[\lambda_t]$ is uniformly bounded by a constant. Consequently, $\lim_{T \to \infty} \frac{\mathbb{E}[\lambda_{T+1}]}{T} = 0$. Taking the limit as $T \to \infty$ gives:
\begin{equation*}
    \lim_{T \to \infty} \frac{1}{T} \sum_{t=1}^T \mathbb{E}\left[D(\boldsymbol{\eta}(t), \mathbf{c}(t)) - \frac{1}{R(t)}\right] \le 0
\end{equation*}

This confirms that the long-term expected delay constraint is strictly satisfied.

\vspace{1em}
\noindent\textbf{Part B: Bounding the Multiplier}\\
Define the Lyapunov function as $L(\lambda_t) = \frac{1}{2}\lambda_t^2$. We utilize the projection property $(\max\{\epsilon, x\})^2 \le x^2 + \epsilon^2$ for all $x \in \mathbb{R}$. Applying this to the dual update rule yields:
\begin{equation*}
    \lambda_{t+1}^2 \le \left(\lambda_t + D(\boldsymbol{\eta}(t), \mathbf{c}(t)) - \frac{1}{R(t)}\right)^2 + \epsilon^2
\end{equation*}

Expanding the square and rearranging provides the Lyapunov drift bound:
\begin{equation*}
    \frac{1}{2}\lambda_{t+1}^2 - \frac{1}{2}\lambda_t^2 \le \lambda_t\left(D(\boldsymbol{\eta}(t), \mathbf{c}(t)) - \frac{1}{R(t)}\right) + \frac{1}{2}\left(D(\boldsymbol{\eta}(t), \mathbf{c}(t)) - \frac{1}{R(t)}\right)^2 + \frac{1}{2}\epsilon^2
\end{equation*}

To introduce the utility-delay trade-off, we subtract the scaled utility $\frac{1}{\mu} A(\boldsymbol{\eta}(t))$ from both sides. Let $\Delta(\boldsymbol{\eta}(t)) = D(\boldsymbol{\eta}(t), \mathbf{c}(t)) - \hat{D}(\boldsymbol{\eta}(t), \hat{\mathbf{c}}(t))$ denote the estimation error. Adding and subtracting the estimator term $\lambda_t \hat{D}(\boldsymbol{\eta}(t), \hat{\mathbf{c}}(t))$ on the right-hand side yields:
\begin{align*}
    \frac{1}{2}\lambda_{t+1}^2 - \frac{1}{2}\lambda_t^2 - \frac{1}{\mu}A(\boldsymbol{\eta}(t)) \le & -\frac{1}{\mu}A(\boldsymbol{\eta}(t)) + \lambda_t\left(\hat{D}(\boldsymbol{\eta}(t), \hat{\mathbf{c}}(t)) - \frac{1}{R(t)}\right) \\
    & + \lambda_t\Delta(\boldsymbol{\eta}(t)) + \frac{1}{2}\left(D(\boldsymbol{\eta}(t), \mathbf{c}(t)) - \frac{1}{R(t)}\right)^2 + \frac{1}{2}\epsilon^2
\end{align*}

By design, our algorithm minimizes the primary terms $-\frac{1}{\mu}A(\boldsymbol{\eta}(t)) + \lambda_t\left(\hat{D}(\boldsymbol{\eta}(t), \hat{\mathbf{c}}(t)) - \frac{1}{R(t)}\right)$ with respect to any alternative choice of $\boldsymbol{\eta}(t)$. We can thus upper-bound this expression by evaluating it at the strictly feasible, maximum compression baseline policy $\boldsymbol{\eta}^{\min}$:
\begin{align*}
    \frac{1}{2}\lambda_{t+1}^2 - \frac{1}{2}\lambda_t^2 - \frac{1}{\mu}A(\boldsymbol{\eta}(t)) \le & -\frac{1}{\mu}A(\boldsymbol{\eta}^{\min}) + \lambda_t\left(\hat{D}(\boldsymbol{\eta}^{\min}, \hat{\mathbf{c}}(t)) - \frac{1}{R(t)}\right) \\
    & + \lambda_t\Delta(\boldsymbol{\eta}(t)) + \frac{1}{2}\left(D(\boldsymbol{\eta}(t), \mathbf{c}(t)) - \frac{1}{R(t)}\right)^2 + \frac{1}{2}\epsilon^2
\end{align*}

Substituting the identity $\hat{D}(\boldsymbol{\eta}^{\min}, \hat{\mathbf{c}}(t)) = D(\boldsymbol{\eta}^{\min}, \mathbf{c}(t)) - \Delta(\boldsymbol{\eta}^{\min})$ and rearranging terms yields:
\begin{align*}
    \frac{1}{2}\lambda_{t+1}^2 - \frac{1}{2}\lambda_t^2 - \frac{1}{\mu}A(\boldsymbol{\eta}(t)) \le & -\frac{1}{\mu}A(\boldsymbol{\eta}^{\min}) + \lambda_t\left(D(\boldsymbol{\eta}^{\min}, \mathbf{c}(t)) - \frac{1}{R(t)}\right) \\
    & + \lambda_t\left(\Delta(\boldsymbol{\eta}(t)) - \Delta(\boldsymbol{\eta}^{\min})\right) + \frac{1}{2}\left(D(\boldsymbol{\eta}(t), \mathbf{c}(t)) - \frac{1}{R(t)}\right)^2 + \frac{1}{2}\epsilon^2
\end{align*}

To bound the error difference term $\Delta(\boldsymbol{\eta}(t)) - \Delta(\boldsymbol{\eta}^{\min})$, we analyze the bottleneck mechanics. Because $\boldsymbol{\eta}(t) \ge \boldsymbol{\eta}^{\min}$, we consider the exhaustive set of physically possible bottleneck transitions:
\begin{itemize}
    \item \textbf{Case 1:} If both configurations are compute-bottlenecked, $\Delta(\boldsymbol{\eta}(t)) = 0$ and $\Delta(\boldsymbol{\eta}^{\min}) = 0$, yielding a difference of $0$.
    \item \textbf{Case 2:} If $\boldsymbol{\eta}^{\min}$ is compute-bottlenecked but $\boldsymbol{\eta}(t)$ is communication-bottlenecked, $\Delta(\boldsymbol{\eta}^{\min}) = 0$, making the difference exactly $\Delta(\boldsymbol{\eta}(t))$.
    \item \textbf{Case 3:} If both are communication-bottlenecked, the delay is monotonic with  $\boldsymbol{\eta}$. If the estimator underestimates the delay ($\Delta \ge 0$), then $\Delta(\boldsymbol{\eta}^{\min}) \ge 0$, yielding a difference $\le \Delta(\boldsymbol{\eta}(t))$. If the estimator overestimates the delay ($\Delta \le 0$), the difference is strictly $\le 0$.
\end{itemize}

In all cases, the difference is strictly bounded by the positive error margin: $\Delta(\boldsymbol{\eta}(t)) - \Delta(\boldsymbol{\eta}^{\min}) \le \max\{0, \Delta(\boldsymbol{\eta}(t))\}$. Based on Assumption 3, we bound the conditional expectation of this positive error by $\delta^+$.

Taking the expectation of the full drift inequality, applying Assumption 4 ($\mathbb{E}[D(\boldsymbol{\eta}^{\min}, \mathbf{c}(t)) - \frac{1}{R(t)}] \le -\gamma$), Assumption 3 ($\mathbb{E}[\max\{0, \Delta(\boldsymbol{\eta}(t))\} \mid \lambda_t] \le \delta^+$), and Assumption 2 ($\mathbb{E}[(D(\boldsymbol{\eta}(t), \mathbf{c}(t)) - \frac{1}{R(t)})^2] \le \xi^2$):
\begin{equation*}
    (\gamma - \delta^+)\mathbb{E}[\lambda_t] \le \frac{\mathbb{E}[A(\boldsymbol{\eta}(t))] - A(\boldsymbol{\eta}^{\min})}{\mu} + \frac{1}{2}(\xi^2 + \epsilon^2) + \frac{1}{2}\mathbb{E}[\lambda_t^2] - \frac{1}{2}\mathbb{E}[\lambda_{t+1}^2]
\end{equation*}

Averaging over $T$ slots, applying Assumption 1 ($A(\boldsymbol{\eta}) \le A_{\max}$), and taking the limit $T \to \infty$ establishes the multiplier bound:
\begin{equation}
    \lim_{T \to \infty} \frac{1}{T} \sum_{t=1}^T \mathbb{E}[\lambda_t] \le \frac{A_{\max} - A_{\min}}{\mu(\gamma - \delta^+)} + \frac{\xi^2 + \epsilon^2}{2(\gamma - \delta^+)} \label{eq:multiplier_bound_ap}
\end{equation}

\vspace{1em}
\noindent\textbf{Part C: Bounding the Optimality Gap}\\
We return to the fundamental drift inequality (prior to evaluating $\boldsymbol{\eta}^{\min}$) and instead evaluate it under the strictly optimal, stationary policy $\boldsymbol{\eta}^*(t)$. This policy makes decisions independent of $\lambda_t$, achieves the optimal expected accuracy $A^*$, and satisfies the desired long-term constraint $\mathbb{E}\left[D(\boldsymbol{\eta}^*(t), \mathbf{c}(t)) - \frac{1}{R(t)}\right] \le 0$. Substituting $\boldsymbol{\eta}^*(t)$ and expanding the estimate yields:
\begin{align*}
    & \frac{1}{2}\lambda_{t+1}^2 - \frac{1}{2}\lambda_t^2 - \frac{1}{\mu}A(\boldsymbol{\eta}(t)) - \lambda_t\Delta(\boldsymbol{\eta}(t)) - \frac{1}{2}\left(D(\boldsymbol{\eta}(t), \mathbf{c}(t)) - \frac{1}{R(t)}\right)^2 - \frac{1}{2}\epsilon^2 \\
    & \le -\frac{1}{\mu}A(\boldsymbol{\eta}^*(t)) + \lambda_t\left(\hat{D}(\boldsymbol{\eta}^*(t), \hat{\mathbf{c}}(t)) - \frac{1}{R(t)}\right) \\
    & = -\frac{1}{\mu}A(\boldsymbol{\eta}^*(t)) + \lambda_t\left(D(\boldsymbol{\eta}^*(t), \mathbf{c}(t)) - \frac{1}{R(t)}\right) - \lambda_t\Delta(\boldsymbol{\eta}^*(t))
\end{align*}

Moving $-\lambda_t\Delta(\boldsymbol{\eta}(t))$ to the right-hand side isolates the error difference $\lambda_t(\Delta(\boldsymbol{\eta}(t)) - \Delta(\boldsymbol{\eta}^*(t)))$. Because $\boldsymbol{\eta}^*(t)$ is a stationary policy independent of $\lambda_t$, the expectation of its cross-terms factorizes cleanly. 

First, for the queue term:
\begin{align*}
    \mathbb{E}\left[\lambda_t\left(D(\boldsymbol{\eta}^*(t), \mathbf{c}(t)) - \frac{1}{R(t)}\right)\right] &= \mathbb{E}_{\lambda_t}\left[\lambda_t \mathbb{E}\left[D(\boldsymbol{\eta}^*(t), \mathbf{c}(t)) - \frac{1}{R(t)} \;\Big|\; \lambda_t\right]\right] \\
    &= \mathbb{E}_{\lambda_t}\left[\lambda_t \mathbb{E}\left[D(\boldsymbol{\eta}^*(t), \mathbf{c}(t)) - \frac{1}{R(t)}\right]\right] \\
    &\le \mathbb{E}_{\lambda_t}[\lambda_t \cdot 0] = 0
\end{align*}

Second, for the error difference term, we apply the definitions from Assumption 3. For the dynamic policy, we use the raw expected error bound $\mathbb{E}[\Delta(\boldsymbol{\eta}(t)) \mid \lambda_t] \le \delta$. For the stationary optimal policy, we use $\mathbb{E}[\Delta(\boldsymbol{\eta}^*(t))] \le \delta^*$. Because $\boldsymbol{\eta}^*$ is independent of $\lambda_t$, the expectation factorizes:
\begin{align*}
    \mathbb{E}\left[\lambda_t\left(\Delta(\boldsymbol{\eta}(t)) - \Delta(\boldsymbol{\eta}^*(t))\right)\right] &= \mathbb{E}[\lambda_t \Delta(\boldsymbol{\eta}(t))] - \mathbb{E}[\lambda_t \Delta(\boldsymbol{\eta}^*(t))] \\
    &= \mathbb{E}_{\lambda_t}[\lambda_t \mathbb{E}[\Delta(\boldsymbol{\eta}(t)) \mid \lambda_t]] - \mathbb{E}[\lambda_t]\mathbb{E}[\Delta(\boldsymbol{\eta}^*(t))] \\
    &\le \delta \mathbb{E}[\lambda_t] - \delta^* \mathbb{E}[\lambda_t] = (\delta - \delta^*)\mathbb{E}[\lambda_t]
\end{align*}

Taking the expectation of the full drift inequality and applying these decoupled bounds gives:
\begin{equation*}
    \mathbb{E}[A(\boldsymbol{\eta}(t))] \ge A^* - \mu(\delta - \delta^*)\mathbb{E}[\lambda_t] - \frac{\mu}{2}(\xi^2 + \epsilon^2) + \frac{\mu}{2}\mathbb{E}[\lambda_t^2] - \frac{\mu}{2}\mathbb{E}[\lambda_{t+1}^2]
\end{equation*}

Averaging over $T$ slots:
\begin{equation*}
    \frac{1}{T} \sum_{t=1}^T \mathbb{E}[A(\boldsymbol{\eta}(t))] \ge A^* - \frac{\mu(\delta - \delta^*)}{T} \sum_{t=1}^T \mathbb{E}[\lambda_t] - \frac{\mu}{2}(\xi^2 + \epsilon^2) + \frac{\mu}{2T}\mathbb{E}[\lambda_1^2] - \frac{\mu}{2T}\mathbb{E}[\lambda_{T+1}^2]
\end{equation*}

Since $\mathbb{E}[\lambda_{T+1}^2] \ge 0$ and $\frac{1}{T}\mathbb{E}[\lambda_1^2] \to 0$, taking the limit as $T \to \infty$ gives:
\begin{equation*}
    \lim_{T \to \infty} \frac{1}{T} \sum_{t=1}^T \mathbb{E}[A(\boldsymbol{\eta}(t))] \ge A^* - \frac{\mu(\xi^2 + \epsilon^2)}{2} - \mu(\delta - \delta^*) \left( \lim_{T \to \infty} \frac{1}{T} \sum_{t=1}^T \mathbb{E}[\lambda_t] \right)
\end{equation*}

By substituting the long-term multiplier bound derived in Part B, we establish the final optimality gap.
\end{proof}

\section{Proof of Theorem \ref{thm:multitask_performance_bounds}} \label{proof:thm:multitask_performance_bounds}

\begin{proof}

\textbf{Part A: Proof of Constraint Satisfaction}\\
From the dual update rule, for time slots when task $k$ is active ($t \in \mathcal{T}_k$), the update is:
\begin{equation*}
    \lambda_k(t+1) = \max\left\{\epsilon, \lambda_k(t) + D_{\text{act},k}(t) - \frac{1}{R_k(t)}\right\}
\end{equation*}
For inactive slots ($t \notin \mathcal{T}_k$), $\lambda_k(t+1) = \lambda_k(t)$. Rearranging the active update to isolate the instantaneous constraint violation yields:
\begin{equation*}
    D_{\text{act},k}(t) - \frac{1}{R_k(t)} \le \lambda_k(t+1) - \lambda_k(t)
\end{equation*}
Taking the expectation and summing over all active time slots up to $T$ (the set $\mathcal{T}_k(T)$) produces a telescoping sum for each task $k$:
\begin{equation*}
    \sum_{t \in \mathcal{T}_k(T)} \mathbb{E}\left[D_{\text{act},k}(t) - \frac{1}{R_k(t)}\right] \le \mathbb{E}[\lambda_k(T+1)] - \mathbb{E}[\lambda_k(1)]
\end{equation*}
Dividing by the cardinality $|\mathcal{T}_k(T)|$ yields:
\begin{equation*}
    \frac{1}{|\mathcal{T}_k(T)|} \sum_{t \in \mathcal{T}_k(T)} \mathbb{E}\left[D_{\text{act},k}(t) - \frac{1}{R_k(t)}\right] \le \frac{\mathbb{E}[\lambda_k(T+1)] - \mathbb{E}[\lambda_k(1)]}{|\mathcal{T}_k(T)|}
\end{equation*}
As established in Part B, the expected dual variable $\mathbb{E}[\lambda_k(t)]$ is uniformly bounded by a constant. Assuming the task is active infinitely often, the number of active time slots approaches infinity ($|\mathcal{T}_k(T)| \to \infty$) as the global time horizon $T \to \infty$. Taking the limit with respect to the active time slots yields:
\begin{equation*}
    \lim_{|\mathcal{T}_k(T)| \to \infty} \frac{1}{|\mathcal{T}_k(T)|} \sum_{t \in \mathcal{T}_k(T)} \mathbb{E}\left[D_{\text{act},k}(t) - \frac{1}{R_k(t)}\right] \le 0
\end{equation*}
This proves the long-term expected delay constraint is strictly satisfied specifically for the time slots where the task is active..

\vspace{1em}
\noindent\textbf{Part B: Bounding the Multipliers}\\
Define the aggregate Lyapunov function as $L(\boldsymbol{\lambda}(t)) = \frac{1}{2} \sum_{k=1}^K \lambda_k^2(t)$. Applying the projection property $(\max\{\epsilon, x\})^2 \le x^2 + \epsilon^2$ to the active tasks ($k \in \mathcal{K}(t)$), and noting that the drift is zero for inactive tasks, we bound the unified single-slot aggregate expected drift $\Delta L(t) \triangleq \mathbb{E}[L(\boldsymbol{\lambda}(t+1))] - \mathbb{E}[L(\boldsymbol{\lambda}(t))]$:
\begin{equation*}
    \Delta L(t) \le \mathbb{E}\left[ \sum_{k \in \mathcal{K}(t)} \left( \lambda_k(t)\left(D_{\text{act},k}(t) - \frac{1}{R_k(t)}\right) + \frac{1}{2}\left(D_{\text{act},k}(t) - \frac{1}{R_k(t)}\right)^2 + \frac{1}{2}\epsilon^2 \right) \right]
\end{equation*}

To introduce the utility-delay trade-off, we subtract the aggregate scaled expected utility $\frac{1}{\mu} \sum_{k \in \mathcal{K}(t)} w_k \mathbb{E}[A_k(\boldsymbol{\eta}_k(t))]$ from both sides. We further add and subtract the estimator term $\lambda_k(t) \hat{D}_k(t)$ inside the expectation:
\begin{align*}
    \Delta L(t) - \frac{1}{\mu}\sum_{k \in \mathcal{K}(t)} w_k \mathbb{E}[A_k(\boldsymbol{\eta}_k(t))] \le & \mathbb{E}\Bigg[ \sum_{k \in \mathcal{K}(t)} \left( -\frac{1}{\mu} w_k A_k(\boldsymbol{\eta}_k(t)) + \lambda_k(t)\left(\hat{D}_k(t) - \frac{1}{R_k(t)}\right) \right) \\
    & + \sum_{k \in \mathcal{K}(t)} \lambda_k(t) \Delta_k(t) + \sum_{k \in \mathcal{K}(t)} \left( \frac{1}{2}\left(D_{\text{act},k}(t) - \frac{1}{R_k(t)}\right)^2 + \frac{1}{2}\epsilon^2 \right) \Bigg]
\end{align*}

Based on Assumption 5, the block coordinate descent algorithm maximizes the partial Lagrangian within an expected additive gap $C_{\text{gap}}$. We upper-bound the primary algorithmic terms by evaluating them against a strictly feasible, robust baseline configuration $(\boldsymbol{\eta}^{\min}, \mathbf{s}^{\min})$. Let $D_k^{\min}(t) \triangleq D_k(\boldsymbol{\eta}^{\min}_k, \mathbf{s}^{\min}_k, \mathbf{c}(t))$ and $\hat{D}_k^{\min}(t)$ be its estimate. Substituting this bound back into the drift inequality, and replacing $\hat{D}_k^{\min}(t)$ with $D_k^{\min}(t) - \Delta_k^{\min}(t)$, we isolate the error difference term:
\begin{align*}
    \Delta L(t) - \frac{1}{\mu}\sum_{k \in \mathcal{K}(t)} w_k \mathbb{E}[A_k(\boldsymbol{\eta}_k(t))] \le & \mathbb{E}\Bigg[ \sum_{k \in \mathcal{K}(t)} \left( -\frac{1}{\mu} w_k A_k(\boldsymbol{\eta}^{\min}) + \lambda_k(t)\left(D_k^{\min}(t) - \frac{1}{R_k(t)}\right) \right) \\
    & + \sum_{k \in \mathcal{K}(t)} \lambda_k(t)\left(\Delta_k(t) - \Delta_k^{\min}(t)\right) \\
    & + \sum_{k \in \mathcal{K}(t)} \left( \frac{1}{2}\left(D_{\text{act},k}(t) - \frac{1}{R_k(t)}\right)^2 + \frac{1}{2}\epsilon^2 \right) \Bigg] + \frac{C_{\text{gap}}}{\mu}
\end{align*}

The bottleneck mechanics dictate that the error difference is strictly bounded by the positive error margin: $\Delta_k(t) - \Delta_k^{\min}(t) \le \max\{0, \Delta_k(t)\}$. By conditioning on the state variables and applying Assumption 3 ($\mathbb{E}[\max\{0, \Delta_k(t)\} \mid \lambda_k(t)] \le \delta^+$) and Assumption 4 ($\mathbb{E}[D_k^{\min}(t) - \frac{1}{R_k(t)}] \le -\gamma$), we can bound the terms:
\begin{align*}
    \mathbb{E}\left[\lambda_k(t)\left(\Delta_k(t) - \Delta_k^{\min}(t)\right)\right] &\le \mathbb{E}_{\lambda_k}\Big[\lambda_k(t) \mathbb{E}[\max\{0, \Delta_k(t)\} \mid \lambda_k(t)]\Big] \le \delta^+ \mathbb{E}[\lambda_k(t)] \\
    \mathbb{E}\left[\lambda_k(t) \left(D_k^{\min}(t) - \frac{1}{R_k(t)}\right)\right] &= \mathbb{E}_{\lambda_k}\Big[\lambda_k(t) \mathbb{E}\left[D_k^{\min}(t) - \frac{1}{R_k(t)}\right]\Big] \le -\gamma \mathbb{E}[\lambda_k(t)]
\end{align*}

Applying these bounds along with Assumption 2 ($\mathbb{E}[(D_{\text{act},k}(t) - \frac{1}{R_k(t)})^2] \le \xi^2$), and rearranging yields:
\begin{equation*}
    (\gamma - \delta^+) \sum_{k \in \mathcal{K}(t)} \mathbb{E}[\lambda_k(t)] \le \sum_{k \in \mathcal{K}(t)} \frac{w_k}{\mu} \left(\mathbb{E}[A_k(\boldsymbol{\eta}_k(t))] - A_k(\boldsymbol{\eta}^{\min})\right) + \sum_{k \in \mathcal{K}(t)} \frac{1}{2}(\xi^2 + \epsilon^2) + \frac{C_{\text{gap}}}{\mu} - \Delta L(t)
\end{equation*}

Summing over $T$ slots, dropping the positive term $\mathbb{E}[L(\boldsymbol{\lambda}(T+1))]$, and dividing by $T$. Note that $\frac{1}{T} \sum_{t=1}^T \sum_{k \in \mathcal{K}(t)} (\cdot) = \sum_{k=1}^K \frac{|\mathcal{T}_k(T)|}{T} (\cdot)$. Defining the asymptotic frequency of task $k$ being active as $\rho_k = \lim_{T \to \infty} \frac{|\mathcal{T}_k(T)|}{T}$, we take the limit $T \to \infty$ to establish the aggregate multiplier bound:
\begin{equation}
    \lim_{T \to \infty} \frac{1}{T} \sum_{t=1}^T \sum_{k \in \mathcal{K}(t)} \mathbb{E}[\lambda_k(t)] \le \frac{\sum_{k=1}^K \rho_k w_k (A_{k,\max} - A_{k,\min}) + C_{\text{gap}}}{\mu(\gamma - \delta^+)} + \frac{\sum_{k=1}^K \rho_k (\xi^2 + \epsilon^2)}{2(\gamma - \delta^+)} \label{eq:multitask_multiplier_bound_ap}
\end{equation}

\vspace{1em}
\noindent\textbf{Part C: Bounding the Optimality Gap}\\
We return to the fundamental drift inequality prior to evaluating the robust baseline. Instead, we evaluate it under the optimal stationary multi-task policy $(\boldsymbol{\eta}^*(t), \mathbf{s}^*(t))$, which satisfies $\mathbb{E}\left[D_k^*(t) - \frac{1}{R_k(t)}\right] \le 0$, where $D_k^*(t) \triangleq D_k(\boldsymbol{\eta}_k^*(t), \mathbf{s}_k^*(t), \mathbf{c}(t))$. Incorporating the approximation gap $C_{\text{gap}}$, we obtain:
\begin{align*}
    \Delta L(t) - \frac{1}{\mu}\sum_{k \in \mathcal{K}(t)} w_k \mathbb{E}[A_k(\boldsymbol{\eta}_k(t))] \le & \mathbb{E}\Bigg[ \sum_{k \in \mathcal{K}(t)} \left( -\frac{1}{\mu} w_k A_k^*(t) + \lambda_k(t)\left(D_k^*(t) - \frac{1}{R_k(t)}\right) \right) \\
    & + \sum_{k \in \mathcal{K}(t)} \lambda_k(t)\left(\Delta_k(t) - \Delta_k^*(t)\right) \\
    & + \sum_{k \in \mathcal{K}(t)} \left( \frac{1}{2}\left(D_{\text{act},k}(t) - \frac{1}{R_k(t)}\right)^2 + \frac{1}{2}\epsilon^2 \right) \Bigg] + \frac{C_{\text{gap}}}{\mu}
\end{align*}

Because the optimal stationary policy makes decisions independent of the queue state, the expectation of the queue cross-term cleanly factorizes by conditioning on the state variables:
\begin{equation*}
    \mathbb{E}\left[\lambda_k(t) \left(D_k^*(t) - \frac{1}{R_k(t)}\right)\right] = \mathbb{E}_{\lambda_k}\Big[\lambda_k(t) \mathbb{E}\left[D_k^*(t) - \frac{1}{R_k(t)}\right]\Big] \le \mathbb{E}[\lambda_k(t)] \cdot 0 = 0
\end{equation*}

For the error difference term, we apply the raw expected error definitions from Assumption 3 ($\mathbb{E}[\Delta_k(t) \mid \lambda_k(t)] \le \delta$ and $\mathbb{E}[\Delta_k^*(t)] = \delta^*$):
\begin{align*}
    \mathbb{E}\left[\lambda_k(t)\left(\Delta_k(t) - \Delta_k^*(t)\right)\right] &= \mathbb{E}[\lambda_k(t) \Delta_k(t)] - \mathbb{E}[\lambda_k(t) \Delta_k^*(t)] \\
    &= \mathbb{E}_{\lambda_k}\Big[\lambda_k(t)\mathbb{E}[\Delta_k(t) \mid \lambda_k(t)]\Big] - \mathbb{E}_{\lambda_k}\Big[\lambda_k(t)\mathbb{E}[\Delta_k^*(t)]\Big] \\
    &\le \delta \mathbb{E}[\lambda_k(t)] - \delta^* \mathbb{E}[\lambda_k(t)] = (\delta - \delta^*)\mathbb{E}[\lambda_k(t)]
\end{align*}

Taking the full expectation of the drift inequality, applying these decoupled bounds:
\begin{equation*}
    \frac{1}{\mu} \sum_{k \in \mathcal{K}(t)} w_k \mathbb{E}[A_k(\boldsymbol{\eta}_k(t))] \ge \frac{1}{\mu} \sum_{k \in \mathcal{K}(t)} w_k A_k^* - \sum_{k \in \mathcal{K}(t)} (\delta - \delta^*) \mathbb{E}[\lambda_k(t)] - \sum_{k \in \mathcal{K}(t)} \frac{1}{2}(\xi^2 + \epsilon^2) - \frac{C_{\text{gap}}}{\mu} + \Delta L(t)
\end{equation*}

Averaging the above inequality over $T$ slots, dropping the stabilizing Lyapunov term $\mathbb{E}[L(\boldsymbol{\lambda}(1))]/T$, multiplying by $\mu$, and taking the limit $T \to \infty$. We define $A^* = \sum_{k=1}^K \rho_k w_k A_k^*$ as the optimal time-averaged expected baseline utility:
\begin{equation*}
    \lim_{T \to \infty} \frac{1}{T} \sum_{t=1}^T \sum_{k \in \mathcal{K}(t)} w_k \mathbb{E}[A_k(\boldsymbol{\eta}_k(t))] \ge A^* - \mu (\delta - \delta^*) \left( \lim_{T \to \infty} \frac{1}{T} \sum_{t=1}^T \sum_{k \in \mathcal{K}(t)} \mathbb{E}[\lambda_k(t)] \right) - \frac{\mu}{2} \sum_{k=1}^K \rho_k (\xi^2 + \epsilon^2) - C_{\text{gap}}
\end{equation*}

\end{proof}

\section{Principled Selection of the Tradeoff Parameter $\mu$}
\label{sec:mu_selection}

The operational performance of the \textit{Estimated Stochastic Dual Descent} algorithm is inherently dictated by the tradeoff parameter $\mu$. This parameter serves as the pivotal weight that balances the maximization of the application-level utility, $A(\boldsymbol{\eta})$, against the strict enforcement of the long-term Quality of Service (QoS) delay constraints, represented by the virtual queue of dual variables $\boldsymbol{\lambda}$. Based on the Lyapunov stability results established in Theorem 3.5 and Theorem 3.6, the time-averaged expected delay violation is governed by the fundamental $[\mathcal{O}(\mu), \mathcal{O}(1/\mu)]$ tradeoff and is upper-bounded as follows:
\begin{equation}
    \lim_{T \to \infty} \frac{1}{T} \sum_{t=1}^{T} \mathbb{E}[\lambda_t] \le \frac{A_{\max} - A_{\min}}{\mu(\gamma - \delta^+)} + \Psi,
\end{equation}
where $\Psi$ captures the constant noise terms. 

To guarantee system stability, our objective is to ensure that the expected accumulated delay violations remain bounded by a reasonable constant order of magnitude $\Theta(1)$, rather than growing infinitely or being unnecessarily suppressed to zero at the severe expense of inference accuracy. Achieving this target requires a principled derivation for $\mu$. This derivation relies on two critical system parameters: $\gamma$, the expected robust safety margin between the target QoS cycle time $1/R$ and the baseline bottleneck delay evaluated at maximum compression $\boldsymbol{\eta}^{\min}$; and $\delta$, the expected absolute impact of the channel estimation error on the realized bottleneck delay. 

Because these parameters govern different mechanics of the system, they must be determined through distinct methodologies:
\begin{itemize}
   \item \textbf{Analytical Calculation of $\gamma$:} Because the feasibility margin is defined strictly at a static, maximum compression configuration ($\boldsymbol{\eta}^{\min}$), it acts as an open-loop physical property of the hardware and network. If the base processing times, activation sizes, and the underlying channel capacity distributions are known, $\gamma$ can be calculated directly. Furthermore, in multi-task scenarios where resource contention complicates exact calculation, a conservative lower bound for $\gamma$ can be derived analytically by assuming a baseline equal resource sharing policy across all active tasks and subsequently evaluating the expected bottleneck delays at $\boldsymbol{\eta}^{\min}$.
    \item \textbf{Empirical Measurement of $\delta$:} Unlike $\gamma$, the expected delay error $\delta$ depends on the dynamic compression factor $\boldsymbol{\eta}(t)$. Because $\boldsymbol{\eta}(t)$ is actively selected by the algorithm in a closed-loop manner based on the capacity estimator $\hat{c}(t)$, the error impact is deeply coupled to both the selection policy and the estimator's behavior. Therefore, $\delta$ must be measured empirically. This is achieved during a warmup phase or simulation by executing the policy, logging the true channel realizations $c(t)$ alongside the causal predictions $\hat{c}(t)$, and averaging the resulting delay discrepancies.
\end{itemize}

A critical component of our experimental setup is the deployment of a 90\% Lower Confidence Bound (LCB) for causal channel estimation. By adopting a deliberately pessimistic view of the network state, the estimator systematically under-predicts the available channel capacity. Consequently, the estimated bottleneck delay safely overestimates the actual realized delay. In the context of our Lyapunov drift analysis, this pessimistic over-estimation drives the expected error impact $\delta$ to be strictly negative ($\delta^+ < 0$). Mathematically, this broadens the effective stability denominator from $(\gamma - \delta^+)$ to $(\gamma + |\delta^+|)$. This augmented, artificial buffer ensures that the virtual queues remain stable and strictly satisfy the QoS constraints, even when the system is subjected to high-variance network environments.

Leveraging this expanded stability margin, we define a heuristic upper bound for $\mu$ to ensure the delay violations remain firmly within the $\Theta(1)$ budget:
\begin{equation}
    U = \frac{A_{\max} - A_{\min}}{\gamma}.
\end{equation}

To systematically isolate the optimal operating point on the Pareto frontier—defined as the minimum $\mu$ that strictly satisfies the QoS delay constraints while maximizing accuracy—we employ an empirical binary search strategy over the interval $(0, U]$. To minimize profiling overhead, we conduct a multiple-step binary search initialized at the midpoint:
\begin{equation}
    \mu = \frac{U}{2}.
\end{equation}
Following each trial, the system evaluates the resulting time-averaged delay compliance. The search direction for the subsequent iteration is dictated by the system's stability:
\begin{itemize}
    \item \textbf{Delay Satisfied (Over-penalized):} If the constraints are strictly met, the system can safely prioritize inference accuracy further. We update the upper bound of the search space and test a smaller penalty (e.g., transitioning from $U/2$ to $U/4$).
    \item \textbf{Delay Violated (Under-penalized):} If the delay constraints are breached, the dual variables lack sufficient weight to stabilize the virtual queues. We update the lower bound of the search space and test a larger penalty (e.g., transitioning from $U/2$ to $3U/4$).
\end{itemize}
Based on our experiments, executing this binary search for three iterations rapidly converges on a highly tuned $\mu$.
In Fig. \ref{fig:mu-sweep}, we observe the sensitivity of the result to the value of $\mu$ in an online Jetson experiment (\scenVisionJetsonTestbedST) for three different values of $\mu$.
We observe that in the left figure, related to approximating the gradient using a fitting model approach, decreasing the tradeoff parameter $\mu$ from $1$ to $0.01$ results in an increase in aggregate utility (from approximately $0.81$ to $0.90$). However, this improvement comes at the direct cost of QoS compliance. As $\mu$ decreases, the system prioritizes utility over the virtual queues, causing the excess delay to cross the feasibility threshold (the zero vertical line) from a safe negative value into the positive domain, indicating delay violations. For the fitting model, only $\mu=1$ successfully maintains a strictly negative excess delay. 
In contrast, the right panel illustrates the performance when employing the Stein's estimator. While the overarching $\mu$-driven tradeoff curve remains structurally similar, the specific operating points shift. Notably, $\mu=0.1$ emerges as a highly optimal configuration for the Stein's estimator; it achieves a strong utility of approximately $0.85$ while still maintaining a strictly negative excess delay (around $-10^{-1}$). Further decreasing the parameter to $\mu=0.01$ yields marginal utility gains (up to roughly $0.92$) but severely breaches the delay constraints, resulting in a large positive excess delay.

Ultimately, both panels empirically validate the $[\mathcal{O}(\mu), \mathcal{O}(1/\mu)]$ bounds derived in our Lyapunov analysis. Sweeping $\mu$ traces out a distinct Pareto frontier for the \algOURSNOCSIst[short] approach, clearly bridging the gap between conservative, low-utility baseline executions and aggressive, delay-violating performance. This visualizes and underscores the necessity of our empirical binary search methodology, as the exact optimal $\mu$ required to maximize utility while safely adhering to the zero-excess-delay constraint heavily depends on the specific gradient estimation strategy deployed.

\begin{figure*}[t]
    \centering
    \includegraphics[width=0.75\textwidth]{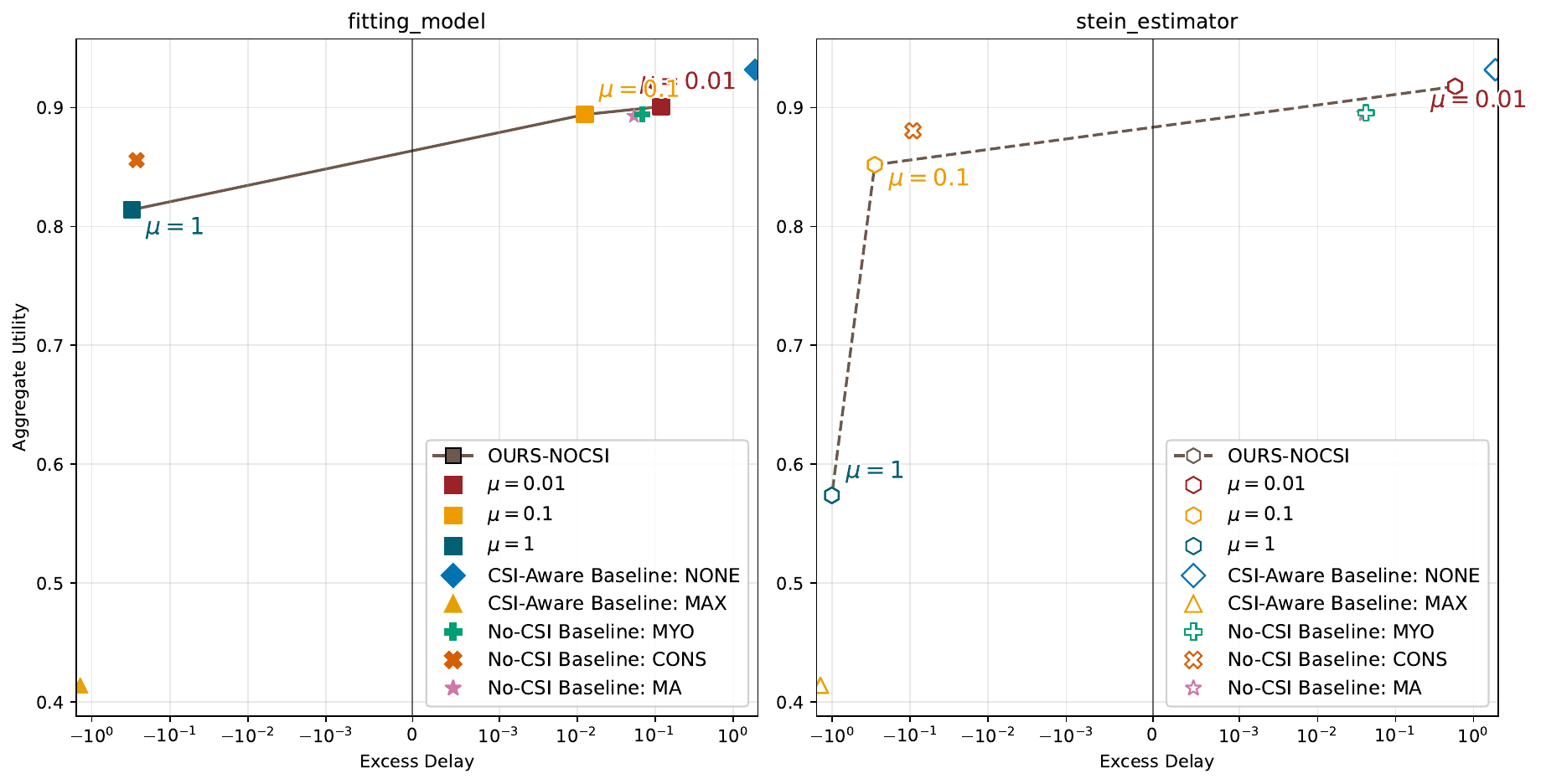}
    \caption{Testbed Top-$k$ single-task results under $\mu$ sweep. The left panel uses the fitting-based estimator, and the right panel uses Stein's estimator.}
    \label{fig:mu-sweep}
\end{figure*}





\section{Estimating the Accuracy-Compression Function}\label{app:acccomp}
\CRedit{This appendix describes how we estimate $A_k(\boldsymbol{\eta}_k)$ from held-out data and how we validate the concavity assumed in Sec.~\ref{sec:meth}. We first recall the Stein gradient oracle, then characterize the shape of $A$ and detect $\eta_{\min}$, statistically test concavity across tasks and compression schemes, and finally describe regression fitting.}

\subsection{Stein Gradient Oracle}
\label{app:stein}

We describe the zeroth-order gradient oracle used in our method, which estimates $\nabla f(\mathbf{x})$ using only scalar evaluations of $f$. The construction is grounded in Stein's lemma.

\subsubsection{Stein's Lemma}
We repeat here Stein's Lemma, which we can use to produce a gradient estimator for the accuracy function via appropriate sampling:
\begin{lemma}[Stein, 1981~\cite{stein1981estimation}] 
Let $\mathbf{z} \sim \mathcal{N}(0, I_d)$ and let $f : \mathbb{R}^d \to \mathbb{R}$ be almost everywhere differentiable with $\mathbb{E}[\|\nabla f(\mathbf{x} + \sigma\mathbf{z})\|] < \infty$ and $\mathbb{E}[f(\mathbf{x} + \sigma\mathbf{z})^2] < \infty$. Then for any $\mathbf{x} \in \mathbb{R}^d$ and $\sigma > 0$:
\begin{equation}
    \mathbb{E}_{\mathbf{z}}\left[\mathbf{z}\, f(\mathbf{x} + \sigma\mathbf{z})\right] = \sigma\, \mathbb{E}_{\mathbf{z}}\left[\nabla f(\mathbf{x} + \sigma\mathbf{z})\right].
    \label{eq:steins_lemma}
\end{equation}
\end{lemma}

The integrability conditions ensure that the boundary term arising in the integration by parts proof vanishes, i.e.\ $\lim_{\|\mathbf{z}\|\to\infty} f(\mathbf{x}+\sigma\mathbf{z})\,\phi(\mathbf{z}) = 0$, where $\phi$ is the standard Gaussian density. These conditions are satisfied whenever $f$ grows slower than exponentially, which covers all cases of practical interest.

\subsubsection{Gradient Estimator of $A(\boldsymbol{\eta})$ using Stein's Lemma}

Equation~\eqref{eq:steins_lemma} implies that the gradient of the Gaussian-smoothed function,
\begin{equation}
    f_\sigma(\mathbf{x}) = \mathbb{E}_{\mathbf{z} \sim \mathcal{N}(0, I_d)}\left[f(\mathbf{x} + \sigma\mathbf{z})\right],
\end{equation}
can be written purely in terms of function evaluations:
\begin{equation}
    \nabla f_\sigma(\mathbf{x}) = \frac{1}{\sigma}\,\mathbb{E}_{\mathbf{z}}\left[\mathbf{z}\, f(\mathbf{x} + \sigma\mathbf{z})\right].
\end{equation}
As $\sigma \to 0$, $f_\sigma \to f$ pointwise, so for small $\sigma$ this approximates $\nabla f(\mathbf{x})$. In practice we estimate the expectation via $N$ antithetic sample pairs $(\mathbf{z}_i, -\mathbf{z}_i)$, which cancels the zeroth-order term $f(\mathbf{x})$ and reduces variance:
\begin{equation}
    \hat{\nabla} f(\mathbf{x}) = \frac{1}{2N\sigma} \sum_{i=1}^{N} \mathbf{z}_i \left[f(\mathbf{x} + \sigma\mathbf{z}_i) - f(\mathbf{x} - \sigma\mathbf{z}_i)\right], \quad \mathbf{z}_i \overset{\text{i.i.d.}}{\sim} \mathcal{N}(0, I_d).
    \label{eq:stein_estimator}
\end{equation}
The estimator requires $2N$ function evaluations regardless of the input dimension $d$, and all queries are independent so they can be executed in parallel. The smoothing parameter $\sigma$ controls a bias-variance trade-off: larger $\sigma$ reduces variance but increases the bias from smoothing, while smaller $\sigma$ reduces bias at the cost of higher variance.

\begin{algorithm}[t]
\caption{Stein Gradient Oracle}
\label{alg:stein}
\small 
\begin{algorithmic}[1]
\Require Point $\mathbf{x} \in \mathbb{R}^d$, smoothing parameter $\sigma > 0$, number of samples $N$
\Ensure Gradient estimate $\hat{\nabla} f(\mathbf{x}) \in \mathbb{R}^d$
\State $\mathbf{g} \leftarrow \mathbf{0}_d$
\For{$i = 1$ \textbf{to} $N$}
    \State Sample $\mathbf{z}_i \sim \mathcal{N}(0, I_d)$
    \State $\mathbf{g} \leftarrow \mathbf{g} + \mathbf{z}_i \bigl[ f(\mathbf{x} + \sigma \mathbf{z}_i) - f(\mathbf{x} - \sigma \mathbf{z}_i) \bigr]$
\EndFor
\State \Return $\hat{\nabla} f(\mathbf{x}) = \mathbf{g}\, /\, (2N\sigma)$
\end{algorithmic}
\end{algorithm}

\subsection{Accuracy Function Analysis: Concavity and $\eta_{\min}$ Detection}
\label{app:concavity}

We next describe the methodology used to (i) characterize the shape of
the accuracy function $A(\boldsymbol{\eta})$ across the compression-ratio
space, (ii) automatically detect the per-link threshold $\eta_{\min}^{\ell}$
below which the activation representations collapse to noise and the predictor
should not be fitted, (iii) produce three-dimensional
surface plots of $A(\boldsymbol{\eta})$\CRedit{, and (iv) statistically validate concavity
across tasks and compression schemes.}
Let $\boldsymbol{\eta} = (\eta_0, \eta_1, \ldots, \eta_{L-1}) \in [0,1]^L$
denote the per-link compression ratios \CRedit{of a task with $L_k$ stages, so $L=L_k-1$}, where $\eta_\ell = 1$ corresponds to no compression on link $\ell$ and $\eta_\ell = 0$ corresponds to maximum
compression.  The true accuracy function $A_{\mathrm{true}}(\boldsymbol{\eta})$
is defined as the task accuracy (e.g.\ top-1 classification accuracy for
ResNet, or MMLU accuracy for the LLM backend) when each link $\ell$ operates
at ratio $\eta_\ell$.
Empirically, $A_{\mathrm{true}}$ exhibits two qualitatively distinct
regimes (see Figure~\ref{fig:accuracy_surface_knn}):

\begin{enumerate}
  \item \textbf{Collapse regime} ($\eta_\ell < \eta_{\min}^{\ell}$): the
    compression is so severe that intermediate activations carry no useful
    information.  $A$ is essentially flat and close to chance level; small
    increases in $\eta_\ell$ produce no consistent improvement.

  \item \textbf{Informative regime} ($\eta_\ell \geq \eta_{\min}^{\ell}$):
    $A$ is monotonically non-decreasing and \emph{concave} in $\eta_\ell$.  Our surrogate
    predictor is only well-specified in this regime.
\end{enumerate}
Operating below $\eta_{\min}^{\ell}$ wastes compute (all configurations are
equally bad) and can mislead the optimizer.  We therefore identify
$\eta_{\min}^{\ell}$ automatically before training the predictor and enforce
it as a hard lower bound on the search domain.

\subsubsection{$\eta_{\min}$ Detection}
\label{app:concavity:tests}

We run three complementary analyses, each implemented as 1-D sweeps along
individual link axes or as multi-dimensional directional profiles. In each case, we used a knn regressor to smoothly approximate the true accuracy function. The knn regressor was trained on $N^2$ samples of the accuracy function, each averaged over $R$ evaluations to reduce noise.

\paragraph{Test 1: Per-axis Monotonicity Scan}
\label{app:concavity:mono}

For each link $\ell$, we construct a 1-D slice of the accuracy function by
sweeping $\eta_\ell$ over a uniform grid of $N$ points in $[0,1]$ while
holding all other links fixed at a reference value $\eta_{\mathrm{mid}} =
0.5$. 
Let $p_j = A(\eta^{(\ell)}_j)$ denote the accuracy at grid
point $j$.  We apply a light three-point moving average $\tilde{p}_j =
\frac{1}{3}(p_{j-1}+p_j+p_{j+1})$ and compute finite-difference slopes
\begin{equation}
  s_j = \frac{\tilde{p}_{j+1} - \tilde{p}_j}{\eta^{(\ell)}_{j+1} -
  \eta^{(\ell)}_j}.
\end{equation}
The per-link monotonicity threshold is defined as the smallest grid value
above which the slope remains consistently positive:
\begin{equation}
  \hat{\eta}_{\min,\mathrm{mono}}^{\ell}
  \;=\;
  \eta^{(\ell)}_{j^* - 1},
  \qquad
  j^* = \min\bigl\{j : s_j \geq \tau_{\mathrm{slope}}\bigr\},
\end{equation}
where $\tau_{\mathrm{slope}} > 0$ is a small noise tolerance
(default $5 \times 10^{-3}$).

\paragraph{Test 2: Concavity Scan via Second Differences}
\label{app:concavity:conc}

Using the same 1-D profiles $p_j$, we compute second (centered) differences
\begin{equation}
  \Delta^2_j = p_{j-1} + p_{j+1} - 2\,p_j, \qquad j = 1, \ldots, N-2.
\end{equation}
Strict concavity of $A$ requires $\Delta^2_j \leq 0$ for all $j$.  A
significantly positive second difference indicates a convex neighborhood in the
accuracy curve.
We scan downward from $\eta=1$ and identify the last (highest) index at which
concavity is violated:
\begin{equation}
  \hat{\eta}_{\min,\mathrm{conc}}^{\ell}
  \;=\;
  \eta^{(\ell)}_{j^{**}},
  \qquad
  j^{**} = \max\bigl\{j : \Delta^2_j > \tau_{\mathrm{conc}}\bigr\},
\end{equation}
where $\tau_{\mathrm{conc}} > 0$ is a small positive tolerance to allow for
measurement noise (default $5 \times 10^{-3}$).  If no violation is detected,
$\hat{\eta}_{\min,\mathrm{conc}}^{\ell} = 0$.  Plots of $\Delta^2_j$ are
shown in Figure~\ref{fig:second_differences}.

\paragraph{Test 3: Directional Profiles}
\label{app:concavity:dir}

We repeat the above process across different directions. We sample $D$ random rays in $\boldsymbol{\eta}$-space.  Each
ray is parameterized as
\begin{equation}
  \boldsymbol{\eta}(s) = \boldsymbol{\eta}^{(d)}_{\mathrm{start}}
  + s\,\bigl(\mathbf{1} - \boldsymbol{\eta}^{(d)}_{\mathrm{start}}\bigr),
  \quad s \in [0,1],
\end{equation}
so that $\boldsymbol{\eta}(0) = \boldsymbol{\eta}^{(d)}_{\mathrm{start}}$
and $\boldsymbol{\eta}(1) = \mathbf{1}$.  The first ray always starts from
$\boldsymbol{0}$; the remaining $D-1$ start from points sampled uniformly in
$[0, 0.5]^L$.  Accuracy is evaluated along each ray at $N$ equispaced values
of $s$, and the profiles are plotted together in
Figure~\ref{fig:directional_profiles}. A vertical marker is drawn on each profile at the ray parameter $s^*$
corresponding to the point where the ray enters the informative regime.

The final per-link threshold combines the two quantitative tests with an
additive safety margin $\delta \geq 0$ (default $0.05$):
\begin{equation}
  \hat{\eta}^{\ell}_{\min,\mathrm{final}}
  \;=\;
  \min\!\left(1,\;
    \max\!\bigl(
      \hat{\eta}_{\min,\mathrm{mono}}^{\ell},\;
      \hat{\eta}_{\min,\mathrm{conc}}^{\ell}
    \bigr)
    + \delta
  \right).
\end{equation}
The global threshold used as the lower bound for all subsequent experiments is
\begin{equation}
  \eta_{\min} = \max_{\ell}\; \hat{\eta}^{\ell}_{\min,\mathrm{final}}.
  \label{eq:etamin_global}
\end{equation}
Axis-sweep profiles with all three threshold lines overlaid are shown in
Figure~\ref{fig:axis_profiles}.

\subsection{Three-Dimensional Accuracy Surface Visualisation}
\label{app:concavity:3d}

To justify the use of convex optimization algorithms, we provide a picture of $A(\boldsymbol{\eta})$ that demonstrates its concave region. To do that, we use the ResNet-56 model with \CRedit{$L_k=4$ stages ($L=3$ hops)}, fix link 0 at
$\eta_0 = 1$ (uncompressed) and produce 3-D surface plots over the
$(\eta_1, \eta_2)$ plane on an $N \times N$ grid (total $N^2$ evaluations,
$R$ repeats each).  Two surfaces are rendered: a) the KNN-smoothed surface, where
A $k$-nearest-neighbor regressor with inverse-distance weighting is fitted
to the $N^2$ observations, yielding a smooth
surface that removes per-point noise while preserving the global shape
(Figure~\ref{fig:accuracy_surface_knn}); and b) the Stein-smoothed surface, where
The Stein (Gaussian-smoothing) view of $A$,
\begin{equation}
  \tilde{A}(\boldsymbol{\eta})
  = \mathbb{E}_{\mathbf{z} \sim \mathcal{N}(\mathbf{0}, \mathbf{I})}
    \bigl[A(\boldsymbol{\eta} + \sigma\,\mathbf{z})\bigr],
  \label{eq:stein_smooth}
\end{equation}
estimated by Monte Carlo with $M$ perturbation samples per grid point and
bandwidth $\sigma$ matching the value used in the gradient oracle
(Figure~\ref{fig:accuracy_surface_stein}).  This surface is what the
optimizer effectively maximizes during training. Table~\ref{tab:concavity_params} summarizes the default hyperparameters used
in our experiments.  

\begin{table}[t!]
\centering
\caption{Default hyperparameters for concavity analysis and 3-D surface plots.}
\label{tab:concavity_params}
\begin{tabular}{llcc}
\toprule
Parameter & Description & ResNet & LLM \\
\midrule
$N$ & Grid points per axis (1-D sweeps) & 20 & 12 \\
$R$ & Repeated evaluations per point & 5 & 1 \\
$D$ & Random rays (directional test) & 10 & 5 \\
$\tau_{\mathrm{slope}}$ & Monotonicity noise tolerance & \multicolumn{2}{c}{$5 \times 10^{-3}$} \\
$\tau_{\mathrm{conc}}$  & Concavity noise tolerance   & \multicolumn{2}{c}{$5 \times 10^{-3}$} \\
$\delta$ & Safety margin & \multicolumn{2}{c}{$0.05$} \\
\midrule
$N^2$ & Grid points (3-D surface)  & 400 & 400 \\
$R$   & Repeats per point (3-D)    & 3   & 3   \\
$k$   & KNN neighbors             & \multicolumn{2}{c}{8} \\
$\sigma$ & Stein bandwidth         & \multicolumn{2}{c}{$0.05$} \\
$M$   & Stein MC samples per point & \multicolumn{2}{c}{10} \\
\bottomrule
\end{tabular}
\end{table}

\label{app:concavity:figures}

\begin{figure}[t!]
  \centering
  \includegraphics[width=\linewidth]{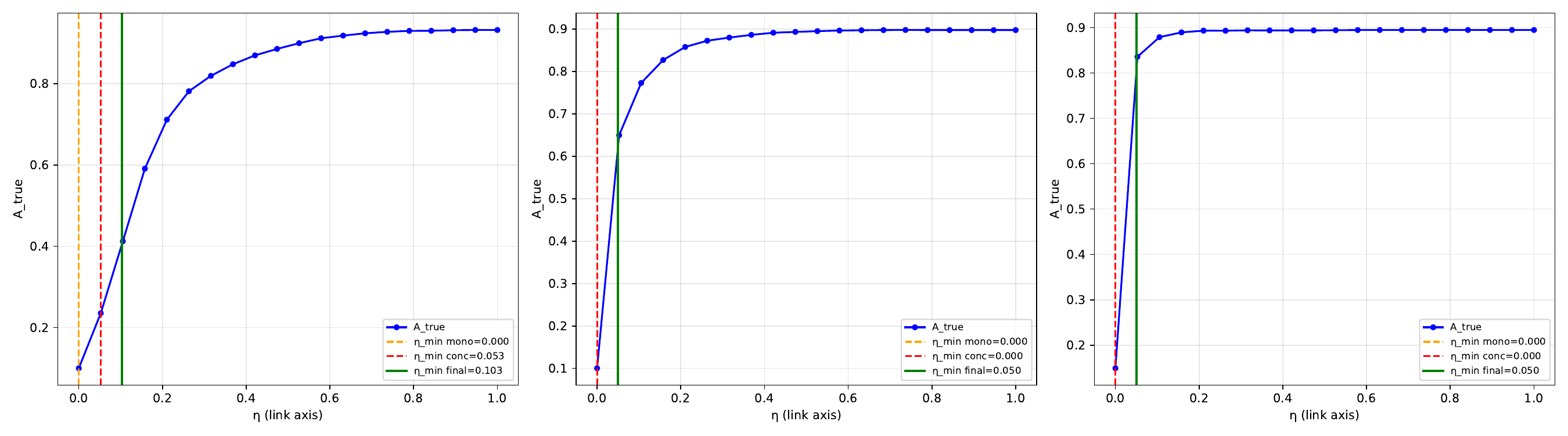}
  \caption{%
    \textbf{Per-axis accuracy profiles with $\eta_{\min}$ thresholds.}
    Each panel shows the 1-D accuracy slice along one link axis (all other
    links held at $\eta_{\mathrm{mid}}=0.5$).  Orange dashed: monotonicity
    threshold $\hat{\eta}_{\min,\mathrm{mono}}^{\ell}$; red dashed: concavity
    threshold $\hat{\eta}_{\min,\mathrm{conc}}^{\ell}$; green solid: final
    recommended threshold $\hat{\eta}^{\ell}_{\min,\mathrm{final}}$
    (equation~\eqref{eq:etamin_global}).
  }
  \label{fig:axis_profiles}
\end{figure}

\begin{figure}[t!]
  \centering
  \includegraphics[width=\linewidth]{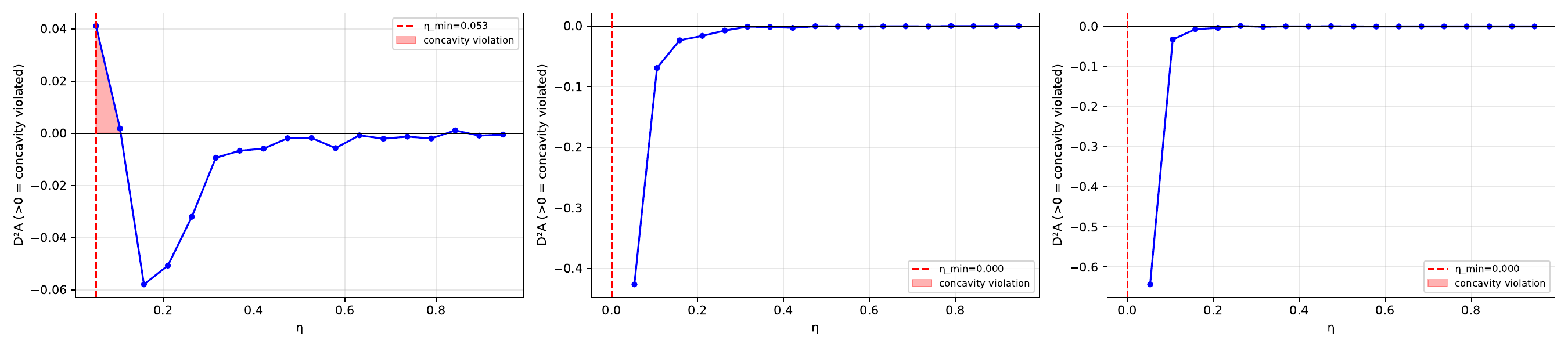}
  \caption{%
    \textbf{Second differences $\Delta^2 A$ along each link axis.}
    Positive values (shaded red) indicate violations of concavity.  The red
    dashed vertical line marks $\hat{\eta}_{\min,\mathrm{conc}}^{\ell}$, i.e.\
    the largest $\eta$ at which a significant violation occurs.
  }
  \label{fig:second_differences}
\end{figure}

\begin{figure}[t!]
  \centering
  \includegraphics[width=0.6\linewidth]{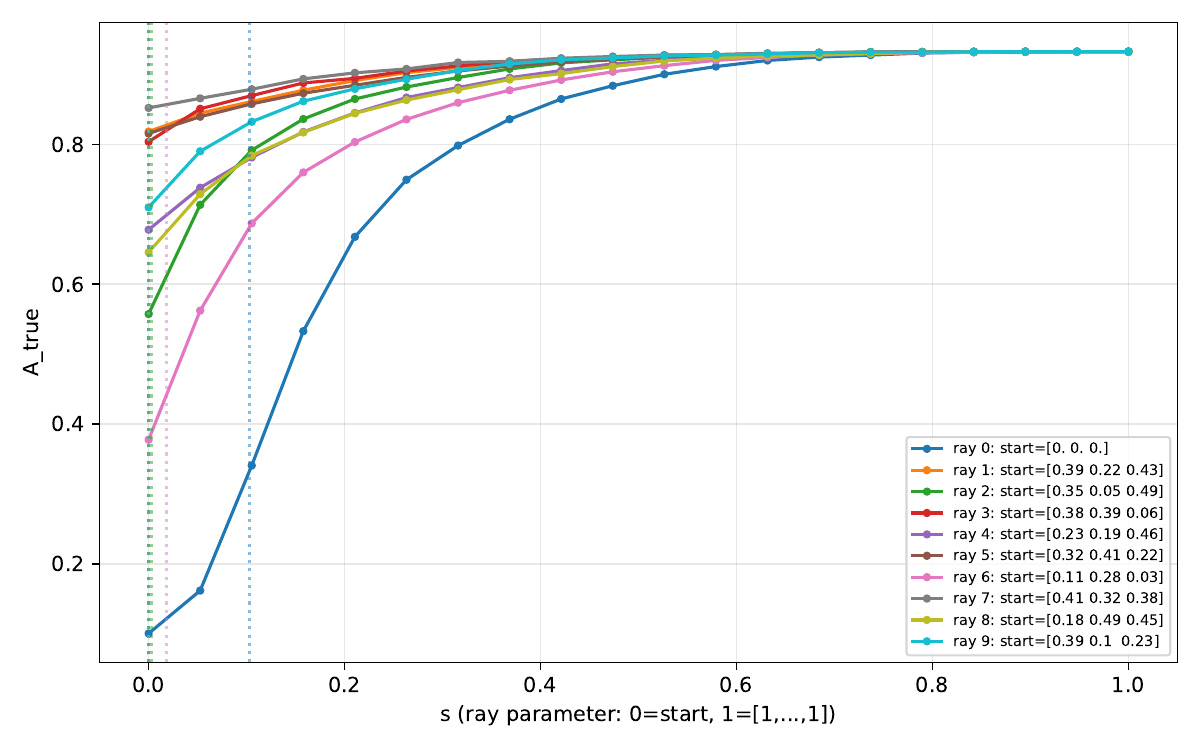}
  \caption{%
    \textbf{Directional accuracy profiles along random rays.}
    Each colored curve traces $A(\boldsymbol{\eta}(s))$ as $s$ increases from
    a random starting point to $\mathbf{1}$.  The dotted vertical line marks
    the ray parameter $s^*$ at which the ray enters the informative regime
    (all links above their respective $\hat{\eta}_{\min,\mathrm{final}}^\ell$).
    The consistent increase in slope at $s^*$ confirms that the collapse is
    correlated across links.
  }
  \label{fig:directional_profiles}
\end{figure}

\begin{figure}[t!]
  \centering
  \includegraphics[width=0.6\linewidth]{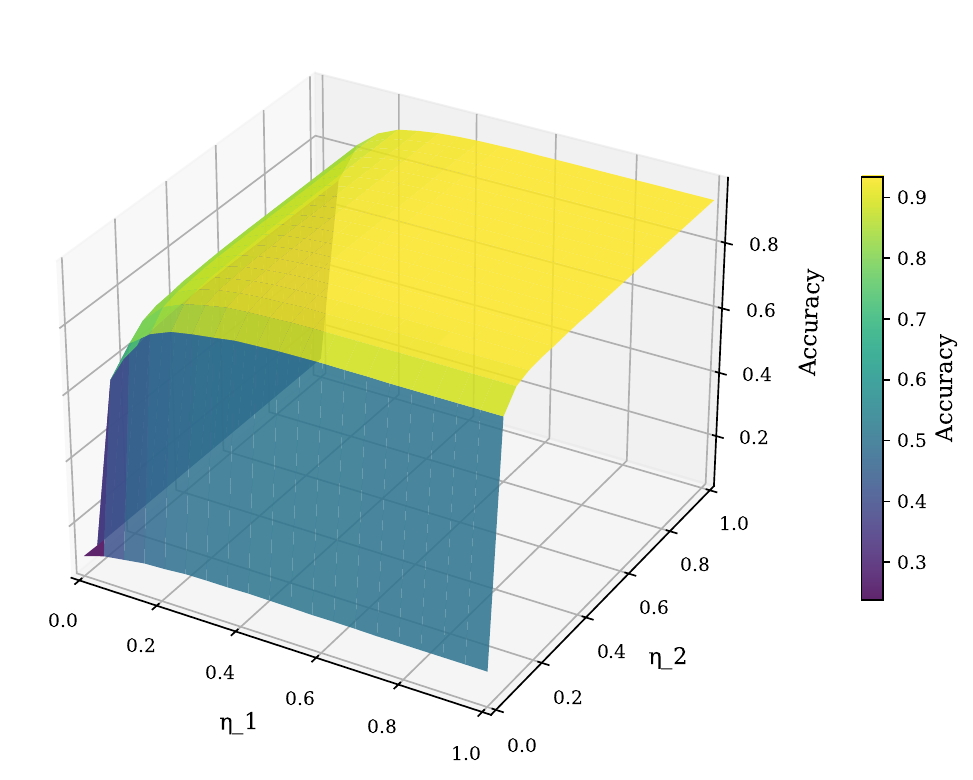}
  \caption{%
    \textbf{KNN-smoothed accuracy surface} ($k=8$, inverse-distance
    weighting).  Fitted on $N^2$ observations; the smooth surface suppresses
    per-point noise while preserving the monotone concave structure of the
    informative regime.
  }
  \label{fig:accuracy_surface_knn}
\end{figure}

\begin{figure}[t!]
  \centering
  \includegraphics[width=0.6\linewidth]{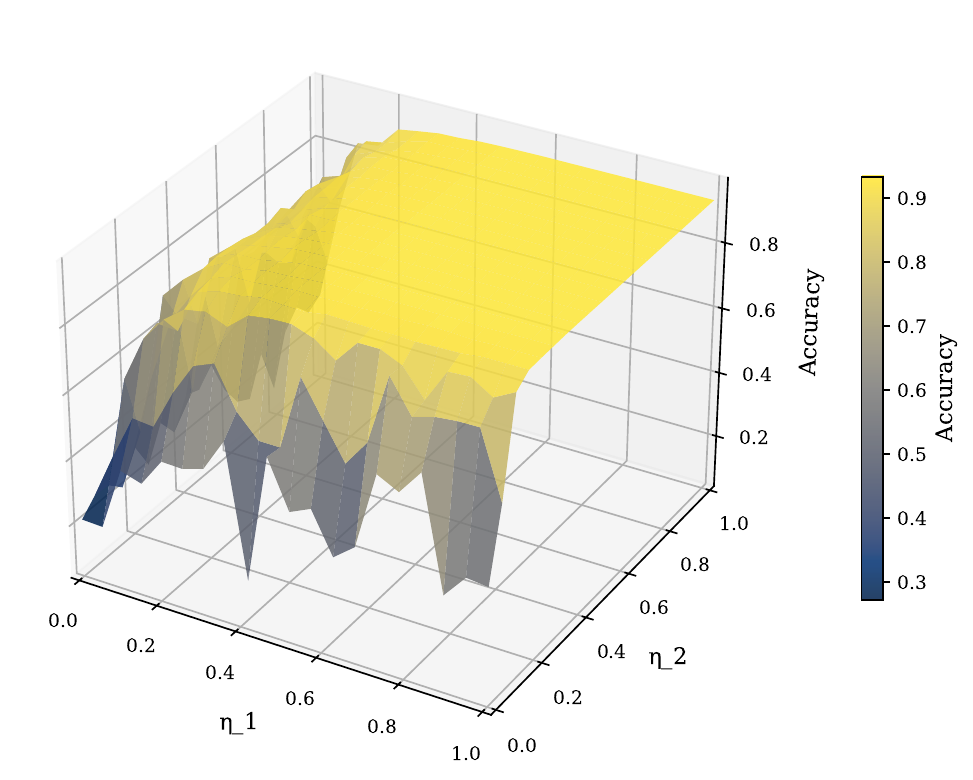}
  \caption{%
    \textbf{Stein-smoothed accuracy surface}
    $\tilde{A}(\boldsymbol{\eta})$ (equation~\eqref{eq:stein_smooth}),
    estimated by Monte Carlo with $M=10$ perturbation samples per point and
    bandwidth $\sigma=0.05$.  This is the surface effectively maximized by
    the gradient oracle during optimization.
  }
  \label{fig:accuracy_surface_stein}
\end{figure}

\subsection{Empirical Validation of the Concavity Assumption}
\label{sec:concavity-validation}

\CRedit{Above, we characterize $A(\boldsymbol{\eta})$ along axes, rays, and surfaces. We now verify the same premise, that each $A_k(\boldsymbol{\eta}_k)$ is monotone non-decreasing and \emph{concave}, statistically over the whole $\boldsymbol{\eta}$-hypercube, separately for each strategy: top-$k$ sparsification, quantization, and LLM.int8 mixed-precision.} Concavity is what reduces the multi-task program to the convex program of Theorem~\ref{thm:multi-convex}, and underpins Assumption~1.

\paragraph{Monte-Carlo Jensen inequality test.}
A function is concave iff its restriction to every line segment in its domain is concave, so we sample line segments through the cube and test each restriction. \CRedit{As above, $\boldsymbol{\eta}_k\in[0,1]^{L_k-1}$, with $\boldsymbol{\eta}_k=\mathbf{1}$ the uncompressed baseline.} We test the tasks of Table~\ref{tab:tasks} with different stage counts $L_k$. The procedure is:
\begin{enumerate}
  \item \textbf{Sample isotropic segments.} Draw endpoints
    $\boldsymbol{\eta}_{\mathrm{start}},\boldsymbol{\eta}_{\mathrm{end}}$
    independently and uniformly from $[\varepsilon,1]^{L_k-1}$
    ($\varepsilon=10^{-3}$), unconstrained in direction, so the ensemble covers
    the cube isotropically.
  \item \textbf{Evaluate along the segment.} Discretize each segment at $m$
    points
    $\boldsymbol{\eta}(t_j)=(1-t_j)\boldsymbol{\eta}_{\mathrm{start}}+t_j\boldsymbol{\eta}_{\mathrm{end}}$,
    $t_j=j/(m-1)$, and record $y_j=A_k(\boldsymbol{\eta}(t_j))$
    (reference-normalized perplexity retention, or raw accuracy).
  \item \textbf{Test Jensen's inequality on random chords.} Along each segment,
    sample triples $i<j<k$ and check the chord form of Jensen's inequality
    for concave functions,
    \begin{equation}
      y_j \;\ge\; \lambda\,y_i + (1-\lambda)\,y_k - \mathrm{tol},
      \qquad \lambda = \frac{t_k - t_j}{t_k - t_i},
      \label{eq:jensen}
    \end{equation}
    where $\mathrm{tol}\ge 0$ absorbs measurement noise without masking
    structural violations. The per-segment \emph{Jensen inequality pass rate} is
    the fraction of sampled inequalities that pass, the vertical axis of every
    figure below. We draw five triples per segment and repeat the draw five
    times, so each reported value is a mean across resamples rather than a
    single binomial estimate.
\end{enumerate}

Top-$k$ and LLM.int8 are effectively continuous in $\boldsymbol{\eta}$ and are
probed directly. Quantization is discrete: each $\eta_\ell$ snaps to a bit-width
rung before the layer is compressed, so we enumerate its lattice exhaustively
and read the same random segments off that grid. This costs no additional
evaluation and probes all three strategies identically. We report quantization
under the floor snap rule that the backends implement.
Testing curvature through randomly placed chords, rather than a bandwidth-dependent derivative estimate, follows the non-parametric simplex approach of Abrevaya and Jiang for noisy regression functions~\cite{AbrevayaJiang2005}.

\paragraph{Scalar $\eta$-floor sweep.}
\CRedit{As above, activations collapse into noise below a per-link threshold $\eta_{\min}$, so the question is not only whether $A_k$ is concave on the full cube, but on which sub-cube $[c,1]^{L_k-1}$.} For each scalar floor $c\in[0,0.7]$ we keep the grid points with $\min_\ell\eta_\ell\ge c$, sample triples among the survivors, and recompute the pass fraction. Sweeping $c$ traces how excluding the aggressive-compression corner changes measured concavity.

\begin{figure*}[t]
  \centering
  \includegraphics[width=\textwidth]{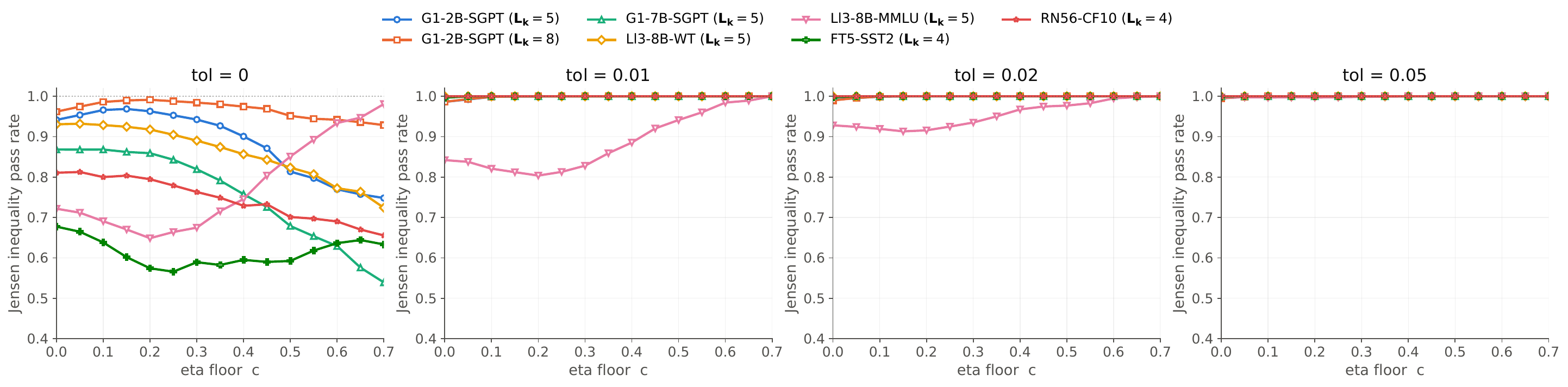}
  \caption{Monte-Carlo Jensen inequality test for \emph{top-$k$} sparsification. Each panel plots the fraction of sampled Jensen triples (Eq.~\eqref{eq:jensen}) satisfied against the scalar $\eta$-floor $c$ that restricts segments to $[c,1]^{L_k-1}$; panels are the chord tolerance $\mathrm{tol}\in\{0,0.01,0.02,0.05\}$, and each line is one scenario with a stage count~$L_k$. Every point is the mean over five resamples of five triples per segment. With a small chord tolerance ($\mathrm{tol} = 0.01$), the pass fraction for most tasks approaches $1.0$, even when the $\boldsymbol{\eta}$-floor is near $0$.}
  \label{fig:concavity-grid-topk}
\end{figure*}

\begin{figure*}[t]
  \centering
  \includegraphics[width=\textwidth]{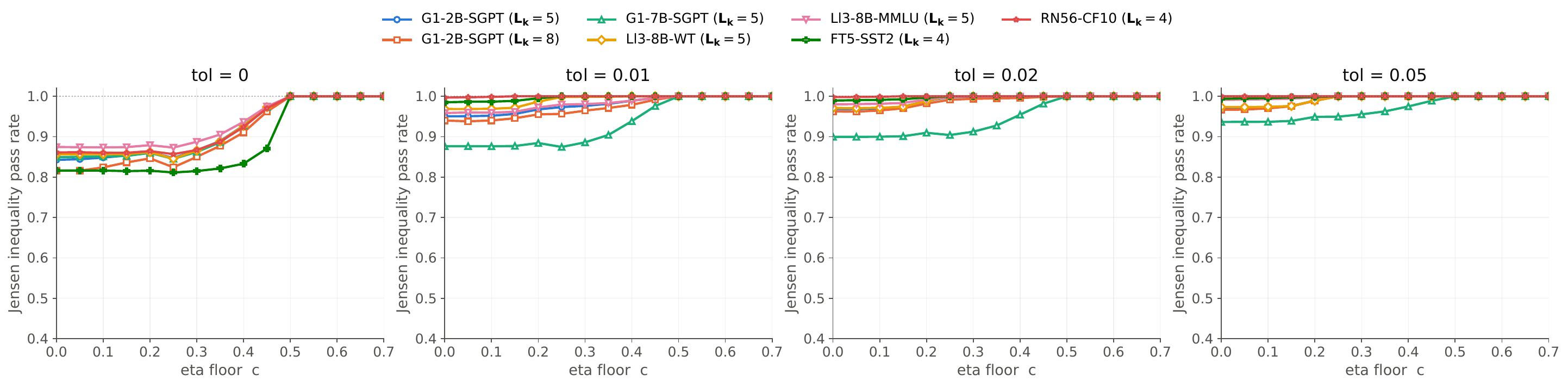}
  \caption{Monte-Carlo Jensen inequality test for the \emph{quantization} strategy. Axes, panels and lines as in Fig.~\ref{fig:concavity-grid-topk}. Past $c=0.5$ the floor rule leaves a single reachable rung, so every segment is constant there and the plateau at~$1.0$ records the absence of measurable curvature rather than a curvature test. Again, with a small chord tolerance ($\mathrm{tol} = 0.02$), the pass fraction for most tasks approaches $1.0$.}
  \label{fig:concavity-grid-quant}
\end{figure*}

\begin{figure*}[t]
  \centering
  \includegraphics[width=\textwidth]{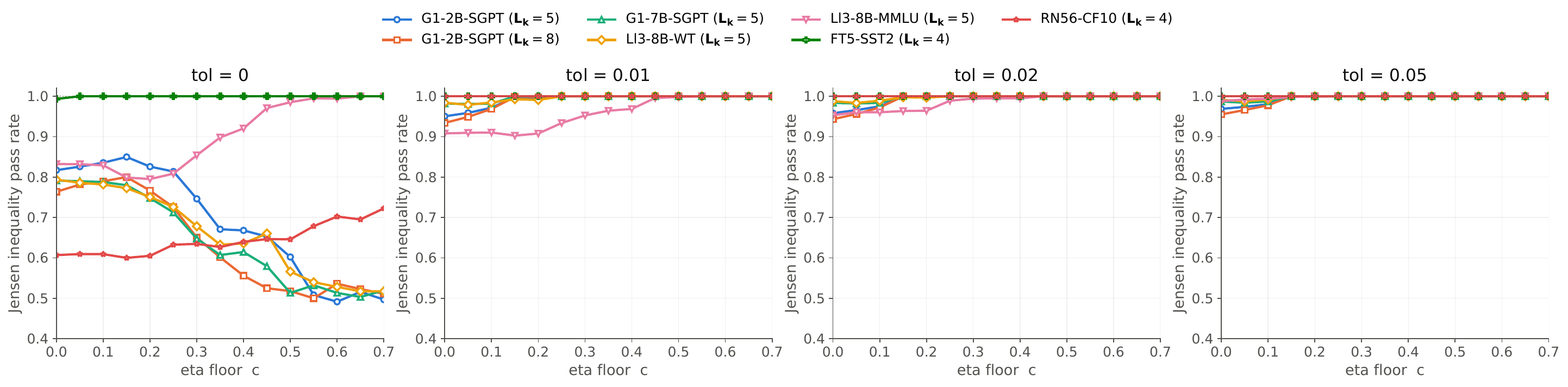}
  \caption{Monte-Carlo Jensen inequality test for the \emph{LLM.int8} strategy. Axes, panels and lines as in Fig.~\ref{fig:concavity-grid-topk}. With a small chord tolerance ($\mathrm{tol} = 0.01$), the pass fraction for most tasks approaches $1.0$.}
  \label{fig:concavity-grid-llmint8}
\end{figure*}

\paragraph{Results.}
Figures~\ref{fig:concavity-grid-topk}--\ref{fig:concavity-grid-llmint8} report the floor sweep for the three strategies. At $\mathrm{tol}=0$ the pass rate at $c=0$ spans $0.61$--$0.99$ across the twenty-one curves, dips to $0.49$ inside the sweep, and raising the floor does not repair it, those curves are flat or non-monotone in $c$ rather than climbing, so only one top-$k$ curve settles at~$1$ anywhere and only two of seven LLM.int8 curves do. At exact tolerance a chord missed by $10^{-6}$ scores identically to a genuine convex dip. Essentially the entire gap closes at the \emph{first} non-zero tolerance: $\mathrm{tol}=0.01$ lifts the worst value at $c=0$ from $0.61$ to $0.84$ and the median to~$0.98$, every curve then reaches~$1.000$ within the sweep, and widening further moves that worst value only $0.84\to0.90\to0.94$ at $\mathrm{tol}=0.01,0.02,0.05$ ($\ge 0.97$ at $\mathrm{tol}=0.1$). The curves snap to~$1$ at the first tolerance different than 0 and stay there, the signature of near-chord misses at evaluation-noise scale rather than of extended convex regions. The only task that behaves differently is \taskLangFive under top-$k$ ($0.804$ at $\mathrm{tol}=0.01$) and \taskLangFour under quantization ($0.900$ from $\mathrm{tol}=0.02$ on). Under every strategy, $A_k$ is concave up to measurement noise.

\paragraph{Why strict ($\mathrm{tol}=0$) concavity can \emph{fall} as $c$ rises.}
Raising the floor moves the test into the saturated neighborhood of $\boldsymbol{\eta}=\mathbf{1}$, where the surface is nearly flat: the returns from further decompression are already exhausted, so the local curvature, and with it the Jensen gap, approaches zero. A chord drawn across a flat plateau has almost no downward room, so any evaluation noise of size $\varepsilon$ registers as a violation. The $\mathrm{tol}=0$ curves on smooth perplexity landscapes therefore \emph{decrease} with $c$ even though the surface is not becoming more convex: the signal shrinks faster than the noise. A modest offset ($\mathrm{tol}\in\{0.01,\ldots,0.1\}$) absorbs exactly these gaps and restores a pass rate of essentially~$1$, confirming that the effect is an artifact of testing a near-zero second derivative, not a structural failure of Assumption~1.

The opposite pattern, concavity \emph{rising} with $c$, appears when the low-$\boldsymbol{\eta}$ collapse regime is itself the source of local convexity (discrete accuracy plateaus, jagged collapse); fencing that regime off with $\eta_{\min}$ is then the right remedy. That pattern is clear for \taskLangFive in Fig.~\ref{fig:concavity-grid-llmint8}.

\paragraph{Implications.}
Within the informative regime $\boldsymbol{\eta}\ge\eta_{\min}$, the only region in which we fit a surrogate or run the optimizer, $A_k$ is concave up to noise-scale violations for all three strategies. The residual strict-test dips fall exactly where concavity matters least: the saturated neighborhood of the uncompressed baseline, where the Jensen gap vanishes and noise dominates. Together these results validate the convex reduction of Theorem~\ref{thm:multi-convex} and the per-link water-filling solution of both algorithms.

\subsection{Regression Fitting}
\label{app:fitting}

As an alternative to the Stein-based approach \CRedit{of Sec.~\ref{app:stein}}, we also estimate the accuracy-compression function $A_k(\boldsymbol{\eta}_k)$ and its gradients using a fitting-based method. The basic idea is to conduct offline experiments on a partitioned model, establish the relationship between $\boldsymbol{\eta}_k$ and the resulting inference accuracy, and then fit a surrogate model to this relationship.

Specifically, for each task $k$ and scenario, we first build an offline partitioned model with a scenario-specific set of fixed cut points. For every cut point $i$, we predefine a candidate set of compression factors within $[\eta_{i,k}^{\min}, 1]$ that covers the feasible range as broadly as possible. We then enumerate all combinations of compression factors across the cut points. For each configuration of $\boldsymbol{\eta}_k$, we run offline inference on the hold-out dataset using the corresponding partitioned model and record the resulting accuracy $A_k(\boldsymbol{\eta}_k)$. In this way, we obtain \(N_f\) sampled compression-accuracy pairs,
\[
\left(\boldsymbol{\eta}_k^{(n)},\, A_k^{(n)}\right), \qquad n=1,2,\dots,N_{f},
\]
where \(\boldsymbol{\eta}_k^{(n)}\) denotes the compression vector of the \(n\)-th sample for task \(k\), and
\(A_k^{(n)}\) is the corresponding measured accuracy.

We then fit a surrogate model \(\hat{A}_k(\boldsymbol{\eta}_k)\) based on the sampled compression-accuracy pairs \((\boldsymbol{\eta}_k^{(n)}, A_k^{(n)})\). To train \(\hat{A}_k(\boldsymbol{\eta}_k)\), we use \emph{scikit-learn}, a widely used Python machine learning library that provides efficient tools for regression fitting. Since the shape of the accuracy-compression function may vary across tasks and scenarios, we evaluate several surrogate families, including \emph{linear\_monotonic}, \emph{poly2}, \emph{poly3}, \emph{gbm}, \emph{rf}, \emph{mlp}, and \emph{mlp\_small}.

\textbf{Linear monotonic (\emph{linear\_monotonic})}
This surrogate uses a linear regressor with nonnegative coefficients, enforcing a monotonic relationship between the compression features and the predicted accuracy. 

\textbf{Polynomial regression (\emph{poly2}, \emph{poly3})}
These surrogates use polynomial feature expansion followed by ridge regression to capture smooth nonlinear accuracy-compression relationships. The polynomial degree is set to 2 for \emph{poly2} and 3 for \emph{poly3}, with ridge regularization strength fixed at 0.1. Because the fitted functions are smooth, their gradients can be computed directly.

\textbf{Multilayer perceptron (\emph{mlp\_small}, \emph{mlp})}
These surrogates use feedforward neural networks to model more flexible nonlinear mappings. The \emph{mlp\_small} model uses two hidden layers of sizes 64 and 32 with up to 2000 training iterations, whereas \emph{mlp} uses three hidden layers of sizes 128, 64, and 32 with up to 3000 iterations. Both models use early stopping to reduce overfitting.

\textbf{Random forest (\emph{rf})}
This method uses a random forest regressor, which aggregates the predictions of multiple decision trees to improve robustness and capture nonlinear effects. In our implementation, the model uses 200 trees with a maximum depth of 15.

\textbf{Gradient boosting (\emph{gbm})}
This method uses a gradient boosting regressor, which builds an ensemble of trees sequentially so that each stage reduces the residual error of the previous one. Its main hyperparameters are 300 boosting stages, a maximum depth of 8, and a learning rate of 0.05.

\begin{table*}[t]
\centering
\scriptsize
\caption{Test accuracy and surrogate inference speed for different fitting models. Inference speed is reported in milliseconds per sample.}
\label{tab:fitting_acc_speed}
\setlength{\tabcolsep}{3pt}
\resizebox{\textwidth}{!}{%
\begin{tabular}{ll
rrr
rrr
rrr
rrr
rrr
rrr
rrr}
\toprule
\multirow{2}{*}{Task} & \multirow{2}{*}{Codec}
& \multicolumn{3}{c}{\emph{linear\_mono}}
& \multicolumn{3}{c}{\emph{poly2}}
& \multicolumn{3}{c}{\emph{poly3}}
& \multicolumn{3}{c}{\emph{rf}}
& \multicolumn{3}{c}{\emph{gbm}}
& \multicolumn{3}{c}{\emph{mlp\_small}}
& \multicolumn{3}{c}{\emph{mlp}} \\
\cmidrule(lr){3-5}\cmidrule(lr){6-8}\cmidrule(lr){9-11}\cmidrule(lr){12-14}\cmidrule(lr){15-17}\cmidrule(lr){18-20}\cmidrule(lr){21-23}
&
& RMSE & $R^2$ & Inf. (ms)
& RMSE & $R^2$ & Inf. (ms)
& RMSE & $R^2$ & Inf. (ms)
& RMSE & $R^2$ & Inf. (ms)
& RMSE & $R^2$ & Inf. (ms)
& RMSE & $R^2$ & Inf. (ms)
& RMSE & $R^2$ & Inf. (ms) \\
\midrule
\taskVisionTwo & \emph{topk}
& 0.0940 & 0.6937 & 0.1752
& 0.0441 & 0.9327 & 0.2709
& 0.0226 & 0.9822 & 0.2794
& 0.0009 & 1.0000 & 13.36
& 0.0019 & 0.9999 & 0.4332
& 0.1661 & 0.0422 & 0.2276
& 0.0269 & 0.9749 & 0.2038 \\

\taskVisionTwo & \emph{quantization}
& 0.3487 & 0.1638 & 0.1752
& 0.2999 & 0.3816 & 0.2892
& 0.2331 & 0.6265 & 0.2868
& 0.0016 & 1.0000 & 13.64
& 0.0018 & 1.0000 & 0.4536
& 0.4368 & -0.3122 & 0.2029
& 0.4289 & -0.2649 & 0.2243 \\

\taskVisionTwo & \emph{llmint8\_fp16\_int2}
& 0.0292 & 0.5688 & 0.1757
& 0.0100 & 0.9501 & 0.2707
& 0.0053 & 0.9859 & 0.2852
& 0.0061 & 0.9811 & 13.28
& 0.0101 & 0.9487 & 0.4338
& 0.0239 & 0.7124 & 0.1981
& 0.1110 & -5.2093 & 0.2075 \\

\taskVisionTwo & \emph{llmint8\_fp16\_int4}
& 0.0006 & 0.3058 & 0.1755
& 0.0007 & 0.2523 & 0.2793
& 0.0006 & 0.2668 & 0.2815
& 0.0007 & 0.2017 & 13.33
& 0.0008 & -0.2244 & 0.4522
& 0.2463 & -105348 & 0.2032
& 0.1451 & -36545 & 0.2081 \\

\taskVisionTwo & \emph{llmint8\_int8\_int4}
& 0.0007 & 0.0495 & 0.1746
& 0.0007 & 0.1081 & 0.2716
& 0.0006 & 0.3141 & 0.2795
& 0.0006 & 0.2561 & 13.30
& 0.0007 & -0.0463 & 0.4319
& 0.0217 & -895.3 & 0.1999
& 0.0912 & -15853 & 0.2087 \\

\taskVisionTwo & \emph{llmint8\_int8\_int2}
& 0.0276 & 0.6158 & 0.1759
& 0.0092 & 0.9569 & 0.2705
& 0.0041 & 0.9915 & 0.2818
& 0.0025 & 0.9968 & 13.11
& 0.0005 & 0.9999 & 0.4468
& 0.1159 & -5.7756 & 0.1984
& 0.1045 & -4.5031 & 0.2062 \\

\midrule
\taskLangSeven & \emph{topk}
& 0.0133 & 0.4869 & 0.1731
& 0.0098 & 0.7201 & 0.2693
& 0.0094 & 0.7401 & 0.2766
& 0.0067 & 0.8695 & 13.15
& 0.0070 & 0.8565 & 0.4196
& 0.0923 & -23.8149 & 0.1971
& 0.0827 & -18.9063 & 0.2040 \\

\taskLangSeven & \emph{quantization}
& 0.0791 & 0.3047 & 0.1794
& 0.0522 & 0.6971 & 0.2764
& 0.0360 & 0.8560 & 0.2800
& 0.0303 & 0.8978 & 13.17
& 0.0359 & 0.8570 & 0.4296
& 0.1110 & -0.3695 & 0.1967
& 0.1193 & -0.5817 & 0.2044 \\

\taskLangSeven & \emph{llmint8\_fp16\_int8}
& 0.0000 & 1.0000 & 0.1763
& 0.0000 & 1.0000 & 0.2709
& 0.0000 & 1.0000 & 0.2901
& 0.0000 & -1520 & 13.21
& 0.0000 & 1.0000 & 0.4052
& 0.9300 & 0.0000 & 0.2014
& 0.7418 & 0.0000 & 0.2068 \\
\bottomrule
\end{tabular}%
}
\end{table*}

\begin{table*}[t]
\centering
\scriptsize
\caption{Fidelity of the Gaussian-smoothed accuracy to the true accuracy. Stein's identity does not approximate $A$; it estimates $\nabla \tilde{A}_\sigma$, the gradient of the Gaussian-smoothed accuracy $\tilde{A}_\sigma(\bm{\eta})=\mathbb{E}_{\bm z}[A(\bm{\eta}+\sigma \bm z)]$. We report RMSE and $R^2$, for $\tilde{A}_\sigma$ against $A$ over the $\bm{\eta}$ grid, as a function of the smoothing scale $\sigma$. }
\label{tab:stein_smoothing}
\setlength{\tabcolsep}{3pt}
\resizebox{\textwidth}{!}{%
\begin{tabular}{llrr rr rr rr rr}
\toprule
\multirow{2}{*}{Task} & \multirow{2}{*}{Codec} & \multicolumn{2}{c}{$\sigma=0.01$} & \multicolumn{2}{c}{$\sigma=0.02$} & \multicolumn{2}{c}{$\bm{\sigma=0.05}$} & \multicolumn{2}{c}{$\sigma=0.1$} & \multicolumn{2}{c}{$\sigma=0.2$} \\
\cmidrule(lr){3-4}\cmidrule(lr){5-6}\cmidrule(lr){7-8}\cmidrule(lr){9-10}\cmidrule(lr){11-12}
 &  & RMSE & $R^2$ & RMSE & $R^2$ & RMSE & $R^2$ & RMSE & $R^2$ & RMSE & $R^2$ \\
\midrule
\taskVisionTwo \scriptsize(V-TRS/V-LOO) & \cmpT & 0.0139 & 0.9961 & 0.0168 & 0.9943 & 0.0196 & 0.9923 & 0.0560 & 0.9370 & 0.1265 & 0.6784 \\
 & \cmpQ & 0.0075 & 0.9996 & 0.0095 & 0.9994 & 0.0389 & 0.9897 & 0.1216 & 0.8994 & 0.1933 & 0.7456 \\
 & \cmpIEight & 0.0000 & 1.0000 & 0.0000 & 1.0000 & 0.0000 & 1.0000 & 0.0000 & 1.0000 & 0.0000 & 1.0000 \\
\bottomrule
\end{tabular}%
}
\end{table*}

We evaluate each fitting method using RMSE, $R^2$, and the inference latency of a single prediction call, since online optimization is sensitive to prediction overhead. Based on preliminary experiments that jointly consider fitting quality and prediction speed, we select the final surrogate for each task and scenario. Taking \taskVisionTwo and \taskLangSeven as representative examples, and following the model setup and partition strategy used in the scenarios \textbf{\scenVisionJetsonTestbedS} and \textbf{\scenLangJetsonTestbedS} on the Jetson platform, we evaluate the fitting quality and inference overhead of different surrogate models. The corresponding results for the two tasks are reported in Table~\ref{tab:fitting_acc_speed}. For comparison, Table~\ref{tab:stein_smoothing} presents the same accuracy metrics on the Gaussian-smoothed accuracy function approximated by Stein. The results show that the fitting quality differs significantly across surrogate families, suggesting that the fitting method should be selected task-dependently. Meanwhile, the inference latency of all fitted models remains far below the corresponding scenario-specific \(\tau\), demonstrating the efficiency advantage of fitting-based accuracy estimation in terms of runtime overhead.

We also evaluate the structural properties of all fitted models through one-dimensional slices, which allow us to examine characteristics such as monotonicity and convexity with respect to $\eta_{i,k}$. The Figure~\ref{fig:vision2-oflv1-poly3-analysis} provides an illustrative surrogate model \(\hat{A}_k(\boldsymbol{\eta}_k)\) example for the scenario \textbf{\scenVisionJetsonTestbedS}.

We further evaluate the fitted models through the behavior of the proposed algorithm in both offline and online experiments. If Algorithms~\ref{alg:NoCSI-SingleTask} and \ref{alg:NoCSI-MultiTask} consistently satisfy the delay constraints and achieve accuracy advances relative to the baselines in the test scenarios, this indicates that the fitting-based accuracy estimator provides high-quality estimates. Referring to the results under scenario \scenVisionJetsonSS in Figures~\ref{fig:topk-resnet-jetson} and \ref{fig:topk-resnet-jetson-stein}, we observe that the fitting-based method achieves better performance than Stein's method in this scenario.

\begin{figure*}[t]
    \centering
    \subfloat[One-Dimensional Accuracy Slices with Fixed Base Values.]{
        \includegraphics[width=0.45\textwidth]{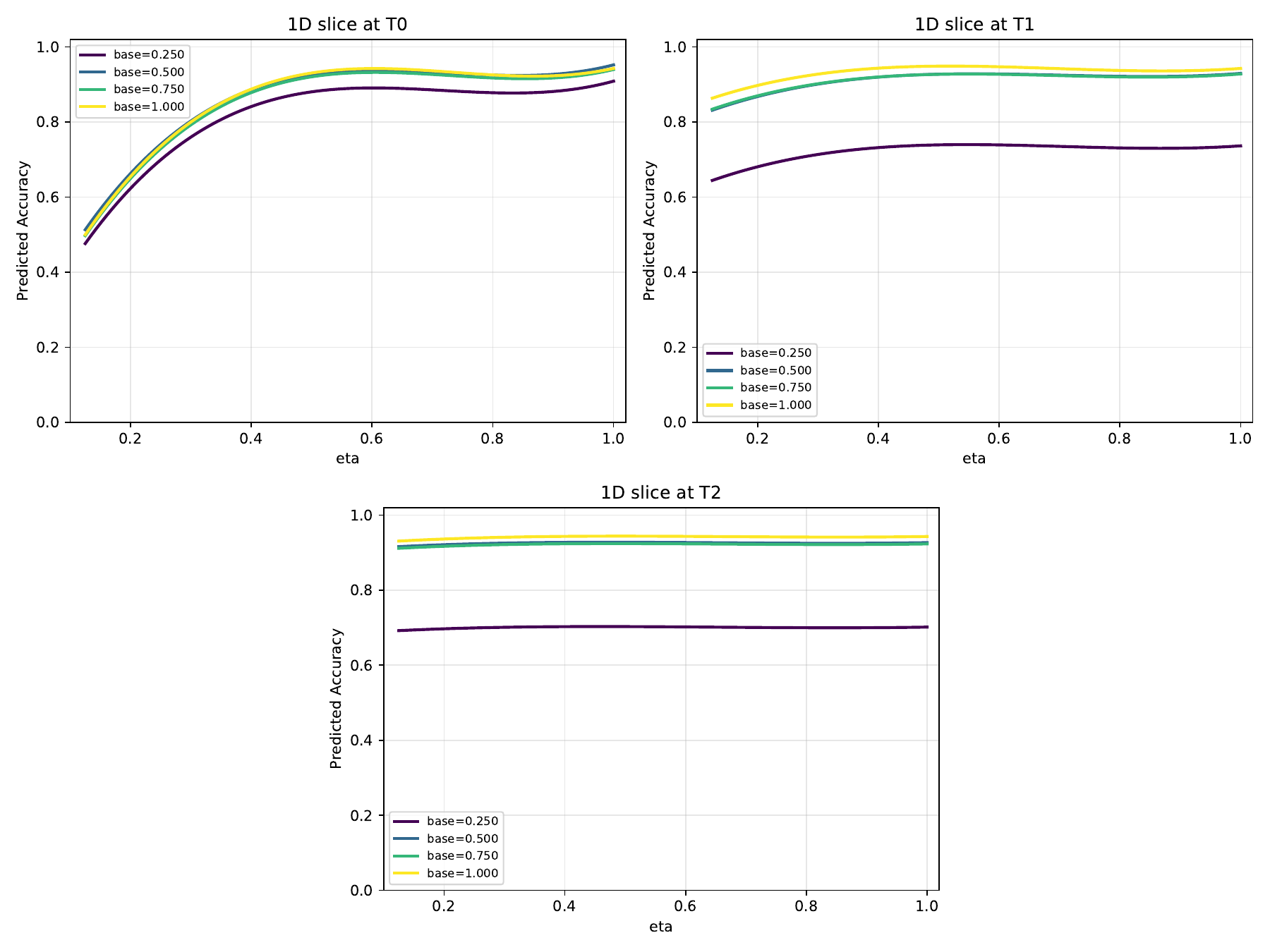}
        \label{fig:vision2-oflv1-poly3-1d-slice}
    }
    \hfill
    \subfloat[Second-Derivative Heatmaps under Different Base Values]{
        \includegraphics[width=0.5\textwidth]{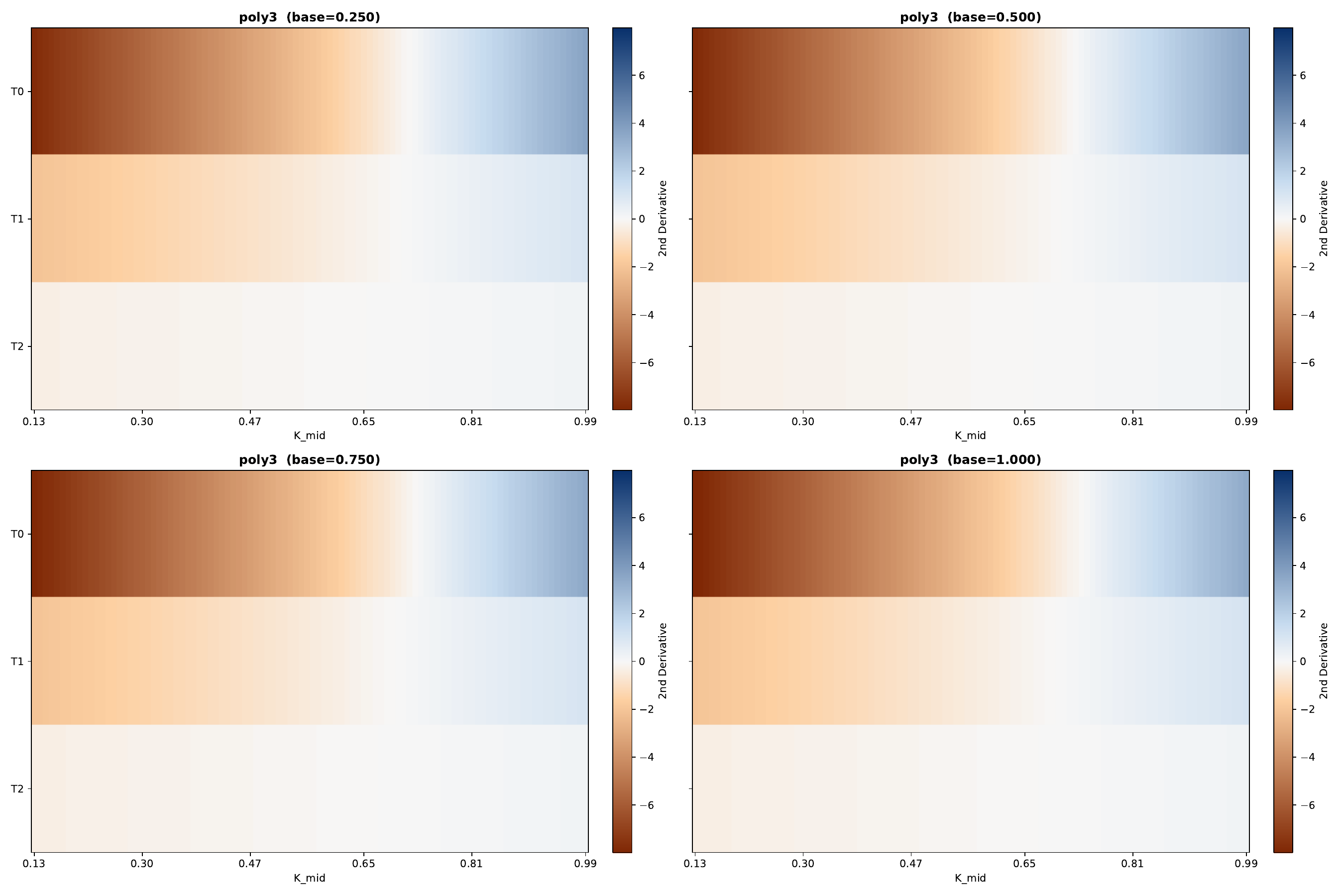}
        \label{fig:vision2-oflv1-poly3-curvature-heatmap}
    }
    \caption{Scenario \scenVisionJetsonTestbedS \emph{Poly3} Surrogate Visualization. The base value denotes the fixed value assigned to the non-scanned transfer ratios while one transfer ratio is varied to generate a one-dimensional slice.}
    \label{fig:vision2-oflv1-poly3-analysis}
\end{figure*}

\section{Additional Experimental Details}\label{app:reproducibility}

\subsection{Tasks and Datasets}

For vision, we consider a two-layer MLP on MNIST and ResNet-56 on CIFAR-10. For language, we consider Gemma-2 2B and 7B on MMLU and ShareGPT, Llama-3.1 8B on MMLU and WikiText, and Flan-T5-Base on SST-2. MMLU is evaluated with a fixed multiple-choice log-likelihood protocol. For the generative tasks (ShareGPT and WikiText), we define a reference-normalized retention score $A_k \in [0,1]$ as follows. For each held-out sample, we compute the mean per-token negative log-likelihood on an evaluation span under the compressed pipeline and under the uncompressed reference; the sample score is $\min\{1,\exp(\bar{H}^{\mathrm{ref}} - \bar{H}^{\mathrm{comp}})\}$, where $\bar{H}$ denotes this mean loss (equivalently, the ratio of reference to compressed perplexity on that span). The task utility is the average over samples, so $A_k=1$ means the compressed model matches the full-precision reference on the evaluation set. On ShareGPT, the span is the assistant-response tokens conditioned on the conversation context; on WikiText, it is all next-token positions within fixed-length text windows. SST-2 is a binary sentiment-classification subset of the Stanford Sentiment Treebank~\cite{sst2}; for Flan-T5-Base on SST-2, we use classification accuracy as the inference-quality metric.

\subsection{Scenarios}
\label{app:scenarios}

This section complements the scenario summaries in the main paper---the offline overview in Table~\ref{tab:scenarios} and the testbed overview in Table~\ref{tab:online-scenarios}---by providing \CRedit{the complete scenario parameter tables (Tables~\ref{tab:scenarios-full} and~\ref{tab:online-scenarios-full}) and} additional information for each scenario not reported in the main body, including the full algorithm parameter listings used to produce our reported results (Table~\ref{tab:results-offline-all} and Table~\ref{tab:results-testbed-all}), and brief explanations describing scenario-specific details and insights. \CRedit{Compressors and rate-selection policies referenced in these descriptions are specified in Appendices~\ref{app:compressors} and~\ref{appendix:algorithms}.}

\CRedit{Tables~\ref{tab:scenarios} and~\ref{tab:online-scenarios} summarize each scenario by identifier, task(s), target throughput $R_k$, and pipeline depth $L_k$. To keep the main-body tables compact, Table~\ref{tab:scenarios} additionally omits capacity ranges $[c_{\min},c_{\max}]$ and horizons $T$, while Table~\ref{tab:online-scenarios} omits network profiles and horizons $T$; both omit stage compute times $\boldsymbol{\tau}_k\in\mathbb{R}_+^{L_k}$ and uncompressed activation sizes $\boldsymbol{a}_k\in\mathbb{R}_+^{L_k-1}$. The complete tables are given below as Tables~\ref{tab:scenarios-full} and~\ref{tab:online-scenarios-full}. These quantities match Sec.~\ref{sec:meth}: $\tau_{i,k}$ is the time to execute stage $i$ of task $k$ on its assigned node, and $a_{i,k}$ is the uncompressed activation volume on inter-stage hop $i$ (hence $|\boldsymbol{a}_k|=L_k-1$). Capacity ranges, horizons, and network profiles are scenario constants used by the optimizers and reported in full below. On Jetson, $\boldsymbol{\tau}$ (and, in testbed runs, the measured $R_k$) varies with the compressor because stage times include encode/decode cost; $\boldsymbol{a}$ depends only on tensor shape and is compressor-invariant.}

\begin{table*}[t]
\centering
\small
\setlength{\tabcolsep}{4pt}
\renewcommand{\arraystretch}{1.05}
\caption{\CRedit{Complete offline scenario parameters corresponding to Table~\ref{tab:scenarios}.
$\boldsymbol{\tau}_k\in\mathbb{R}_+^{L_k}$ lists per-stage compute times and $\boldsymbol{a}_k\in\mathbb{R}_+^{L_k-1}$ lists uncompressed activation sizes on the inter-stage hops (Sec.~\ref{sec:meth}).
Where $\boldsymbol{\tau}$ differs across compression schemes within the same scenario, it reflects per-scheme computational profiles measured on the Jetson testbed; $\boldsymbol{a}$ is determined by tensor shape and is compressor-invariant.}}
\label{tab:scenarios-full}
{\color{black}
\begin{tabular}{@{}lccccc ll@{}}
\toprule
\textbf{Scenario} & \textbf{Task ID} & $\bm{c_{\min}}$--$\bm{c_{\max}}$ & $\bm{R_k}$ & $\bm{T}$ & $\bm{L_k}$ & $\boldsymbol{\tau}$ (ms) & $\boldsymbol{a}$ (MB) \\
 &  & (MB) & (Hz) &  &  &  &  \\
\midrule
\multicolumn{8}{c}{\textit{Single-task}} \\
\midrule
\scenLangSGTwoB
  & \taskLangThree  & 0.46--2.08  & 8  & 30  & 4
  & $[8, 10, 8, 10]$
  & $[0.176, 0.176, 0.176]$ \\
\scenLangSGTwoBSevenN
  & \taskLangThree  & 3.4--16     & 6  & 48  & 7
  & $[3.4, 3.4, 3.4, 2.3, 2.3, 2.3, 3.4]$
  & $[2.1, 2.1, 2.1, 2.1, 2.1, 2.1]$ \\
\scenLangSGSevenB
  & \taskLangFour   & 0.575--2.6  & 7  & 36  & 4
  & $[13, 13, 13, 13]$
  & $[0.252, 0.252, 0.252]$ \\
\scenVisionTRS
  & \taskVisionTwo  & 0.06--0.3   & 10 & 50  & 4
  & $[1, 1, 1, 1]$
  & $[0.066, 0.033, 0.016]$ \\
\scenVisionLOO
  & \taskVisionTwo  & 0.2--0.6    & 10 & 50  & 4
  & $[1, 1, 1, 1]$
  & $[0.066, 0.033, 0.016]$ \\
\scenLangFiveShot
  & \taskLangFive   & 10--40      & 10 & 50  & 3
  & $[14.7, 16.2, 14.8]$
  & $[4.2, 4.2]$ \\
\scenLangWTEightB
  & \taskLangSix    & 5--9        & 4  & 30  & 4
  & $[11.4, 11.4, 11.4, 11.4]$
  & $[4.19, 4.19, 4.19]$ \\
& & & & & & \cmpT: $[120, 61, 36, 19]$ & \\
\scenVisionJetsonSS
  & \taskVisionTwo  & 1.08--1.79  & 1  & 100 & 4
  & \cmpQ: $[131, 56, 36, 18]$
  & $[6.6, 3.3, 1.6]$ \\
& & & & & & \cmpIEight: $[334, 159, 81, 21]$ & \\
& & & & & & \cmpT: $[61, 52, 53, 105]$ & \\
\scenLangJetsonS
  & \taskLangSeven  & 1.35--1.95  & 2  & 100 & 4
  & \cmpQ: $[43, 37, 38, 103]$
  & $[1.2, 0.39, 0.39]$ \\
& & & & & & \cmpIEight: $[64, 47, 48, 104]$ & \\
\midrule
\multicolumn{8}{c}{\textit{Multi-task}} \\
\midrule
\multirow{3}{*}{\scenMTLangSGPTThree}
  & \taskLangThree & \multirow{3}{*}{4.2--7}   & 6  & \multirow{3}{*}{20}  & 4
  & $[8, 10, 8, 10]$ & $[0.176, 0.176, 0.176]$ \\
  & \taskLangThree &                           & 5  &                      & 4
  & $[8, 10, 8, 10]$ & $[0.176, 0.176, 0.176]$ \\
  & \taskLangThree &                           & 4  &                      & 4
  & $[8, 10, 8, 10]$ & $[0.176, 0.176, 0.176]$ \\
\multirow{2}{*}{\scenMTVisFiveLooseNTenTHundred}
  & \taskVisionTwo & \multirow{2}{*}{0.2--0.6} & 10 & \multirow{2}{*}{100} & 4
  & $[1, 1, 1, 1]$ & $[0.033, 0.049, 0.008]$ \\
  & \taskVisionTwo &                           & 10 &                      & 4
  & $[1, 1, 1, 1]$ & $[0.098, 0.016, 0.025]$ \\
\multirow{6}{*}{\scenMTLVJM}
  & \multirow{3}{*}{\taskVisionTwo} & \multirow{6}{*}{1.08--1.79} & \multirow{3}{*}{0.5} & \multirow{3}{*}{100} & \multirow{3}{*}{4}
  & \cmpT: $[120, 61, 36, 19]$ & \multirow{3}{*}{$[6.6, 3.3, 1.6]$} \\
  & & & & & & \cmpQ: $[131, 56, 36, 18]$ & \\
  & & & & & & \cmpIEight: $[334, 159, 81, 21]$ & \\
  & \multirow{3}{*}{\taskLangSeven} & & \multirow{3}{*}{1.5} & \multirow{3}{*}{100} & \multirow{3}{*}{4}
  & \cmpT: $[61, 52, 53, 105]$ & \multirow{3}{*}{$[1.2, 0.39, 0.39]$} \\
  & & & & & & \cmpQ: $[43, 37, 38, 103]$ & \\
  & & & & & & \cmpIEight: $[64, 47, 48, 104]$ & \\
\bottomrule
\end{tabular}%
}
\end{table*}

\begin{table*}[t]
\centering
\small
\setlength{\tabcolsep}{3.5pt}
\renewcommand{\arraystretch}{1.05}
\caption{\CRedit{Complete testbed scenario parameters corresponding to Table~\ref{tab:online-scenarios}.
As in Table~\ref{tab:scenarios-full}, $\boldsymbol{\tau}_k\in\mathbb{R}_+^{L_k}$ and $\boldsymbol{a}_k\in\mathbb{R}_+^{L_k-1}$.
Where $\boldsymbol{\tau}$ differs across compression schemes within the same scenario, it reflects per-scheme computational profiles measured on that testbed.
On Jetson, both the measured target $R_k$ and $\boldsymbol{\tau}$ depend on the compressor because stage times include encode/decode; $\boldsymbol{a}$ is compressor-invariant.}}
\label{tab:online-scenarios-full}
{\color{black}
\begin{tabular}{@{}lccccc ll@{}}
\toprule
\textbf{Scenario} & \textbf{Task ID} & \textbf{Network} & $\bm{R_k}$ & $\bm{T}$ & $\bm{L_k}$ & $\boldsymbol{\tau}$ (ms) & $\boldsymbol{a}$ (MiB) \\
 &  &  & (Hz) &  &  &  &  \\
\midrule
\multicolumn{8}{c}{\textit{PRESCIENT}} \\
\midrule
\scenPrsntVizWire
  & \taskVisionTwo & \netProfileLan  & 5.4 & 50 & 3
  & $[8.796, 7.871, 2.239]$
  & $[6.25, 1.5625]$ \\
\scenPrsntVizEdge
  & \taskVisionTwo & \netProfileEdge & 4.3 & 50 & 3
  & $[8.796, 7.871, 2.239]$
  & $[6.25, 1.5625]$ \\
\scenPrsntLangWikiWire
  & \taskLangSix   & \netProfileLan  & 4.0 & 50 & 3
  & $[21.078, 29.998, 19.108]$
  & $[4.0, 4.0]$ \\
\scenPrsntLangWikiEdge
  & \taskLangSix   & \netProfileEdge & 3.0 & 50 & 3
  & $[21.078, 29.998, 19.108]$
  & $[4.0, 4.0]$ \\
\scenPrsntLangMmluEdge
  & \taskLangFive  & \netProfileEdge & 2.5 & 50 & 3
  & $[24.856, 40.983, 31.539]$
  & $[16, 16]$ \\
\multirow{2}{*}{\scenPrsntVizVizWire}
  & \taskVisionTwo & \multirow{2}{*}{\netProfileLan}  & 2.7
  & \multirow{2}{*}{50} & \multirow{2}{*}{3}
  & \multirow{2}{*}{$[8.796, 7.871, 2.239]$}
  & \multirow{2}{*}{$[6.25, 1.5625]$} \\
  & \taskVisionTwo & & 2.7 & & & & \\
\multirow{2}{*}{\scenPrsntVizVizEdge}
  & \taskVisionTwo & \multirow{2}{*}{\netProfileEdge} & 1.2
  & \multirow{2}{*}{50} & \multirow{2}{*}{3}
  & \multirow{2}{*}{$[8.796, 7.871, 2.239]$}
  & \multirow{2}{*}{$[6.25, 1.5625]$} \\
  & \taskVisionTwo & & 1.2 & & & & \\
\multirow{2}{*}{\scenPrsntVizLangEdge}
  & \taskVisionTwo & \multirow{2}{*}{\netProfileEdge} & 1.1
  & \multirow{2}{*}{50} & \multirow{2}{*}{3}
  & $[8.796, 7.871, 2.239]$
  & $[6.25, 1.5625]$ \\
  & \taskLangSix & & 1.1 & &
  & $[21.078, 29.998, 19.108]$
  & $[4.0, 4.0]$ \\
\multirow{2}{*}{\scenPrsntVizLangWlan}
  & \taskVisionTwo & \multirow{2}{*}{\netProfileWlan} & 0.7
  & \multirow{2}{*}{50} & \multirow{2}{*}{3}
  & $[8.796, 7.871, 2.239]$
  & $[6.25, 1.5625]$ \\
  & \taskLangSix & & 0.5 & &
  & $[21.078, 29.998, 19.108]$
  & $[4.0, 4.0]$ \\
\multirow{2}{*}{\scenPrsntLangLangEdge}
  & \taskLangSix & \multirow{2}{*}{\netProfileEdge} & 0.9
  & \multirow{2}{*}{50} & \multirow{2}{*}{3}
  & \multirow{2}{*}{$[21.078, 29.998, 19.108]$}
  & \multirow{2}{*}{$[4.0, 4.0]$} \\
  & \taskLangSix & & 0.9 & & & & \\
\midrule
\multicolumn{8}{c}{\textit{Jetson}} \\
\midrule
& & & \cmpT: $0.530$ & & & \cmpT: $[120, 61, 36, 19]$ & \\
\scenVisionJetsonTestbedS
  & \taskVisionTwo & Wi-Fi AP & \cmpQ: $0.527$ & 50 & 4
  & \cmpQ: $[131, 56, 36, 18]$
  & $[6.6, 3.3, 1.6]$ \\
& & & \cmpIEight: $0.525$ & & & \cmpIEight: $[334, 159, 81, 21]$ & \\
& & & \cmpT: $5.17$ & & & \cmpT: $[61, 52, 53, 105]$ & \\
\scenLangJetsonTestbedS
  & \taskLangSeven & Wi-Fi AP & \cmpQ: $5.51$ & 50 & 4
  & \cmpQ: $[43, 37, 38, 103]$
  & $[1.2, 0.39, 0.39]$ \\
& & & \cmpIEight: $3.97$ & & & \cmpIEight: $[64, 47, 48, 104]$ & \\
\multirow{6}{*}{\scenMTJetsonTestbedM}
  & \multirow{3}{*}{\taskVisionTwo} & \multirow{3}{*}{Wi-Fi AP}
  & \cmpT: $0.274$ & \multirow{3}{*}{20} & \multirow{3}{*}{4}
  & \cmpT: $[120, 61, 36, 19]$ & \multirow{3}{*}{$[6.6, 3.3, 1.6]$} \\
  & & & \cmpQ: $0.273$ & & & \cmpQ: $[131, 56, 36, 18]$ & \\
  & & & \cmpIEight: $0.273$ & & & \cmpIEight: $[334, 159, 81, 21]$ & \\
  & \multirow{3}{*}{\taskLangSeven} & \multirow{3}{*}{Wi-Fi AP}
  & \cmpT: $0.820$ & \multirow{3}{*}{20} & \multirow{3}{*}{4}
  & \cmpT: $[61, 52, 53, 105]$ & \multirow{3}{*}{$[1.2, 0.39, 0.39]$} \\
  & & & \cmpQ: $0.821$ & & & \cmpQ: $[43, 37, 38, 103]$ & \\
  & & & \cmpIEight: $0.817$ & & & \cmpIEight: $[64, 47, 48, 104]$ & \\
\midrule
\multicolumn{8}{c}{\textit{Raspberry Pi}} \\
\midrule
\scenTbSSGPT
  & \taskLangThree & Wi-Fi 802.11ac & 5.0 & 100 & 7
  & $[378, 172, 145, 246, 155, 177, 143]$
  & $[4, 4, 4, 4, 4, 4]$ \\
\bottomrule
\end{tabular}%
}
\end{table*}

For each scenario, we describe the pipeline topology and the dataset subsets used for both utility evaluation and surrogate fitting. We also report dataset versions and model checkpoints at the task level, so that the same inputs can be regenerated from public sources.

Per-task data is drawn from two disjoint sample constructions. \datasetTest{} is the evaluation subset used inside the optimization loop to compute the utility metric $\metU$. \datasetGradFit{} is a separate set drawn with an independent seed, used \emph{only} to either fit the surrogate utility model or supply samples to the Stein zeroth-order estimator; its rows do not overlap with \datasetTest{}. The \textbf{Estimator} column of Tables~\ref{tab:repro-offline-alg}--\ref{tab:repro-online-alg} indicates which estimator is used to produce the headline numbers for that scenario---the fitted surrogate (\emph{Fit}) or the Stein oracle (\emph{Stein}). For Stein-based estimators, $\bm{N}$ and \datasetGradFit{} \# samples denote the number of random perturbation directions and the number of samples used per direction, respectively; for surrogate accuracy estimators, $\bm{N_f}$ denotes the number of $\eta$ points evaluated and \datasetGradFit{} \# samples the number of samples per $\eta$ point, respectively.

\subsubsection{Offline Scenarios}
\begin{table*}[t]
\centering
\caption{Offline experiment algorithm parameters. \datasetTest{} reports the number of dataset evaluation samples (per task for multi-task scenarios). The Estimator column indicates whether the quality metric is estimated via Stein's identity or a surrogate accuracy model. For Stein-based estimators, $\bm{N}$ and \datasetGradFit{} \# samples denote the number of random perturbation directions and the number of samples used per direction, respectively. For surrogate accuracy estimators, $\bm{N_f}$ and \datasetGradFit{} \# samples denote the number of $\eta$ points evaluated and the number of \datasetGradFit{} samples used per $\eta$ point, respectively. For Stein-based estimators, $\bm{\sigma}$ denotes the noise standard deviation. See Table~\ref{tab:scenarios} for scenario parameters.}

\label{tab:repro-offline-alg}
\small
\begin{tabular}{@{}l cccc cc cc}
\toprule
\textbf{Scenario} & $\bm{\mu}$ & $\bm{\varepsilon}$ & \textbf{\datasetTest{} samples} & \textbf{Estimator} & \textbf{\datasetGradFit{} \# samples} & $\bm{\sigma}$ & $\bm{N}$ & $\bm{N_f}$ \\
\midrule
\multicolumn{9}{c}{\textit{Single-task}} \\
\midrule
\scenLangSGTwoB          & 4.0   & 0.1 & 32         & Fit (\emph{poly2}) & 16  & -- & --  & 256 \\
\scenLangSGTwoBSevenN    & 6.0   & 0.1 & 32         & Fit (\emph{poly2}) & 16  & -- & --  & 256 \\
\scenLangSGSevenB        & 4.5   & 0.1 & 32         & Fit (\emph{poly2}) & 16  & -- & --  & 256 \\
\scenVisionTRS           & 3.0   & 0.1 & 10{,}000   & Stein & 100 & 0.05 & 20 & -- \\
\scenVisionLOO           & 3.0   & 0.1 & 10{,}000   & Stein & 100 & 0.05 & 20 & -- \\
\scenLangFiveShot        & 0.001 & 0.1 & 60 & Fit (MLP)& 30 & -- & -- & 256 \\
\scenLangWTEightB        & 4.5   & 0.1 & 32         & Fit (\emph{poly2}) & 16  & -- & --  & 256 \\
\midrule
\scenVisionJetsonSS     & 0.1 & 0.1 & 10{,}000 & Fit (\emph{poly3}) & 10{,}000 & -- &-- & 125\\  
\scenLangJetsonSS       & 0.1 & 0.1 & 100      & Fit (\emph{poly3}) & 100 & -- &--& 125\\
\midrule
\multicolumn{9}{c}{\textit{Multi-task}} \\
\midrule
\scenMTLangSGPTThree            & 4.0 & 0.1 & [16, 16, 16]         & Fit (\emph{poly2}) & [8, 8, 8]     & --         &  -- &   [256, 256, 256]   \\
\scenMTVisFiveLooseNTenTHundred & 3.0 & 0.1 & [10{,}000, 10{,}000] & Stein & [100, 100]    & [0.05, 0.05] & [10, 10] &  --  \\

\midrule
\scenMTLVJM & 0.1 & 0.1 & [10{,}000, 100] & Fit (\emph{poly3}) & [10{,}000, 100]  & -- & -- & [125, 125]\\
\bottomrule
\end{tabular}
\end{table*}
We describe all single-task and multi-task offline scenarios below, with algorithm parameters listed in Table~\ref{tab:repro-offline-alg}. For each scenario, we describe the pipeline topology, estimator choice, and dataset configuration.

\smallskip\noindent\textbf{\scenLangSGTwoB.}
Single-task ShareGPT evaluation with Gemma-2~2B on a four-stage pipeline (three links). Forward hooks compress activations after decoder blocks $3$, $8$, and $12$ (18-layer architecture). ShareGPT rows come from \texttt{Aeala/ShareGPT\_Vicuna\_unfiltered} (default \texttt{train} split) with prompt/response windows capped at $256$ tokens. Utility during the run is evaluated on \datasetTest{} using the sample count in Table~\ref{tab:repro-offline-alg}; surrogate training fits on \datasetGradFit{} with a different seed. Reported results follow the \emph{poly2} estimator.

\smallskip\noindent\textbf{\scenLangSGTwoBSevenN.}
Same ShareGPT task and model family as \scenLangSGTwoB, but the pipeline is elongated to \emph{seven} stages (six links) with cut indices $2,5,8,10,12,14$ so that each stage contains fewer transformer blocks; this stresses deeper device chains and matches the seven-node linear topology used for the Raspberry~Pi ShareGPT testbed. The intent is to observe scaling effects when only topology depth and link budgets change while the evaluation objective stays ShareGPT perplexity retention. \datasetTest{} / \datasetGradFit{} conventions are as above. The results use the \emph{poly2} estimator.

\smallskip\noindent\textbf{\scenLangSGSevenB.}
ShareGPT with Gemma-2~7B on four stages (three links), hooks after layers $6$, $13$, and $20$ (28-layer model). We use the same ShareGPT dataset protocol as \scenLangSGTwoB. Again, we report results obtained with \emph{poly2} as our regression fitting.

\smallskip\noindent\textbf{\scenVisionTRS.}
This scenario runs ResNet-56 on CIFAR-10 in a \emph{threshold} capacity
regime ($c \in [0.06, 0.3]$\,MB): feasible compression ratios exist in most
channel realizations but are occasionally out of reach, placing the optimizer
near the boundary of feasibility.
Gradients are estimated using a Stein gradient oracle, which
queries the fast accuracy callable at $2N$ antithetic perturbations of the
current~$\eta$; each query runs inference on \datasetTest{} with a reduced
sample budget.

\smallskip\noindent\textbf{\scenVisionLOO.}
Also using ResNet-56 on CIFAR-10, this scenario operates in a \emph{loose}
capacity regime ($c \in [0.2, 0.6]$\,MB), where the channel is almost always
wide enough to accommodate feasible compression ratios.
As in \scenVisionTRS, a Stein gradient oracle provides gradient
estimates by evaluating perturbed accuracy values on \datasetTest{}.

\smallskip\noindent\textbf{\scenLangFiveShot.}
This scenario evaluates Llama-3.1-8B on MMLU in a 5-shot prompting setting
under high-capacity conditions ($c \in [10, 40]$\,MB).
Unlike the vision scenarios, whose Stein oracles query \datasetTest{} at
every gradient step, this scenario uses a
Predictor gradient oracle: an accuracy predictor surrogate is
fitted \emph{offline} by evaluating the full-budget accuracy callable on a
grid of $(\eta, A_{\mathrm{true}})$ samples drawn from \datasetGradFit{},
training a concave monotone network to interpolate the accuracy surface.
At runtime, exact gradients are obtained in milliseconds via PyTorch autograd
through the fitted surrogate, bypassing the need for any further model
inference---a critical advantage given the high cost of Llama evaluation.
\datasetTest{} is used only for final performance reporting, not during
optimization.

\smallskip\noindent\textbf{\scenLangWTEightB.}
WikiText-2~\cite{merity2016pointer} (raw) language modeling with Llama-3.1~8B on four stages (three links) and decoder hooks after layers $7$, $15$, and $23$ (32-layer stack). By default we load \texttt{wikitext/wikitext-2-raw-v1} on the \texttt{test} split, extract non-empty text rows, chunk to fixed length $512$ tokens with stride $256$, and score next-token perplexity retention. Table~\ref{tab:repro-offline-alg} lists the \emph{poly2} hyperparameters.

\smallskip\noindent\textbf{\scenVisionJetsonSS}, \textbf{\scenLangJetsonS}, \textbf{\scenMTLVJM}
These scenarios correspond to the single-task and multi-task offline experiments whose parameters, such as $\tau$ and $c$, are obtained from emulated experiments on the Jetson testbed. Model and dataset details are provided in Appendix~\ref{app:testbed_jetson}. Before running the offline experiments, we first perform a set of preliminary measurements on the Jetson testbed to obtain the required scenario parameters for each model-compressor pair, including $\bm{c}$, $\bm{\tau}$, and $\bm{a}$ \CRedit{(Table~\ref{tab:scenarios-full})}. The collected scenario parameters are then used as inputs to the offline experiments on a high-performance computer. In the offline setting, $\bm{\tau}$ includes only the inference delay and the compressor compression/decompression delay. For utility evaluation, we use the full test set. We evaluate both \emph{fitting} and \emph{kernel-smoothing} under multiple values of $R_k$, while a representative result is reported in Table~\ref{tab:results-offline-all}. For \emph{Top-$k$}, the single-task tradeoff curves are shown in Figures.~\ref{fig:topk-resnet-jetson} and \ref{fig:topk-resnet-jetson-stein}, while the multi-task tradeoff curves are shown in Figures.~\ref{fig:multi-topk-fitting-task1-fixed-r2} and \ref{fig:multi-topk-stein-task1-fixed-r2}.


\smallskip\noindent\textbf{\scenMTLangSGPTThree.} Three concurrent ShareGPT replicas on Gemma-2~2B ($K{=}3$), each with its own four-stage pipeline but identical topology parameters. This stress-tests proportional baselines when workloads are symmetric and comparatively slack. \datasetTest{} / \datasetGradFit{} follow the ShareGPT protocol with per-task seeds; reporting is \emph{poly2}.

\smallskip\noindent\textbf{\scenMTVisFiveLooseNTenTHundred.}
A multi-task scenario pairing two ResNet-56/CIFAR-10 inference tasks that
share the same channel under loose capacity conditions ($c \in [0.2, 0.6]$\,MB).
The two tasks are assigned heterogeneous activation sizes, testing the
scheduler's ability to allocate capacity fairly across tasks with different
compression footprints.
Each task maintains an independent Stein gradient oracle that queries
\datasetTest{} with $N=10$ antithetic sample pairs; the scenario runs for
$T=100$ time steps.

\subsubsection{Testbed Scenarios}



\begin{table*}[t]
\centering
\caption{Testbed experiment algorithm parameters. \datasetTest{} reports the number of dataset evaluation samples, respectively (per task for multi-task scenarios). The Estimator column indicates whether the quality metric is estimated via Stein's identity or a surrogate accuracy model. For Stein-based estimators, $\bm{N}$ and \datasetGradFit{} \# samples denote the number of random perturbation directions and the number of samples used per direction, respectively; for surrogate accuracy estimators, $\bm{N_f}$ and \datasetGradFit{} \# samples denote the number of $\eta$ points evaluated and the number of \datasetGradFit{} samples per $\eta$ point, respectively. For Stein-based estimators, $\bm{\sigma}$ denotes the noise standard deviation. See Table~\ref{tab:online-scenarios} for scenario parameters.}

\label{tab:repro-online-alg}
\small
\begin{tabular}{@{}l cccc cc cc}
\toprule
\textbf{Scenario} & $\bm{\mu}$ & $\bm{\varepsilon}$ & \textbf{\datasetTest{} samples} & \textbf{Estimator} & \textbf{\datasetGradFit{} \# samples} & $\bm{\sigma}$ & $\bm{N}$ & $\bm{N_f}$ \\
\midrule
\multicolumn{9}{c}{\textit{Single-task}} \\
\midrule
\scenTbSSGPT & 1.0 & 0.1 & 7 & Fit (\emph{poly2}) & 5 & -- & -- & 10 \\

\scenPrsntVizWire & 4.0 & 0.1 & 1{,}000 & Stein & 100 & 0.05 & 10 & -- \\

\scenPrsntVizEdge & 4.0 & 0.1 & 1{,}000 & Stein & 100 & 0.05 & 10 & -- \\

\scenPrsntLangWikiWire & 4.0 & 0.01 & 20 & Stein & 10 & 0.05 & 5 & -- \\

\scenPrsntLangWikiEdge & 4.0 & 0.01 & 20 & Fit (\emph{poly3}) & 10 & -- & -- & 1{,}000\\

\scenPrsntLangMmluEdge & 4.0 & 0.1 & 60 & Fit (\emph{poly3}) & 20 & -- & -- & 1{,}000 \\

\scenVisionJetsonTestbedS & 0.1 & 1e-7 & 10{,}000 & Fit (\emph{poly3}) & 10{,}000 & -- & -- & 125 \\

\scenLangJetsonTestbedS & 0.1 & 1e-7 & 100 & Fit (\emph{poly3}) & 100 & -- & -- & 125 \\

\midrule
\multicolumn{9}{c}{\textit{Multi-task}} \\
\midrule

\scenPrsntVizVizWire & 4.0 & 0.1 & [1{,}000, 1{,}000] & Stein & [100, 100] & [0.05, 0.05] & [10, 10] & -- \\

\scenPrsntVizVizEdge & 4.0 & 0.1 & [1{,}000, 1{,}000] & Stein & [100, 100] & [0.05, 0.05] & [10, 10] & -- \\

\scenPrsntVizLangEdge & 4.0 & 0.1 & [600, 12] & Stein & [100, 10] & [0.05, 0.05] & [10, 5] & -- \\

\scenPrsntVizLangWlan & 4.0 & 0.1 & [600, 12] & Stein & [100, 10] & [0.05, 0.05] & [10, 5] & -- \\

\scenPrsntLangLangEdge & 4.0 & 0.1 & [12, 12] & Stein & [10, 10] & [0.05, 0.05] & [5, 5] & -- \\

\scenMTJetsonTestbedM & 0.1 & 0.1 & [10{,}000, 100] & Fit (\emph{poly3}) & [10{,}000, 100] & -- & -- & [125, 125]\\

\bottomrule
\end{tabular}
\end{table*}
We describe all testbed scenarios below, grouped by platform (PRESCIENT, Jetson, Raspberry Pi), with algorithm parameters listed in Table~\ref{tab:repro-online-alg}. For each scenario, we describe the pipeline topology, network profile, and dataset configuration.

\smallskip\noindent\textbf{\scenTbSSGPT.} ShareGPT with Gemma-2~2B on seven stages (six links), hooks after layers $3$, $5$, $8$, $10$, $12$, and $15$ (18-layer model). We use the same ShareGPT dataset protocol as \scenLangSGTwoB with prompt/response windows capped at 50 tokens. The \datasetTest{} and \datasetGradFit{} counts are indicated in Table~\ref{tab:repro-online-alg}.

\smallskip\noindent\textbf{\scenPrsntVizWireEdge.} ResNet-56 partitions pipelined across 3 nodes in a linear topology with the network being the differentiator - effectively unbounded or moderately constrained for ResNet-sized activations with some packet loss. The throughput ceiling is compute bound in \netProfileLan{} while \netProfileEdge{} enables us to measure the network share of the end-to-end cost for small-activation workloads. 

\smallskip\noindent\textbf{\scenPrsntLangWikiWire.} The unbounded network of before is no longer enough to completely decouple task throughput from network owing to the large Llama model activations. Useful to measure how much a gigabit link actually benefits LLM-scale activations and whether optimizers could be necessary in such deployment scenarios as well.

\smallskip\noindent\textbf{\scenPrsntLangWikiEdge.} Large fp16 hidden states on a constrained link — the strongly bandwidth-bound regime. WikiText perplexity gives a smooth, continuous accuracy signal, which makes the resulting trade-off surface the cleanest to read in the whole set. The first experiment scenario where our fitted \emph{poly3} accuracy model outperforms the Stein oracle. The per-slot $\lambda$ trace shows the active oscillations as the dual variable constantly adjusts compression when a slot runs slow and relaxes when one runs fast. The fitted \emph{poly3} model captures the llmint8 accuracy surface for Llama well enough to guide $\eta$ precisely into the narrow band ($\eta \approx 0.12$) where both accuracy and throughput are good. The Stein oracle over-compresses at ($\eta \approx 0.052$). On Llama activations, the Stein oracle evaluates accuracy by perturbing $\eta$ and measuring perplexity changes through the simulation pipeline. The perplexity function may have a flatter gradient in the low-$\eta$ region
making it harder for the finite-difference gradient estimate to pull $\eta$ back up. The LCB channel estimator also contributes — it systematically underestimates capacity (~90\%, amplified in testbed deployment), and the combined pessimism results in more aggressive compression than necessary.

\smallskip\noindent\textbf{\scenPrsntLangMmluEdge.} Identical setup with Llama model in a heavily constrained network but with the MMLU dataset instead. MMLU accuracy is a 0/1 per-question signal aggregated to a noisy, step-like curve over compression rate. The direct contrast with the WikiText run — same model, same network, same activations, different metric — shows how much the choice of evaluation metric shapes the apparent sensitivity to compression for LLM inference. Another scenario where the fitted \emph{poly3} model outperforms the Stein oracle \algOURSNOCSIst{} optimizer. Due to the nature of the MMLU accuracy function and the limited number of test samples feasible in this network, our optimizers are unable to exploit the flat accuracy function. In fact, some baselines beat \algOURSNOCSIst{} variants (fitted model and Stein oracle) comfortably.

\smallskip\noindent\textbf{\scenPrsntVizVizWireEdge.} Two identical ResNet-56 pipelines with complete node overlap. This overlap is maintained across all multi-task scenarios to introduce FIFO contention dynamics and compute sharing on nodes. Tested under two network profiles - unconstrained and reasonably constrained for multiple pipelines even with ResNet-sized activations.

\smallskip\noindent\textbf{\scenPrsntVizLangEdge.} Heterogeneous compute and activation sizes on a shared constrained link. Small ResNet-56 tasks queued behind longer running Llama tasks in a shared FIFO queue means total task throughput is bounded by the slower Llama throughput.

\smallskip\noindent\textbf{\scenPrsntVizLangWlan.} Same pipeline configurations but on an asymmetric, lossy, jittery link prone to packet losses. Link capacity is stochastic and due to the asymmetry of links, the per-link compression trade-off is no long symmetric across the pipeline. Dropped ResNet-56 packets have significantly less impact than dropping large Llama activation payloads resulting in interest throughput performance over long durations.

\smallskip\noindent\textbf{\scenPrsntLangLangEdge.} Node contention becomes a significant factor as large Llama workloads fill the FIFO queue whose dynamics dominate throughput in contrast to the \scenPrsntVizVizWireEdge{} scenarios where the same topology and bandwidth suffered from compute queuing instead of network queuing.

\smallskip\noindent\textbf{\scenVisionJetsonTestbedS}, \textbf{\scenLangJetsonTestbedS}, \textbf{\scenMTJetsonTestbedM}
This group of scenarios corresponds to real experiments on the Jetson testbed, where both our algorithm and the baselines are executed entirely on the Jetson platform. Details of the topology, network, models, and datasets are provided in Appendix~\ref{app:testbed_jetson}. Testbed experiment results are reported in Table~\ref{tab:results-testbed-all-nocsi}. We also provide a \emph{Top-$k$} single-task testbed result in Figure~\ref{fig:mu-sweep}.

Task processing and transmission are fully pipelined on the testbed. The first node serves both as the control node and as the host of the first inference partition; it fetches input samples, runs the first partition, and updates the control policy. The remaining nodes host only their respective inference partitions. On the first node, inference and optimization run on separate threads. Each inference instance reads the latest available policy before execution, while the optimization thread updates the policy in the background and assigns the corresponding $\eta$, $s_{i,k}^{\mathrm{comm}}$, and $s_{i,k}^{\mathrm{comp}}$.

The time-slot interval is dynamic. A slot starts when optimization begins and ends after at least five instance-completion feedback signals of the slower task have been collected for the current policy. If policy updating takes longer, the slot includes more instances and ends when the new policy update finishes. In practice, \emph{kernel-smoothing} is slower, so the slot duration is usually aligned with the optimization time, whereas \emph{fitting} is faster and the slot typically ends after five completed instances. To speed up online updates, we set the iteration number $J$ in Algorithm~\ref{alg:NoCSI-MultiTask} to 1.

At each node, transmissions for different tasks are carried out independently and in parallel. The allocation of $s_{i,k}^{\mathrm{comm}}$ is realized by assigning different token-bucket rates to different tasks. Each node also stores multiple task models and runs them in parallel on separate threads. Although $s_{i,k}^{\mathrm{comp}}$ is included in the algorithm, it is not explicitly enforced in the Jetson testbed. In our Jetson scenarios, however, transmission delay is much larger than processing delay based on measurement, so this has little effect on the results.

Because precise probing and control of link bandwidth are difficult in a shared Wi-Fi environment, CSI-aware baselines are not evaluated on the Jetson testbed. Moreover, accuracy evaluation in the online testbed setting differs from that in the offline experiments: rather than evaluating over the full test set, we compute accuracy using only the completed instances for which feedback is received.

Unlike the offline setting, where each task has one instance per slot, the number of completed instances may vary across tasks within the same slot in the online testbed. Therefore, the actual
bottleneck delay $D_{\mathrm{act},k}(t)$ is first averaged across completed instances within each task $k$ and slot $t$, then aggregated across tasks with weights $w_k$.

\subsection{Activation Compression Methods}
\label{app:compressors}

Throughout the experiments, the compression factor $\eta\in[0,1]$ is the \emph{retained fraction} relative to the uncompressed activation tensor in its reference floating-point dtype (for language models, activations are taken at selected decoder layers; for vision models, at hooked layers in the convolutional pipeline). 

\textbf{\emph{Top-$k$}} sparsification keeps a fraction $\eta$ of coordinates with largest magnitude and zeros the rest: for transformer hidden states we apply this mask \emph{per token} along the hidden dimension, while for image backbones we flatten each sample's activation vector and keep the top $\eta$ fraction of entries \emph{per sample}. For \emph{Top-$k$} compression, the indices of the retained values must also be conveyed. We use a PackBits mask, in which each value is assigned one bit to indicate whether it is transmitted, and the resulting bit-packed mask is sent together with the compressed activation. The size of this mask is fixed for a given activation tensor shape and does not vary with the value of $\eta$, thus introducing a constant overhead.

\textbf{\emph{Uniform quantization} }uses symmetric uniform quantization with a per-token (per-row) scale along the hidden dimension. By default, the uncompressed representation uses FP32 precision. To reduce quantization loss, we apply several implementation techniques for lower-precision compressors. For FP16, we use direct datatype conversion together with value clipping to avoid overflow. For INT8, INT4, and INT2, activation values are uniformly mapped to the corresponding integer range based on their minimum and maximum values. In addition, INT4 uses stochastic rounding, which reduces systematic bias and makes the quantization error closer to unbiased in expectation. $\eta$ is mapped to an integer target bit-width between a minimum of 2 bits (INT2) and the native element width (FP32), thereby enabling the transition from continuous compression control to discrete quantization levels.

\textbf{\emph{LLM.int8}}-style compression~\cite{dettmers2022llmint8} implements a mixed-precision payload: the outlier values are stored separately at outlier precision (FP16 or INT8 with a scalar scale), while the remaining mass is quantized to INT8, INT4, or INT2 using row-wise scaling. We still use a PackBits-style mask to transmit the outlier indices, which introduces a fixed transmission overhead. Within a given experiment, the precision pair for outlier and regular values remains fixed, while the outlier ratio serves as the control parameter for the compression level. This compressor requires an offline-measured mapping between each triplet of (outlier precision, regular precision, outlier ratio) and the resulting $\eta$. In the optimization algorithm, we still optimize only over $\eta$. During compression and decompression, the corresponding triplet is retrieved by table lookup using the selected $\eta$.

\subsection{Compression Rate Selection Algorithms}
\label{appendix:algorithms}

In this section, we provide the detailed formulations and underlying mechanics for the established baseline policies and our proposed optimization frameworks. 

\subsubsection{Single-Task}
For isolated single-task environments, the optimization objective is restricted to finding the optimal data compression ratios along the inference pipeline. We evaluate the following policies:

\begin{description}
  \item[\algMAX{}] 
  Operating as a strict lower bound for network delay, this baseline forcefully applies the absolute minimum admissible compression factor \(\eta_i = \eta_i^{\min}\) across all communication links. While it aggressively minimizes transmission latency, it establishes the empirical lower bound on downstream inference accuracy.

  \item[\algNONE{}] 
  Prioritizing maximum inference accuracy without regard for network congestion, this policy transmits the intermediate layer representations entirely uncompressed (\(\eta_i=1\)) over every link. It serves as an upper bound for accuracy at the severe risk of violating throughput constraints in bandwidth-limited environments.

  \item[\algUNI{}] 
  This CSI-aware baseline enforces a strictly homogeneous compression rate \(\eta\) across all transmission hops along the pipeline. Because the rate must satisfy the end-to-end constraint uniformly, it is inherently bottlenecked by the most restrictive instantaneous link capacity, taking the form:
  \begin{equation}
      \eta \;=\; \max \left(\,\max_{i}\eta_i^{\min},\; \min\left(1,\;\min_{i}\frac{c_i(t)}{R(t)\,a_i}\right)\right).
  \end{equation}

  \item[\algMYO{}] 
  This No-CSI policy employs the exact CSI-aware optimal formulation but operates under the naive assumption that the network is entirely static. Specifically, it relies on the most recent channel observation as a deterministic proxy for current capacity, mapping \(\hat{c}_i(t) = c_i(t-1)\) at every time slot.

  \item[\algCONS{}] 
  Sharing the optimal formulation of the myopic policy, this No-CSI baseline adopts a highly pessimistic approach to channel estimation. To safeguard against transient network collapses, it substitutes \(\hat{c}_i(t)\) with the running minimum of historically observed capacities: \(\hat{c}_i(t) = \min_{\tau < t} c_i(\tau)\).

  \item[\algMA{}] 
  To mitigate the impact of high-frequency channel noise, this No-CSI baseline utilizes the myopic optimal policy but replaces the true channel capacity with a moving average of past observations spanning a finite temporal window \(W\), such that \(\hat{c}_i(t) = \frac{1}{W} \sum_{w=1}^W c_i(t-w)\).

  \item[\algOURSCSI{}]
  Our proposed framework computes the exact closed-form solution to the single-task program. It assumes oracle access to perfect channel knowledge (\(c_i(t)\)), yielding the theoretical Pareto frontier between delay and accuracy.

\item[\algOURSNOCSIst{}] Our proposed solver relies entirely on causal channel estimates. To robustly handle channel uncertainty, it employs a Lower Confidence Bound (LCB) estimator for the network capacity, formulated as:
\begin{equation}
    \hat{c}(t) = \hat{\mu}_c(t) - \beta \hat{\sigma}_c(t)
\end{equation}
where $\hat{\mu}_c(t)$ and $\hat{\sigma}_c(t)$ represent the empirically predicted mean and standard deviation of the channel state at time $t$, and $\beta$ is a tunable parameter controlling the conservatism of the estimate. By deliberately under-predicting the available capacity, this pessimistic estimator creates a safety buffer. The solver leverages this within our Estimated Stochastic Dual Descent formulation to iteratively optimize compression rates, strictly satisfying long-term QoS constraints even under unpredictable network volatility.
\end{description}

\subsubsection{Multi-Task}
In multi-task environments, resource contention introduces a coupled optimization challenge necessitating policies that jointly govern both physical resource allocation partitions (\(s^{\text{comp}}, s^{\text{comm}}\)) and task-specific data compression rates (\(\eta_{i,k}\)):

\begin{description}
  \item[\algMAX{}] 
  This baseline strictly enforces uniform partitioning of compute and bandwidth resources among all active tasks on shared nodes and links (\(s_{i,k}^{\text{comp}} = s_{i,k}^{\text{comm}} = 1/|\mathcal{K}|\)). Concurrently, it enforces maximum compression \(\eta_{i,k}=\eta_{i,k}^{\min}\) universally to aggressively minimize queueing and transmission delays.

  \item[\algNONE{}] 
  Maintaining the uniform resource partitioning scheme (\(s_{i,k}^{\text{comp}} = s_{i,k}^{\text{comm}} = 1/|\mathcal{K}|\)), this baseline transmits all intermediate data completely uncompressed (\(\eta_{i,k}=1\)), establishing an upper bound on multi-task accuracy while largely ignoring shared link congestion.

  \item[\algEQ{}] 
  This policy enforces a fair-share uniform resource partition across all tasks. Given these static, isolated resource slices, the compression factors \(\eta_{i,k}\) are independently optimized using the single-task CSI-aware optimal mapping.

  \item[\algPROP{}] 
  Rather than splitting resources equally, this CSI-aware baseline allocates fractions proportional to the inherent architectural demands of each task. Specifically, compute shares are allocated relative to base processing loads (\(s_{i,k}^{\text{comp}} \propto \tau_{i,k}\)), and bandwidth shares are allocated relative to uncompressed activation dimensions (\(s_{i,k}^{\text{comm}} \propto a_{i,k}\)). Task compression rates subsequently follow the single-task optimal formulation over these weighted partitions.

  \item[\algPRIO{}] 
  This baseline implements a greedy resource scheduling algorithm governed strictly by the application-level utility weights \(w_k\). The scheduler first secures the minimum-feasible resource allocations required to keep all active tasks viable. Subsequently, it funnels all residual bandwidth and compute capacity entirely to the tasks with the highest utility priorities to maximize their respective accuracy (\(\eta_{k}\)).

  \item[\algDES{}] 
  Operating without explicit channel state information, this policy partitions resources equally at every time slot, effectively decoupling the optimization landscape. Each task \(k\) independently executes its local stochastic compression subproblem (Alg. 1) using its dedicated effective channel capacity \(\hat{c}_{i,k}^{\text{eff}}(t) = s_{i,k}^{\text{comm}}\hat{c}_i(t)\) and effective processing time \(\tau_{i,k}^{\text{eff}} = \tau_{i,k} / s_{i,k}^{\text{comp}}\).

  \item[\algDLPROP{}] 
  Operating under the same decoupled framework as \algDES, this No-CSI baseline dynamically modulates the resource allocations proportionally to the current values of the  dual variables, \(s_{i,k} \propto \lambda_k(t)\). This acts as a feedback-driven triage mechanism, explicitly prioritizing resources for tasks that are currently struggling to meet their required throughput thresholds.

  \item[\algMA{}] 
  This predictive baseline constructs a historical temporal mean of the channel capacities. It subsequently feeds these smoothed deterministic estimates into the exact CSI-aware multi-task problem (Problem 2), jointly optimizing the resource splits and compression rates under the assumption that the average reflects the true environment.

  \item[\algOURSCSI{}]
  Assuming perfect instantaneous knowledge of the shared channel capacities, this method solves the convex optimization program (Problem 2), performing optimal multi-dimensional water-filling across all tasks and network links.

  \item[\algOURSNOCSImt{}]
  Our comprehensive multi-task framework leverages block coordinate descent to alternatingly optimize resource allocation variables and intermediate compression variables on the fly. It uses an LCB channel estimator and dynamically balances fair resource contention with individual task QoS requirements.
\end{description}

\section{Testbed Details}
\label{app:testbed_details}

We ran our experiments on the following three testbeds:

\paragraph{Raspberry Pi} We deploy the system on a cluster of seven Raspberry Pi~5 nodes (2.4\,GHz quad-core ARM, 8\,GB RAM), interconnected over 802.11ac Wi-Fi. The inference model is Gemma~2B (18 hidden layers), distributed across the cluster as a pipeline: a dedicated user-equipment node handles embedding and sampling, while the 17 transformer layers are partitioned across the seven Pis. Compression is applied at the seven inter-node cut points where activation tensors cross the wireless link. Multiple requests are served concurrently via pipeline parallelism, each identified by a unique tag routed to per-request KV caches. A central controller on a separate machine periodically collects per-node metrics, runs the optimizer, and broadcasts updated per-link compression ratios without pausing inference. \CRedit{The implementation is available at \url{https://github.com/neu-spiral/rasp_compression}.}

\paragraph{NVIDIA Jetson} Our Jetson-based testbed consists of four NVIDIA Jetson Orin Nano Developer Kits interconnected through a shared 802.11ac Wi-Fi access point. Each node is equipped with a compact Ampere-architecture GPU, a 6-core Arm Cortex-A78 CPU, and 8 GB of 128-bit LPDDR5 memory. Inference is executed on the GPU, while compression, decompression, and other computations run on the CPU. Nodes communicate over HTTP, and bandwidth is controlled at the application layer using a token-bucket mechanism, yielding only coarse-grained control in practice. On this platform, we evaluate both a vision model, ResNet-56 on CIFAR-10 and a language model Flan-T5-Base on SST-2, each partitioned into pipeline stages across the nodes to balance the workload. In both cases, performance is measured by classification accuracy. \CRedit{The corresponding codebase can be accessed at \url{https://github.com/UIC-Networking-Research-Lab/CAMDI_RC_MobiHoc2026_Fitting_and_Jetson}.}

\paragraph{PRESCIENT} An experiment framework deployed on PRESCIENT, a GPU-dense network research testbed, converts pipeline and topology specifications along with scenario and algorithmic configurations into experiments that deploy as ``pods'' in PRESCIENT's Kubernetes environment. For each experiment, pods consist of an orchestrator and a set of nodes each running a stage of one (single-task) or more (multi-task) pipelines. The orchestrator runs the optimizers, manages configuration updates on node pods, and generates and measures pipeline tasks over the lifetime of an experiment. \CRedit{The testbed codebase is open-sourced at \url{https://github.com/PRESCIENT-osu/DNN-comm-compression}.}

We elaborate on each of the three testbed implementations below.

\subsection{Raspberry Pi}

\subsubsection{Hardware and Network}

The testbed consists of seven Raspberry Pi~5 single-board computers, each equipped with a 2.4\,GHz 64-bit quad-core ARM Cortex-A76 CPU and 8\,GB of LPDDR4X RAM. All nodes connect to a shared 802.11ac Wi-Fi access point. An eighth device (laptop) on the same network hosts the central controller process. All inference is CPU-bound; no GPU acceleration is available. Inter-node Wi-Fi links are subject to contention, fading, and interference from other devices on the network. 

\subsubsection{Model and Layer Assignment}

The inference model is Gemma~2B, a decoder-only transformer with $L = 18$ hidden layers. A dedicated user-equipment (UE) node hosts layer 0 and the remaining 17 transformer layers are spread across seven Pis: layers 1--3 on the first node, layers 4--5 on the second, layers 6--8 on the third, layers 9--10 on the fourth, layers 11--12 on the fifth, layers 13--15 on the sixth, and layers 16--17 on the seventh.

\subsubsection{Cut Points and Pipeline Parallelism}

The model has seven inter-node cut points, located at the output of the last hosted layer on each node: layers 3, 5, 8, 10, 12, and 15. At each cut point the output activation tensor is serialized, compressed according to the optimizer-assigned ratio $\eta_i(t)$, and transmitted to the downstream node over Wi-Fi. The ten remaining layer boundaries are intra-node transfers carried over local memory without compression. Multiple inference requests are served concurrently. Each request carries a unique identifier that routes it to a dedicated KV cache at every pipeline node. Requests are injected with a stagger delay, so that at steady state every node processes a different request's token simultaneously. 

\subsubsection{Central Controller}

A lightweight controller process runs on a laptop connected to the same Wi-Fi network, operating on a periodic schedule. At each control step, the controller first collects per-node metric reports delivered over UDP, including the computation latency, link throughput, and uncompressed activation size from every pipeline node. It then invokes the selected optimization algorithm to determine the compression ratio for each inter-node link. Finally, it broadcasts the updated compression ratio and codec identifier to every node via UDP. Inference proceeds continuously and the pipeline does not pause between control updates.

\subsubsection{Autoregressive Profiling}

\CRedit{The Raspberry Pi scenario is autoregressive: a one-time prefill over the $[S,d]$ prompt
followed by repeated single-token decode steps. Rather than model these two passes
separately, we use a fixed compute/communication profile measured on the same full-sequence
sample. The activation size $a_i$ is profiled on the full-sequence $[S,d]$ prefill activation
($S\approx512$, $d=4096$, FP16, so $a_i\approx4$\,MB), and the per-stage time $\tau_i$ is the
decompress, forward, and compress time for that same sequence, not a per-token decode. Both
therefore refer to a single full-sequence pass. Since we do not model the heterogeneous
prefill and decode passes separately, this request-level profile is approximate for full
generation.}

\subsection{NVIDIA Jetson}
\label{app:testbed_jetson}

\subsubsection{Setup}

Our Jetson-based testbed consists of four Jetson Orin Nano Developer Kits. Each node is an edge computing device with a compact Ampere-architecture GPU, a 6-core Arm Cortex-A78 CPU, and 8GB of 128-bit LPDDR5 memory. CUDA-supported workloads, such as model inference, are executed on the GPU, while non-CUDA components, such as compression/decompression and optimization, are executed on the CPU. All nodes are connected through a shared 802.11ac Wi-Fi access point. Each node runs an HTTP service, and inter-node communication is carried out over HTTP. Bandwidth is controlled at the application layer on each node using a token-bucket mechanism to avoid exceeding the hardware bandwidth limit of the access point.  As this control is implemented above the network stack, the resulting bandwidth shaping is not precise and is susceptible to environmental interference.

\subsubsection{Models and Datasets}

We use two models and their corresponding datasets on the Jetson platform. For the vision task, we use ResNet-56 on CIFAR-10. The ResNet-56 model we used is a CIFAR-10-specific variant whose architecture is adapted to the input size of CIFAR-10 samples. ResNet-56 is split into four pipeline stages: The four model partitions are deployed sequentially across the four nodes, hosting the stem and layer1.0--layer1.6, layer1.7--layer2.3, layer2.4--layer3.1, and layer3.2--layer3.8 plus the classification head, respectively. This partition is determined based on the computational cost of each layer to balance the inference workload across nodes.
For CIFAR-10, we use only the test set, which is divided into 100 batches, thus each batch consists of 100 images. Evaluation is based on classification accuracy; For the language task, we use Flan-T5-Base on SST-2. Flan-T5-Base follows an encoder-decoder architecture and is also split into four stages: Node A hosts encoder layers 0--3; Node B hosts encoder layers 4--7; Node C hosts encoder layers 8--11; and Node D hosts the full decoder together with the final label decision. For SST-2, we use only the first 100 samples from the test set and divide them into 100 batches, each containing one sample. Since SST-2 is a binary classification task, we also use classification accuracy as the evaluation metric. Under SST-2, the output length is only one token. In addition, all SST-2 inputs are padded to a fixed length of 128 tokens to avoid truncation and to keep the inference workload consistent across samples. Note that, for ResNet-56 on CIFAR-10, \emph{LLM.int8} uses the FP16+INT4 precision combination, whereas for Flan-T5 on SST-2, it uses FP16+INT8.

\subsection{PRESCIENT}

PRESCIENT (Platform for RESearch on Intelligence in Edge-enabled Next-generation Telecommunications) is a wireless network research platform at the Ohio State University. The platform is designed for over-the-air and emulated wireless experiments at scale, providing researchers with programmable radio hardware, dense GPU-accelerated compute nodes, and a Kubernetes-managed software infrastructure for experiment orchestration. 

\newlength{\figscale}

\begin{figure}[t]
\centering
\setlength{\figscale}{0.49\textwidth}
\includegraphics[width=0.46\textwidth, trim=0 0 0 0, clip]{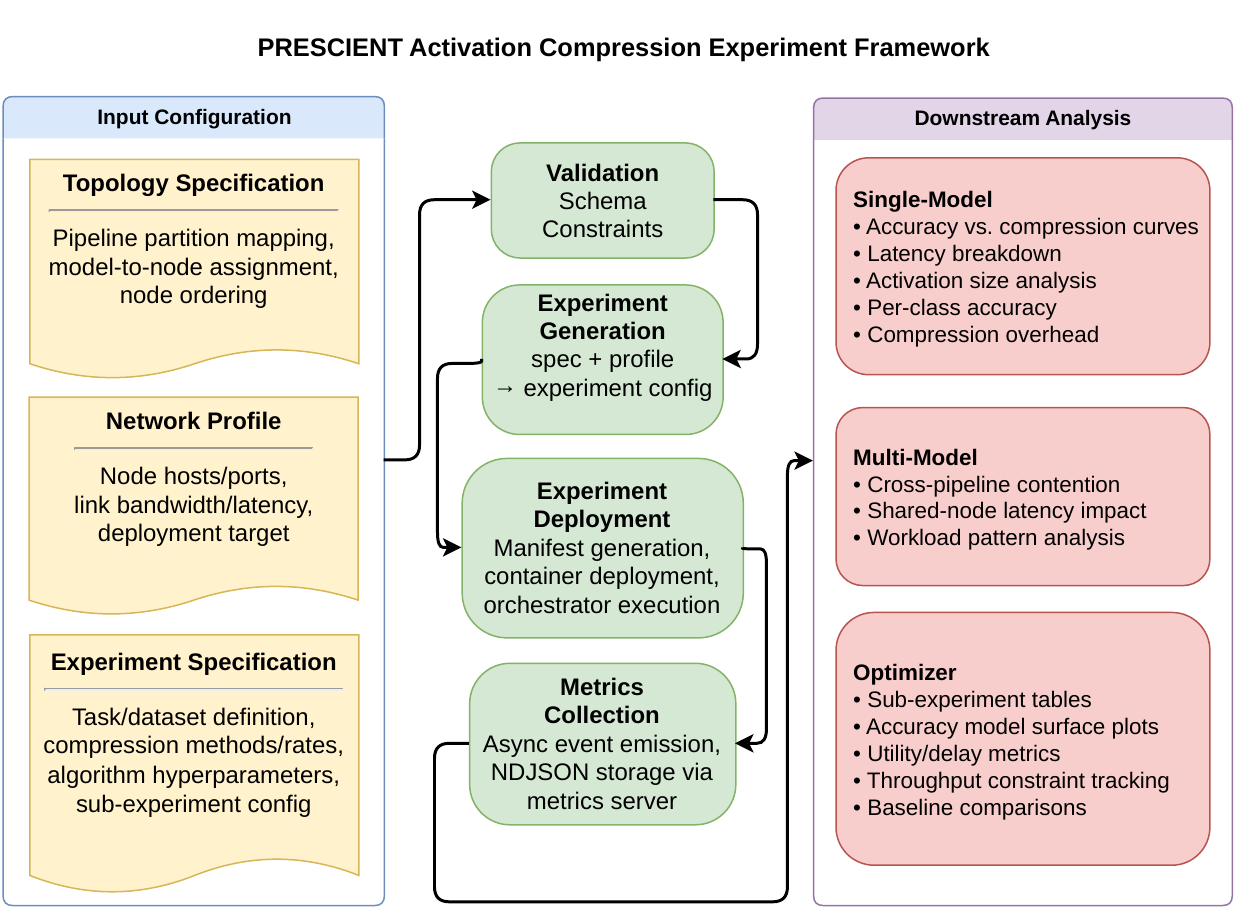}%
\hfill
\includegraphics[width=\figscale, trim=0 0 0 0, clip]{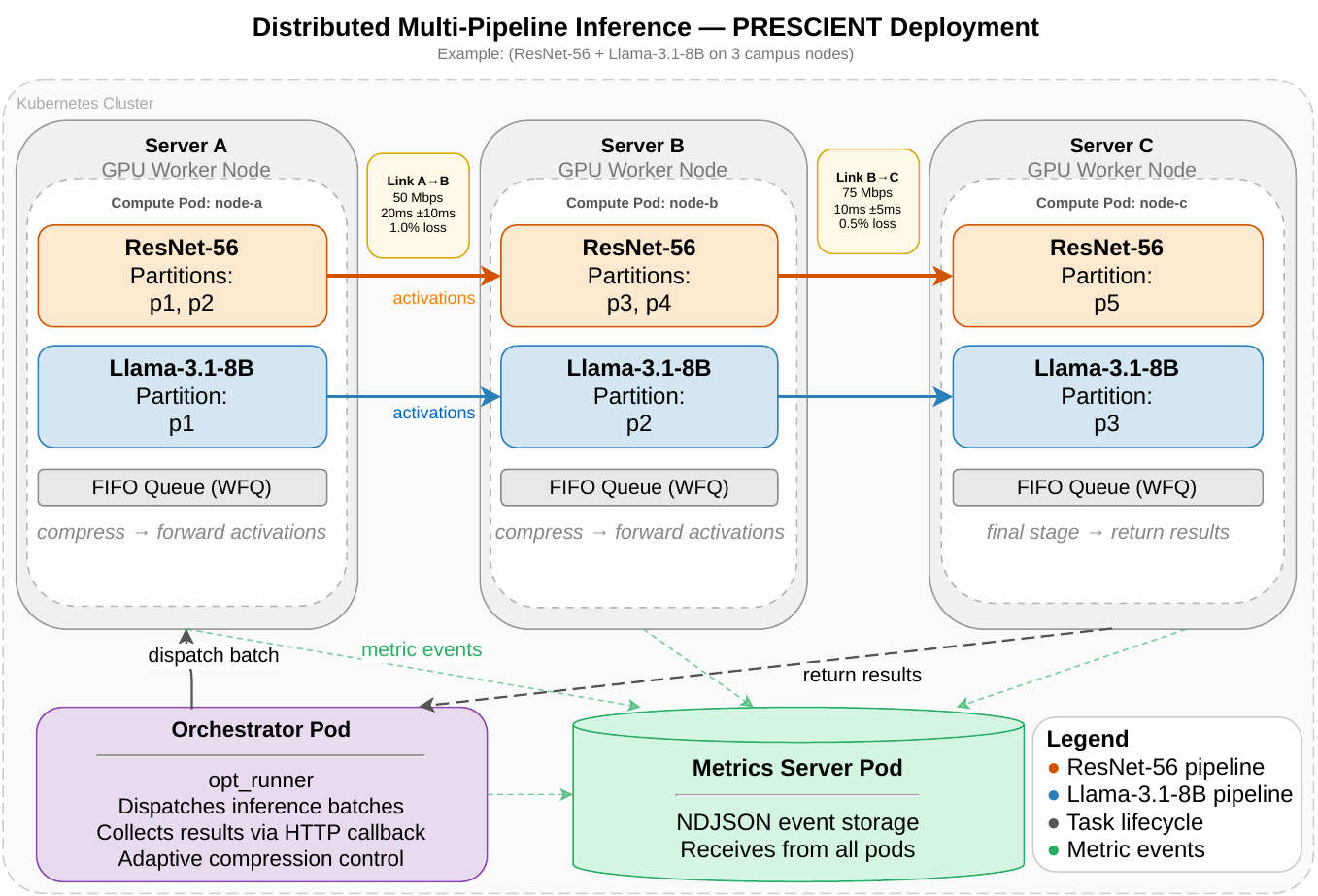}
\caption{(Left) Experiment framework pipeline from input YAML specifications
  through validation, generation, deployment, and metrics collection to
  downstream analysis.  (Right) Kubernetes deployment of a multi-pipeline
  optimizer experiment (ResNet-56 + Llama-3.1-8B) across three campus
  servers with TC-emulated inter-pod links.}
\label{fig:framework-overview}
\end{figure}

\subsubsection{Hardware Configuration}

\begin{itemize}
  \item \textbf{4$\times$ multi-GPU nodes.}  Multi-GPU x86 nodes equipped with either Nvidia's H100 or L40s GPUs. A single NVIDIA GH200 Grace Hopper (ARM64) server that supports Nvidia's GPU-accelerated AERIAL RAN research platform.
  \item \textbf{4$\times$ USRP X410 software-defined radios.}  High-performance SDR platforms capable of wideband transmission; used for over-the-air 5G NR experiments.
  \item \textbf{5$\times$ USRP N300 software-defined radios.}  Mid-range SDR platforms used as additional radio access points or for channel sounding.
  \item \textbf{User Equipment (UEs).} COTS UEs attached to the platform for end-to-end RAN experiments.
\end{itemize}

\subsubsection{Experimentation Framework}

The platform is managed as a \textbf{Kubernetes} cluster with each server accessible as a worker node.
Experiment workloads are packaged as containers and scheduled across nodes via standard Kubernetes manifests.
Network policies, resource quotas, and inter-pod connectivity can either be expressed through the Kubernetes control plane or deployed per node for experiments requiring higher granularity or specific controls.

An experimentation framework developed to accelerate experiment specification, deployment, and data analysis runs on PRESCIENT and is responsible for generating numerous experiment configurations given topology, profile and algorithmic configuration templates. The framework generates experiments based on scenario definition and for each optimizer (baselines and ours), runs the optimization loop for \textbf{T} slots. Profiling is performed once per scenario and cached for re-use later by experiments based on the same scenario.

\begin{algorithm}[htb]
\caption*{\textbf{PRESCIENT Experiment Framework - Optimization Algorithm}}\label{alg:opt-loop}
\begin{algorithmic}[1]

\Statex \textbf{Input:}
$K$ pipelines with throughput targets $R_k$, weights $w_k$;
profiled latencies $\tau_i^{(k)}$, activation sizes $a_\ell^{(k)}$;
compression bounds $[\eta_\ell^{\min}, \eta_\ell^{\max}]$ per link;
surrogate accuracy models $\hat{A}_k(\boldsymbol{\eta}_k)$;
channel estimator $\mathcal{E}$;
slots $T$, batches $B$, step size $\mu$, probe interval $P$

\Statex \textbf{Output:}
per-slot $\boldsymbol{\eta}_k(t)$, $\mathbf{s}(t)$, $\hat{R}_k(t)$

\Statex

\State Initialize $\lambda_k \gets 0$ for all $k$
\For{$t = 0, 1, \dots, T-1$}

\If{$t \bmod P = 0$}
  \State Probe links $\to \mathbf{c}_t$; \quad update $\mathcal{E}$ with $\mathbf{c}_t$
\EndIf
\State $\hat{\mathbf{c}}_t \gets \mathcal{E}.\mathrm{estimate}(t)$

\Statex \hspace{\algorithmicindent} \textit{// Primal step}
\State $\displaystyle \{\boldsymbol{\eta}^*,\, \mathbf{s}^*\} \gets
  \arg\max_{\boldsymbol{\eta},\,\mathbf{s}}
  \sum_{k=1}^{K} \bigl[ w_k \hat{A}_k(\boldsymbol{\eta}_k)
  - \lambda_k \, d_k(\boldsymbol{\eta}_k, \mathbf{s}_k, \hat{\mathbf{c}}_t) \bigr]$
\Statex \hspace{\algorithmicindent}\hspace{\algorithmicindent}
  s.t.\ \ $\eta_{k,\ell} \in [\eta_\ell^{\min},\, \eta_\ell^{\max}]$,
  \quad $\sum_k s_{k,i} = 1 \ \forall\, i$

\State Map each $\eta_{k,\ell}^*$ to compression method; push config to nodes
\State Run $B$ batches; observe achieved throughput $\hat{R}_k(t)$

\Statex \hspace{\algorithmicindent} \textit{// Dual update}
\For{each pipeline $k$}
  \State $\lambda_k \gets \max\!\bigl(0,\; \lambda_k + \mu \cdot \max(0,\; 1/\hat{R}_k - 1/R_k)\bigr)$
\EndFor

\EndFor

\end{algorithmic}

\hspace{2pt}
\noindent\small
\textbf{Preceding offline phases:}\;
\begin{enumerate}
  \item ~\emph{Profiling} --- run inference at $\eta = \eta^{\max}$ to measure $\tau_i^{(k)}$, $a_\ell^{(k)}$, and nominal link capacities.
  \item \emph{Accuracy model fitting} --- sweep $\boldsymbol{\eta}_k$ via simulation, fit surrogate $\hat{A}_k$.
\end{enumerate}
\end{algorithm}

Similarly, accuracy model oracles are trained offline once per a set of pipeline, dataset, and compression and cached for re-use by subsequent experiments. The training allows for multiple modes of sampling compression rates - along the diagonal and random sampling. For each compression rate sample, accuracy is evaluated over a set of batches of the pipeline model dataset.

\begin{algorithm}[hbt]
  \caption*{\textbf{PRESCIENT Experimentation Framework - Surrogate Accuracy Model Training}}
  \begin{algorithmic}[1]

  \Statex \textbf{Input:}
    pipeline $k$ with $L_k$ inter-node links, compression bounds $[\eta_\ell^{\min}, \eta_\ell^{\max}]$;
    sweep count $N_s$, sweep mode $\in \{\texttt{random}, \texttt{diagonal}\}$;
    dataset $\mathcal{D}_k$;
    compression scheme $s$;
    surrogate type (e.g.\ \texttt{poly3}, \texttt{gbm})

  \Statex \textbf{Output:}
    fitted surrogate $\hat{A}_k : \mathbb{R}^{L_k} \to [0,1]$

  \Statex

  \State Load simulation pipeline from stored partitions and dataset $\mathcal{D}_k$
  \State Configure compression scheme $s$ on all links

  \Statex

  \If{sweep mode $=$ \texttt{diagonal}}
    \State $\boldsymbol{\alpha} \gets \mathrm{linspace}(0, 1, N_s)$
    \For{$j = 1, \dots, N_s$}
      \State $\boldsymbol{\eta}^{(j)} \gets \boldsymbol{\eta}^{\min} + \alpha_j \, (\boldsymbol{\eta}^{\max} - \boldsymbol{\eta}^{\min})$
    \EndFor
  \Else    \hspace{55.0pt}\textit{// random}
    \For{$j = 1, \dots, N_s$}
      \State $\boldsymbol{\eta}^{(j)} \sim \mathrm{Uniform}[\boldsymbol{\eta}^{\min},\, \boldsymbol{\eta}^{\max}]$
    \EndFor
  \EndIf

  \Statex

  \For{$j = 1, \dots, N_s$}
    \State Apply compression at $\boldsymbol{\eta}^{(j)}$ under scheme $s$
    \State $y_j \gets$ evaluate task accuracy on $\mathcal{D}_k$ via simulation
  \EndFor

  \Statex

  \State Fit surrogate $\hat{A}_k$ on $\{(\boldsymbol{\eta}^{(j)},\, y_j)\}_{j=1}^{N_s}$

  \end{algorithmic}

  \vspace{4pt}
  \noindent\small
  \textbf{Diagonal} sweeps all links at the same relative compression level (useful for low-dimensional characterization);
  \textbf{random} covers the full $L_k$-dimensional hypercube.
  Surrogate types include polynomial ridge regression (\texttt{poly2}, \texttt{poly3}), gradient boosted trees (\texttt{gbm}), and MLPs.

\end{algorithm}
For running a complete experiment with multiple baseline and proposed optimizers, the choice of $R_k$, especially in a real, on-line deployment is important as it is easy to select one which is either trivially satisfied or is infeasible given communication overhead not captured during link bandwidth probing. The framework performs an off line $R_k$ selection that is feasible and appropriate for the network profile being emulated per experiment.

\paragraph{Infrastructure/network profiles.}
  
\begin{description}
  \item[\texttt{\netProfileEdge}.] Both inter-node links (A$\to$B and B$\to$C) are shaped to 100\,Mbps with 10ms of delay and 0.2\% packet loss.
    This profile represents a constrained wireless or cellular backhaul link and is the primary baseline for evaluating the benefit of adaptive compression.
  \item[\texttt{\netProfileLan}.] Both links are shaped to 1\,Gbps with no added delay.
    This profile models a high-bandwidth local-area or fronthaul network; compression overhead is minimal and end-to-end latency is dominated by compute.
  \item[\texttt{\netProfileWlan}.] Asymmetric link conditions modeling a wireless or mobile ad hoc network: the A$\to$B link is constrained to 50\,Mbps with 20\,ms propagation delay ($\pm$5\,ms jitter) and 0.5\,\% packet loss, while B$\to$C is constrained to 100\,Mbps with 5\,ms delay ($\pm$1\,ms jitter) and 0.1\,\% packet loss.
    This profile reflects a scenario in which the first pipeline stage operates at a remote or low-bandwidth site (e.g., an edge device in the field) while later stages are co-located in a slightly higher-capacity facility connected by a cellular backhaul link.
\end{description}

\begin{algorithm}[h]
  \caption*{\textbf{PRESCIENT Experimentation Framework - $R_k$ Selection Algorithm}}
  \begin{algorithmic}[1]

  \Statex \textbf{Input:}
    experiment configuration (pipelines, nodes, links, compression schemes $\mathcal{S}$)

  \Statex \textbf{Output:}
    per-pipeline throughput target $R_k$

  \Statex

  \State Run profiling at $\eta = \eta^{\max}$; measure $\tau_i^{(k)}$, $a_\ell^{(k)}$, link capacities $\mathbf{c}$
  \For{each scheme $s \in \mathcal{S}$}
    \State Run no-compression baseline ($\eta = \eta^{\max}$); record $\hat{R}_k^{\,\mathrm{none},s}$
    \State Run max-compression baseline ($\eta = \eta^{\min}$); record $\hat{R}_k^{\,\mathrm{max},s}$
  \EndFor

  \Statex

  \For{each pipeline $k$}
    \State $\underline{R}_k \gets \min_{s} \; \hat{R}_k^{\,\mathrm{none},s}$ \hfill\textit{// slowest no-compression run}
    \State $\overline{R}_k \gets \min_{s} \; \hat{R}_k^{\,\mathrm{max},s}$ \hfill\textit{// slowest max-compression run}
    \State $R_k \gets \underline{R}_k + \alpha \, (\overline{R}_k - \underline{R}_k)$, \quad $\alpha \in (0, 1)$
  \EndFor

  \end{algorithmic}

  \vspace{4pt}
  \noindent\small
  $\alpha$ controls difficulty: values near $0$ are easily met without compression; values near $1$ approach the physical throughput ceiling and risk infeasibility.
  A moderate $\alpha$ (e.g.\ $0.4$--$0.6$) ensures baselines that ignore channel or accuracy structure violate the constraint while leaving headroom for adaptive optimizers.

  \end{algorithm}

\subsubsection{Deployment}

The distributed model pipeline is realized as a set of Kubernetes pods scheduled across the available compute nodes with each pod provisioned access to an Nvidia GPU.
Each pod runs a container that hosts one or more consecutive model partitions; activations are transmitted between pods over the cluster network.
Network emulation via \texttt{tc/netem} is applied at each link (pod egress and ingress) to instantiate the infrastructure profiles described below.
The optimizer experiments additionally use a GPU provisioned orchestrator pod for Stein-oracle gradient estimation, which requires running multiple forward passes through the simulated model in parallel for each optimization time slot. Fig.~\ref{fig:framework-overview} shows the functional components of the framework and its workflow while also visualizing an example scenario experiment deployed on PRESCIENT.

\section{Additional Experiments}
\label{app:additional_exp}

\subsection{Testbed Experiments}

\begin{table*}[t]
\centering
\caption{Testbed results for single-task (upper block) and multi-task (lower block) settings including the CSI-aware methods. We report aggregate utility $\metU$, average delay $\metD$, and excess delay $\metED$, across the evaluated methods, grouped into reference baselines, CSI-aware methods, and No-CSI methods. Compression schemes $\cmpT$, $\cmpQ$, and $\cmpIEight$ denote Top-$k$, quantization, and LLM.int8, respectively. Infeasible entries ( $\metED/\metD > 0.05$) are shaded light red. Bold and underlined values denote the best and second-best accuracies, respectively, among feasible entries within the CSI-aware and No-CSI groups.}
\label{tab:results-testbed-all}

\medskip
\hrule
\smallskip
{\bfseries Single-task Scenarios}
\smallskip
\hrule
\medskip

{\scriptsize
\setlength{\tabcolsep}{2pt}
\resizebox{\textwidth}{!}{%
\begin{tabular}{ll c ccc !{\vrule width 1pt} ccc ccc !{\vrule width 1pt} ccc ccc ccc ccc}
\toprule
\multirow{3}{*}{\textbf{Scenario}} & \multirow{3}{*}{\textbf{Metric}}
 & \multicolumn{4}{c!{\vrule width 1pt}}{\textbf{Reference}}
 & \multicolumn{6}{c!{\vrule width 1pt}}{\textbf{CSI-aware}}
 & \multicolumn{12}{c}{\textbf{No-CSI}} \\
\cmidrule(lr){3-6}\cmidrule(lr){7-12}\cmidrule(lr){13-24}
 & & \multirow{2}{*}{\algNONE[short]} & \multicolumn{3}{c!{\vrule width 1pt}}{\algMAX[short]}
   & \multicolumn{3}{c}{\algUNI[short]} & \multicolumn{3}{c!{\vrule width 1pt}}{\algOURSCSI[short]}
   & \multicolumn{3}{c}{\algMYO[short]} & \multicolumn{3}{c}{\algCONS[short]} & \multicolumn{3}{c}{\algMA[short]} & \multicolumn{3}{c}{\algOURSNOCSIst[short]} \\
\cmidrule(lr){4-6}\cmidrule(lr){7-9}\cmidrule(lr){10-12}\cmidrule(lr){13-15}\cmidrule(lr){16-18}\cmidrule(lr){19-21}\cmidrule(lr){22-24}
 & & & \cmpT & \cmpQ & \cmpIEight & \cmpT & \cmpQ & \cmpIEight & \cmpT & \cmpQ & \cmpIEight & \cmpT & \cmpQ & \cmpIEight & \cmpT & \cmpQ & \cmpIEight & \cmpT & \cmpQ & \cmpIEight & \cmpT & \cmpQ & \cmpIEight \\
\midrule
\multirow{3}{*}{\scenTbSSGPT}
 & $\metU$  & \nfeas{1.00} & 0.331 & 0.242 & 0.552 & 0.657 & 0.449 & 0.671 & \underline{0.698} & 0.490 & \best{0.747} & \nfeas{0.686} & \nfeas{0.470} & \nfeas{0.742} & 0.607 & 0.443 & \underline{0.702} & \nfeas{0.718} & \nfeas{0.506} & \nfeas{0.798} & \nfeas{0.676} & 0.469 & \best{0.713} \\
 & $\metD$  & \nfeas{320} & 102 & 91.1 & 97.0 & 142 & 124 & 166 & \underline{155} & 137 & \best{159} & \nfeas{307} & \nfeas{285} & \nfeas{293} & 154 & 143 & \underline{162} & \nfeas{288} & \nfeas{265} & \nfeas{284} & \nfeas{86.0} & 81.1 & \best{128} \\
 & $\metED$ & \nfeas{120} & -97.9 & -109 & -103 & -58.0 & -75.9 & -33.8 & \underline{-44.8} & -63.1 & \best{-41.5} & \nfeas{107} & \nfeas{84.9} & \nfeas{92.7} & -45.7 & -57.3 & \underline{-38.0} & \nfeas{88.2} & \nfeas{65.1} & \nfeas{84.3} & \nfeas{16.5} & -119 & \best{-71.8} \\
\midrule


\multirow{3}{*}{\shortstack[l]{\scenPrsntVizWire}}
 & $\metU$  & 0.927 & 0.120 & 0.114 & 0.751 & 0.897 & 0.927 & \best{0.929} & \nfeas{0.899} & \underline{0.927} & \nfeas{0.934} & \nfeas{0.900} & \best{0.927} & \nfeas{0.934} & \nfeas{0.880} & \underline{0.927} & \nfeas{0.935} & \nfeas{0.907} & 0.927 & \nfeas{0.934} & 0.855 & 0.648 & 0.892 \\
 & $\metD$  & 236.334 & 226.023 & 169.156 & 221.468 & 235.713 & 178.333 & \best{230.703} & \nfeas{287.965} & \underline{175.316} & \nfeas{278.749} & \nfeas{285.990} & \best{175.103} & \nfeas{278.655} & \nfeas{287.929} & \underline{178.763} & \nfeas{278.887} & \nfeas{286.529} & 179.709 & \nfeas{278.993} & 183.209 & 183.649 & 187.401 \\
 & $\metED$ & 9.061 & -1.250 & -58.117 & -5.805 & 8.441 & -48.940 & \best{3.430} & \nfeas{60.692} & \underline{-51.957} & \nfeas{51.476} & \nfeas{58.718} & \best{-52.170} & \nfeas{51.382} & \nfeas{60.656} & \underline{-48.510} & \nfeas{51.615} & \nfeas{59.256} & -47.564 & \nfeas{51.720} & -44.064 & -43.624 & -39.871 \\
\midrule


\multirow{3}{*}{\shortstack[l]{\scenPrsntVizEdge}}
 & $\metU$  & \nfeas{0.927} & 0.120 & 0.114 & 0.751 & 0.563 & \best{0.928} & \underline{0.903} & \nfeas{0.551} & \nfeas{0.933} & \nfeas{0.908} & \nfeas{0.551} & \nfeas{0.933} & \nfeas{0.908} & \nfeas{0.543} & \nfeas{0.933} & \best{0.909} & \nfeas{0.546} & \nfeas{0.933} & \nfeas{0.909} & 0.533 & 0.684 & \underline{0.845} \\
 & $\metD$  & \nfeas{676.002} & 197.228 & 197.498 & 209.463 & 210.430 & \best{242.690} & \underline{241.475} & \nfeas{274.360} & \nfeas{259.501} & \nfeas{255.704} & \nfeas{285.748} & \nfeas{264.321} & \nfeas{255.146} & \nfeas{284.461} & \nfeas{268.449} & \best{248.806} & \nfeas{285.830} & \nfeas{263.583} & \nfeas{252.235} & 235.449 & 222.658 & \underline{205.917} \\
 & $\metED$ & \nfeas{437.907} & -40.867 & -40.597 & -28.632 & -27.665 & \best{4.595} & \underline{3.380} & \nfeas{36.265} & \nfeas{21.406} & \nfeas{17.609} & \nfeas{47.653} & \nfeas{26.226} & \nfeas{17.051} & \nfeas{46.366} & \nfeas{30.354} & \best{10.711} & \nfeas{47.735} & \nfeas{25.487} & \nfeas{14.139} & -2.646 & -15.437 & \underline{-32.178} \\
\midrule


\multirow{3}{*}{\shortstack[l]{\scenPrsntLangWikiWire}}
 & $\metU$  & \nfeas{1.000} & \nfeas{0.192} & 0.028 & 0.682 & \nfeas{0.995} & \best{1.000} & \nfeas{0.999} & \nfeas{1.000} & \nfeas{1.000} & \nfeas{1.000} & \nfeas{1.000} & \nfeas{1.000} & \nfeas{1.000} & \nfeas{1.000} & \nfeas{1.000} & \nfeas{1.000} & \nfeas{1.000} & \nfeas{1.000} & \nfeas{1.000} & \underline{0.639} & 0.066 & \best{0.942} \\
 & $\metD$  & \nfeas{296.553} & \nfeas{281.692} & 236.484 & 273.836 & \nfeas{294.012} & \best{254.720} & \nfeas{284.012} & \nfeas{358.007} & \nfeas{329.074} & \nfeas{321.507} & \nfeas{338.583} & \nfeas{294.891} & \nfeas{307.931} & \nfeas{338.732} & \nfeas{287.429} & \nfeas{316.593} & \nfeas{343.194} & \nfeas{294.311} & \nfeas{331.681} & \underline{262.093} & 260.805 & \best{236.245} \\
 & $\metED$ & \nfeas{33.395} & \nfeas{18.534} & -26.674 & 10.678 & \nfeas{30.854} & \best{-8.438} & \nfeas{20.855} & \nfeas{94.849} & \nfeas{65.916} & \nfeas{58.349} & \nfeas{75.425} & \nfeas{31.733} & \nfeas{44.773} & \nfeas{75.575} & \nfeas{24.271} & \nfeas{53.435} & \nfeas{80.036} & \nfeas{31.153} & \nfeas{68.523} & \underline{-1.065} & -2.353 & \best{-26.913} \\
\midrule


\multirow{3}{*}{\shortstack[l]{\scenPrsntLangWikiEdge}}
 & $\metU$  & \nfeas{1.000} & 0.192 & 0.028 & 0.677 & \underline{0.891} & \nfeas{0.013} & \best{0.981} & \nfeas{0.957} & \nfeas{0.018} & \nfeas{0.989} & \nfeas{0.957} & \nfeas{0.013} & \nfeas{0.989} & \nfeas{0.956} & \nfeas{0.013} & \nfeas{0.987} & \nfeas{0.957} & \nfeas{0.013} & \nfeas{0.989} & \underline{0.691} & 0.087 & \best{0.857} \\
 & $\metD$  & \nfeas{790.384} & 286.550 & 335.390 & 291.654 & \underline{355.782} & \nfeas{451.075} & \best{383.187} & \nfeas{460.639} & \nfeas{473.970} & \nfeas{437.899} & \nfeas{462.651} & \nfeas{457.479} & \nfeas{434.263} & \nfeas{453.741} & \nfeas{453.214} & \nfeas{437.973} & \nfeas{458.281} & \nfeas{448.635} & \nfeas{435.632} & \underline{337.721} & 314.389 & \best{343.769} \\
 & $\metED$ & \nfeas{420.014} & -83.820 & -34.981 & -78.716 & \underline{-14.589} & \nfeas{80.704} & \best{12.817} & \nfeas{90.269} & \nfeas{103.600} & \nfeas{67.528} & \nfeas{92.281} & \nfeas{87.109} & \nfeas{63.892} & \nfeas{83.370} & \nfeas{82.843} & \nfeas{67.602} & \nfeas{87.911} & \nfeas{78.265} & \nfeas{65.262} & \underline{-32.649} & -55.982 & \best{-26.601} \\
\midrule


\multirow{3}{*}{\shortstack[l]{\scenPrsntLangMmluEdge}}
 & $\metU$  & \nfeas{0.467} & 0.267 & 0.283 & 0.283 & \underline{0.387} & 0.300 & \nfeas{0.474} & \best{0.413} & \nfeas{0.300} & \nfeas{0.475} & 0.408 & \nfeas{0.300} & \nfeas{0.477} & \underline{0.417} & \nfeas{0.300} & \nfeas{0.481} & \best{0.417} & \nfeas{0.300} & \nfeas{0.478} & 0.269 & 0.282 & 0.397 \\
 & $\metD$  & \nfeas{901.130} & 274.393 & 280.272 & 297.253 & \underline{294.294} & 439.337 & \nfeas{819.023} & \best{390.254} & \nfeas{655.895} & \nfeas{867.009} & 385.275 & \nfeas{645.917} & \nfeas{874.369} & \underline{390.544} & \nfeas{651.298} & \nfeas{862.062} & \best{386.328} & \nfeas{655.560} & \nfeas{863.350} & 290.246 & 309.594 & 358.331 \\
 & $\metED$ & \nfeas{401.130} & -225.607 & -219.728 & -202.746 & \underline{-205.706} & -60.663 & \nfeas{319.023} & \best{-109.746} & \nfeas{155.895} & \nfeas{367.009} & -114.725 & \nfeas{145.917} & \nfeas{374.369} & \underline{-109.456} & \nfeas{151.298} & \nfeas{362.062} & \best{-113.672} & \nfeas{155.560} & \nfeas{363.350} & -209.754 & -190.406 & -141.669 \\
\midrule

\multirow{3}{*}{\shortstack[l]{\scenVisionJetsonTestbedS}}
 & $\metU$  & \nfeas{0.933} & 0.414 & 0.870 & 0.930 & -- & -- & -- & -- & -- & -- & 0.894 & 0.933 & 0.932 & 0.856 & \best{0.933} & 0.932 & 0.893 & \underline{0.933} & 0.932 & 0.894 & 0.925 & 0.930 \\
 & $\metD$  & \nfeas{3772} & 476 & 469 & 510 & -- & -- & -- & -- & -- & -- & 1955 & 1954 & 1960 & 1620 & \best{1780} & 1173 & 1941 & \underline{1942} & 1992 & 1900 & 1270 & 972 \\
 & $\metED$ & \nfeas{1878} & -1410 & -1430 & -1400 & -- & -- & -- & -- & -- & -- & 68.3 & 60.1 & 54.5 & -267 & \best{-114} & -733 & 53.4 & \underline{47.5} & 86.9 & 12.5 & -624 & -934 \\
\midrule

\multirow{3}{*}{\shortstack[l]{\scenLangJetsonTestbedS}}
 & $\metU$  & \nfeas{0.934} & 0.863 & 0.484 & 0.939 & -- & -- & -- & -- & -- & -- & \nfeas{0.939} & 0.486 & 0.939 & \best{0.942} & 0.486 & 0.939 & \nfeas{0.939} & 0.486 & 0.939 & 0.939 & 0.540 & \underline{0.939} \\
 & $\metD$  & \nfeas{237} & 98.1 & 95.5 & 96.2 & -- & -- & -- & -- & -- & -- & \nfeas{235} & 179 & 249 & \best{148} & 153 & 129 & \nfeas{230} & 178 & 243 & 149 & 95.9 & \underline{98.9} \\
 & $\metED$ & \nfeas{65.2} & -95.2 & -76.1 & -156 & -- & -- & -- & -- & -- & -- & \nfeas{41.5} & 7.86 & -3.39 & \best{-45.0} & -18.1 & -123 & \nfeas{36.2} & 6.60 & -9.04 & -44.0 & -75.7 & \underline{-153} \\
\bottomrule
\end{tabular}}}

\medskip
\hrule
\smallskip
{\bfseries Multi-task Scenarios}
\smallskip
\hrule
\medskip

{\tiny
\setlength{\tabcolsep}{1.5pt}
\resizebox{\textwidth}{!}{%
\begin{tabular}{ll c ccc !{\vrule width 1pt} ccc ccc ccc ccc !{\vrule width 1pt} ccc ccc ccc cccc}
\toprule
\multirow{3}{*}{\textbf{Scenario}} & \multirow{3}{*}{\textbf{Metric}}
 & \multicolumn{4}{c!{\vrule width 1pt}}{\textbf{Reference}}
 & \multicolumn{12}{c!{\vrule width 1pt}}{\textbf{CSI-aware}}
 & \multicolumn{12}{c}{\textbf{No-CSI}} \\
\cmidrule(lr){3-6}\cmidrule(lr){7-18}\cmidrule(lr){19-30}
 & & \multirow{2}{*}{\algNONE[short]} & \multicolumn{3}{c!{\vrule width 1pt}}{\algMAX[short]}
   & \multicolumn{3}{c}{\algEQ[short]} & \multicolumn{3}{c}{\algPROP[short]} & \multicolumn{3}{c}{\algPRIO[short]} & \multicolumn{3}{c!{\vrule width 1pt}}{\algOURSCSI[short]}
   & \multicolumn{3}{c}{\algDES[short]} & \multicolumn{3}{c}{\algDLPROP[short]} & \multicolumn{3}{c}{\algMA[short]} & \multicolumn{3}{c}{\algOURSNOCSImt[short]} \\
\cmidrule(lr){4-6}\cmidrule(lr){7-9}\cmidrule(lr){10-12}\cmidrule(lr){13-15}\cmidrule(lr){16-18}\cmidrule(lr){19-21}\cmidrule(lr){22-24}\cmidrule(lr){25-27}\cmidrule(lr){28-30}
 & & & \cmpT & \cmpQ & \cmpIEight & \cmpT & \cmpQ & \cmpIEight & \cmpT & \cmpQ & \cmpIEight & \cmpT & \cmpQ & \cmpIEight & \cmpT & \cmpQ & \cmpIEight & \cmpT & \cmpQ & \cmpIEight & \cmpT & \cmpQ & \cmpIEight & \cmpT & \cmpQ & \cmpIEight & \cmpT & \cmpQ & \cmpIEight \\
\midrule
\multirow{3}{*}{\scenPrsntVizVizWire}
& $\metU$  &  \nfeas{0.924}   &  \nfeas{0.094}   &  \nfeas{0.086}   &  \nfeas{0.758}   &  \nfeas{0.675}   &  \nfeas{0.924}   &  \nfeas{0.913}   &  \nfeas{0.672}   &  \nfeas{0.924}   &  \nfeas{0.912}   &  \nfeas{0.486}   &  \nfeas{0.524}   &  \nfeas{0.844}   &  \nfeas{0.094}   &  \nfeas{0.086}   &  \nfeas{0.758}   &  \nfeas{0.094}   &  \nfeas{0.086}   &  \nfeas{0.758}   &  \nfeas{0.094}   &  \nfeas{0.086}   &  \nfeas{0.758}   &  \nfeas{0.094}   &  \nfeas{0.086}   &  \nfeas{0.758}   &  0.109   &  \underline{0.109}   &  \best{0.759}  \\
 & $\metD$  &  \nfeas{430}   &  \nfeas{415}   &  \nfeas{376}   &  \nfeas{383}   &  \nfeas{417}   &  \nfeas{396}   &  \nfeas{411}   &  \nfeas{416}   &  \nfeas{387}   &  \nfeas{414}   &  \nfeas{365}   &  \nfeas{355}   &  \nfeas{402}   &  \nfeas{384}   &  \nfeas{355}   &  \nfeas{399}   &  \nfeas{390}   &  \nfeas{351}   &  \nfeas{394}   &  \nfeas{388}   &  \nfeas{343}   &  \nfeas{393}   &  \nfeas{393}   &  \nfeas{343}   &  \nfeas{396}   &  286   &  \underline{283}   &  \best{294}  \\
 & $\metED$  &  \nfeas{193}   &  \nfeas{163}   &  \nfeas{85.6}   &  \nfeas{98.9}   &  \nfeas{167}   &  \nfeas{126}   &  \nfeas{156}   &  \nfeas{165}   &  \nfeas{108}   &  \nfeas{162}   &  \nfeas{63}   &  \nfeas{43.1}   &  \nfeas{137}   &  \nfeas{101}   &  \nfeas{43.9}   &  \nfeas{131}   &  \nfeas{113}   &  \nfeas{34.9}   &  \nfeas{120}   &  \nfeas{109}   &  \nfeas{18.4}   &  \nfeas{120}   &  \nfeas{119}   &  \nfeas{20.2}   &  \nfeas{125}   &  -95.1   &  \underline{-101}   &  \best{-79.5}  \\

\midrule

\multirow{3}{*}{\shortstack[l]{\scenPrsntVizVizEdge}}
 & $\metU$  &  \nfeas{0.924}   &  0.094   &  0.086   &  0.758   &  \nfeas{0.541}   &  \nfeas{0.924}   &  \nfeas{0.895}   &  \nfeas{0.543}   &  \nfeas{0.924}   &  \nfeas{0.895}   &  \nfeas{0.432}   &  \nfeas{0.524}   &  \best{0.843}   &  \underline{0.094}   &  0.086   &  \nfeas{0.763}   &  0.094   &  0.086   &  0.763   &  0.094   &  0.086   &  \underline{0.763}   &  0.094   &  0.086   &  \nfeas{0.763}   &  0.389   &  0.325   &  \best{0.765}  \\
 & $\metD$  &  \nfeas{1397}   &  465   &  412   &  469   &  \nfeas{606}   &  \nfeas{531}   &  \nfeas{517}   &  \nfeas{599}   &  \nfeas{529}   &  \nfeas{509}   &  \nfeas{545}   &  \nfeas{540}   &  \best{486}   &  \underline{427}   &  437   &  \nfeas{495}   &  424   &  432   &  487   &  436   &  433   &  \underline{483}   &  426   &  428   &  \nfeas{493}   &  358   &  344   &  \best{369}  \\
 & $\metED$  &  \nfeas{1842}   &  -21.8   &  -128   &  -13.6   &  \nfeas{259}   &  \nfeas{110}   &  \nfeas{82.3}   &  \nfeas{246}   &  \nfeas{107}   &  \nfeas{65.8}   &  \nfeas{139}   &  \nfeas{127}   &  \best{20.3}   &  \underline{-98.9}   &  -78   &  \nfeas{37}   &  -104   &  -89.1   &  21.1   &  -80.6   &  -86.8   &  \underline{13.2}   &  -101   &  -95.8   &  \nfeas{34}   &  -236   &  -264   &  \best{-215}  \\

\midrule

\multirow{3}{*}{\shortstack[l]{\scenPrsntVizLangEdge}}
 & $\metU$  &  \nfeas{0.962}   &  0.153   &  0.0593   &  0.737   &  \underline{0.812}   &  \nfeas{0.483}   &  \nfeas{0.948}   &  \best{0.812}   &  \nfeas{0.483}   &  \nfeas{0.948}   &  \nfeas{0.56}   &  \nfeas{0.551}   &  \nfeas{0.879}   &  0.153   &  0.0593   &  0.737   &  0.153   &  0.0593   &  \underline{0.737}   &  0.153   &  0.0593   &  0.737   &  0.153   &  0.0593   &  0.737   &  0.45   &  0.232   &  \best{0.844}  \\
 & $\metD$  &  \nfeas{1350}   &  492   &  431   &  487   &  \underline{668}   &  \nfeas{816}   &  \nfeas{786}   &  \best{661}   &  \nfeas{849}   &  \nfeas{770}   &  \nfeas{756}   &  \nfeas{837}   &  \nfeas{879}   &  443   &  456   &  475   &  475   &  480   &  \underline{471}   &  471   &  449   &  483   &  474   &  452   &  487   &  388   &  380   &  \best{446}  \\
 & $\metED$  &  \nfeas{1368}   &  -349   &  -471   &  -360   &  \underline{2.23}   &  \nfeas{298}   &  \nfeas{239}   &  \best{-12.3}   &  \nfeas{364}   &  \nfeas{208}   &  \nfeas{178}   &  \nfeas{341}   &  \nfeas{424}   &  -448   &  -422   &  -384   &  -384   &  -374   &  \underline{-392}   &  -391   &  -436   &  -368   &  -384   &  -429   &  -360   &  -557   &  -574   &  \best{-441}  \\
 
\midrule

\multirow{3}{*}{\shortstack[l]{\scenPrsntVizLangWlan}}
 & $\metU$  &  \nfeas{0.947}   &  0.147   &  0.059   &  0.707   &  \underline{0.616}   &  \nfeas{0.411}   &  \nfeas{0.901}   &  \nfeas{0.619}   &  \nfeas{0.404}   &  \nfeas{0.901}   &  \nfeas{0.516}   &  \nfeas{0.0846}   &  \nfeas{0.845}   &  0.147   &  0.059   &  \best{0.707}   &  0.147   &  0.059   &  0.707   &  0.147   &  0.059   &  0.707   &  0.147   &  0.059   &  \underline{0.707}   &  0.184   &  0.0718   &  \best{0.716}  \\
 & $\metD$  &  \nfeas{6107}   &  803   &  962   &  1249   &  \underline{1438}   &  \nfeas{1909}   &  \nfeas{2136}   &  \nfeas{1487}   &  \nfeas{1872}   &  \nfeas{2150}   &  \nfeas{1855}   &  \nfeas{2455}   &  \nfeas{2369}   &  753   &  974   &  \best{1237}   &  766   &  940   &  1266   &  767   &  951   &  1252   &  769   &  968   &  \underline{1228}   &  865   &  996   &  \best{1225}  \\
 & $\metED$  &  \nfeas{9357}   &  -1250   &  -934   &  -358   &  \underline{19.3}   &  \nfeas{961}   &  \nfeas{1415}   &  \nfeas{118}   &  \nfeas{888}   &  \nfeas{1444}   &  \nfeas{854}   &  \nfeas{2054}   &  \nfeas{1881}   &  -1350   &  -909   &  \best{-382}   &  -1324   &  -976   &  -325   &  -1323   &  -956   &  -352   &  -1318   &  -922   &  \underline{-401}   &  -1126   &  -864   &  \best{-405}  \\
\midrule

\multirow{3}{*}{\shortstack[l]{\scenPrsntLangLangEdge}}
 & $\metU$  &  \nfeas{1.00}   &  0.211   &  0.0327   &  0.716   &  \nfeas{0.994}   &  \nfeas{0.0421}   &  \nfeas{1.00}   &  \nfeas{0.994}   &  \nfeas{0.0421}   &  \nfeas{1.00}   &  \nfeas{0.606}   &  \nfeas{0.516}   &  \nfeas{0.858}   &  \underline{0.211}   &  0.0327   &  \best{0.716}   &  0.211   &  0.0327   &  0.716   &  0.211   &  0.0327   &  \underline{0.716}   &  0.211   &  0.0327   &  0.716   &  0.594   &  0.0327   &  \best{0.898}  \\
 & $\metD$  &  \nfeas{1492}   &  554   &  519   &  518   &  \nfeas{769}   &  \nfeas{1134}   &  \nfeas{807}   &  \nfeas{764}   &  \nfeas{1151}   &  \nfeas{804}   &  \nfeas{747}   &  \nfeas{861}   &  \nfeas{793}   &  \underline{585}   &  507   &  \best{515}   &  573   &  488   &  512   &  577   &  491   &  \underline{510}   &  576   &  477   &  511   &  491   &  495   &  \best{518}  \\
 & $\metED$  &  \nfeas{1555}   &  -321   &  -390   &  -393   &  \nfeas{109}   &  \nfeas{841}   &  \nfeas{185}   &  \nfeas{98.8}   &  \nfeas{875}   &  \nfeas{180}   &  \nfeas{66.1}   &  \nfeas{294}   &  \nfeas{157}   &  \underline{-258}   &  -414   &  \best{-398}   &  -283   &  -453   &  -405   &  -274   &  -446   &  \underline{-409}   &  -277   &  -475   &  -407   &  -446   &  -438   &  \best{-393}  \\
\midrule

\multirow{3}{*}{\shortstack[l]{\scenMTJetsonTestbedM}}
 & $\metU$  & \nfeas{0.926} & 0.418 & 0.860 & 0.929 & -- & -- & -- & -- & -- & -- & -- & -- & -- & -- & -- & -- & 0.707 & \underline{0.915} & \nfeas{0.932} & \nfeas{0.836} & \best{0.915} & \nfeas{0.932} & \nfeas{0.928} & \nfeas{0.927} & \nfeas{0.933} & 0.747 & 0.913 & \nfeas{0.932} \\
 & $\metD$  & \nfeas{4330} & 1090 & 957 & 1680 & -- & -- & -- & -- & -- & -- & -- & -- & -- & -- & -- & -- & 2440 & \underline{1230} & \nfeas{2880} & \nfeas{3590} & \best{1210} & \nfeas{2670} & \nfeas{4880} & \nfeas{3990} & \nfeas{13800} & 2400 & 1450 & \nfeas{2730} \\
 & $\metED$ & \nfeas{1900} & -1350 & -1480 & -760 & -- & -- & -- & -- & -- & -- & -- & -- & -- & -- & -- & -- & 9.55 & \underline{-1210} & \nfeas{443} & \nfeas{1160} & \best{-1230} & \nfeas{234} & \nfeas{2450} & \nfeas{1550} & \nfeas{11400} & -30.3 & -986 & \nfeas{285} \\
\bottomrule
\end{tabular}}}
\end{table*}

Implementing CSI-aware methods and baselines on real testbeds is challenging. On the Jetson testbed, kernel-level limitations prevent accurate network control without external devices, so both traffic shaping and link probing must be done at the application layer, which is coarse and imprecise. In addition, the shared Wi-Fi medium is easily affected by interference. Therefore, we do not evaluate CSI-aware methods on Jetson.

PRESCIENT provides much stronger network control support, enabling CSI-aware experiments. In our testbed experiments in Table \ref{tab:results-testbed-all}, the \algOURSCSI{} optimizer is observed to perform poorly compared to the other CSI-aware baselines. This is in contrast to the results of the offline, simulated experiments where the \algOURSCSI{} is shown to be optimal. This under-performance is consistent across compression schemes. While it mostly outshines the baselines in the multi-task scenario, the improved throughput masks the overtly aggressive compression, especially compared to its CSI-aware counterparts, resulting in degraded utility.

This discrepancy reveals a fundamental mismatch in capacity measurement on our real testbed to that in the simulated environment. CSI-awareness in the PRESCIENT testbed is achieved via link probing in every optimization time slot ($t$) where a payload is transmitted and timed over each link ($\ell$) to produce a bandwidth measurement $c_\ell(t)$. This kind of na\"ive probing measures bulk throughput in isolation and systematically overestimates the effective capacity available to actual activation transfers, which are smaller (100 KiB– 16 MiB) and incur significant per-request HTTP overhead. Moreover, to ensure model pipeline stages are filled to maximize pipeline throughput, there will always be as many inference tasks in-flight, $W$, as there are stages and hence, activation transfers also need to contend with $W$ in-flight tasks. The testbed CSI measurement attempts to account for some of these losses -- a wire overhead factor $\omega_\ell \in [1.15, 1.5]$ corrects the activation size parameter to $a_\ell \cdot \omega_\ell$ and $c_\ell(t)/W$ is the contention-corrected capacity. So the delay model computes $$\frac{(a_\ell \cdot \omega_\ell) \cdot \eta_\ell}{\hat{c}_\ell / W}$$ which is a tighter measurement but still fails to capture second-order sensitivities that get amplified in the \algOURSCSI{} optimizer. 

Probing noise is significant under the \netProfileEdge{} and \netProfileWlan{} emulated network profiles. With jitter and packet loss, each 20 MiB probe produces a noisy $c_\ell(t)$. CSI-aware uses this instantaneous value directly with no smoothing — an upward noise draw produces an inflated $\eta$ for the entire slot. The \algUNI{} is partially shielded by this phenomenon as it computes $\eta = \min_\ell(\hat{c}\ell / R_k a_\ell)$, where even if one link's probe comes in high due to noise, the optimizers selects the lowest ratio across links, thus ignoring the upward outlier. If the bottle-necked link gets an upward noise draw it will affect the \algUNI{} decision, but other links will have excess capacity to absorb it. The no-CSI optimizers use channel estimators (moving average, LCB) that explicitly smooth this variance. \algOURSCSI{} aggressively adapts $\eta$ per link capacity thus lacking the cross-link dampening inherent to \algUNI{} or \algEQ{} optimizers. 

Additionally, the wire overhead factor ($\omega$) is calculated during a no-compression profiling run at maximum activation sizes and applied as a constant through all slots. Wire overhead at those activation sizes is a much smaller fraction of the payload than at lower $\eta$. HTTP connection setup, headers, and other such per-request costs can be more efficiently amortized at higher payloads but as $\eta$ shrinks the fixed cost fraction grows and the true overhead is significantly higher.

Mitigations, therefore, include probing with activation-realistic payloads and with multiple concurrent probe trains which aim to recreate channel contention proportional to the window-size required to keep the pipeline fully loaded ($W$). Thus, instrumenting robust CSI measurement on a real testbed that matches the simple delay model of the simulation environment can be non-trivial and provide a strong case for \algOURSNOCSIst{} and \algOURSNOCSImt{} optimizers for such deployments.

The Raspberry Pi testbed shares the same Wi-Fi interference sensitivity as the Jetson testbed, yet its full Linux kernel access and lightweight communication stack make CSI-aware operation feasible. Channel probing on this testbed operates in a regime where the probe-to-activation mismatch is minimal. Each non-terminal pipeline node runs a background thread that periodically opens a TCP connection to its downstream neighbor, transmits a fixed-size probe payload $B_{\mathrm{probe}}$ sized to match the decode-phase activation, and records the round-trip time $\Delta t_{\mathrm{probe}}$. The per-link capacity estimate is then $\hat{c}_i^{\mathrm{probe}}(t) = B_{\mathrm{probe}} / \Delta t_{\mathrm{probe}}$. Because the probe payload is deliberately matched to the hidden-state activation size $O_i$, the payload-dependent overhead discrepancy is eliminated. Furthermore, inter-node communication uses raw TCP sockets with a lightweight serialization protocol rather than HTTP, making per-request wire overhead negligible.

\begin{figure*}[!t]
    \centering
    \includegraphics[width=0.5\textwidth]{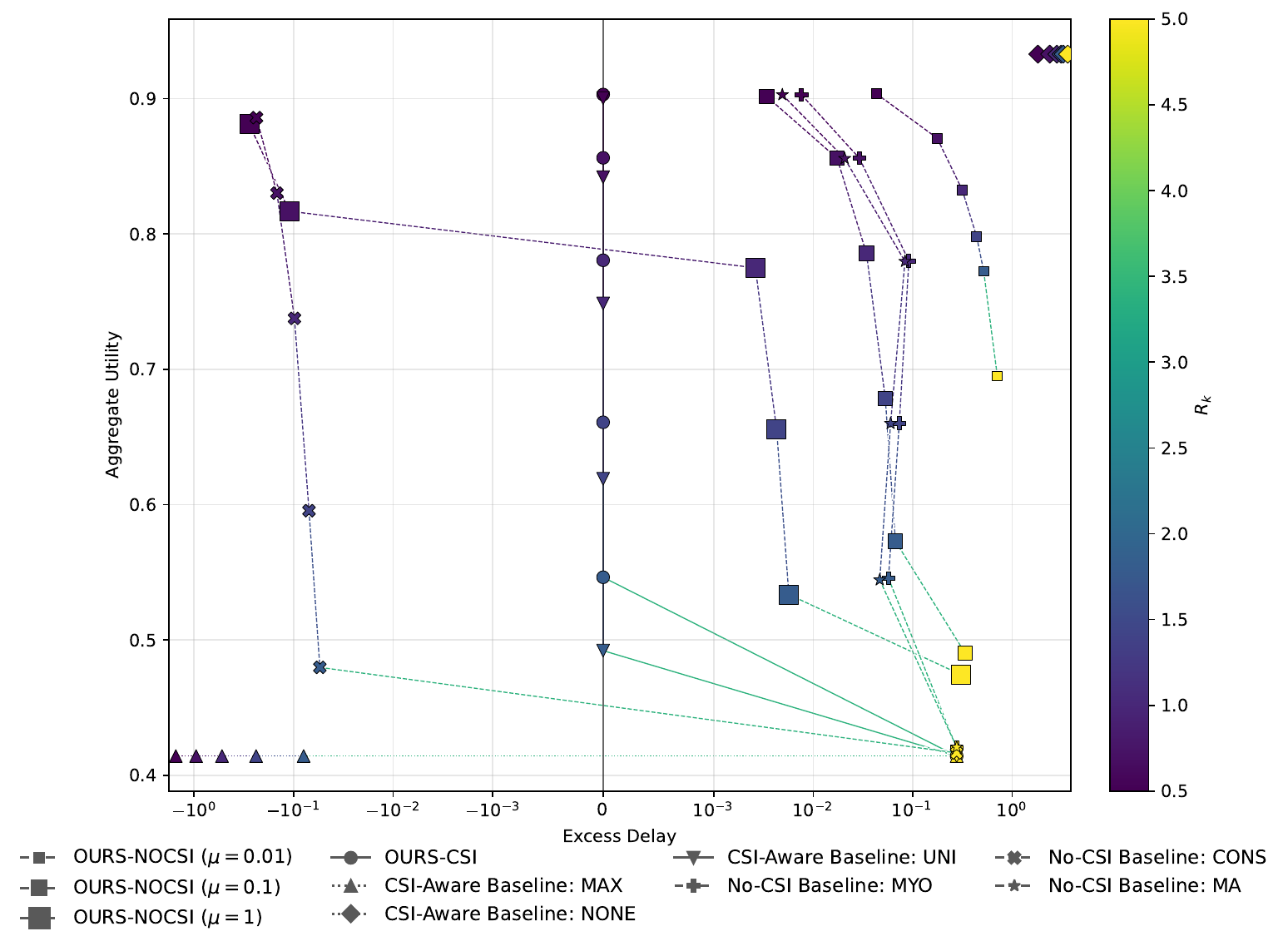}
    \caption{Scenario \scenVisionJetsonS, \taskVisionTwo single task results (Fitting-Based, Top-$k$)}
    \label{fig:topk-resnet-jetson}
\end{figure*}

\begin{figure*}[t]
    \centering
    \includegraphics[width=0.5\textwidth]{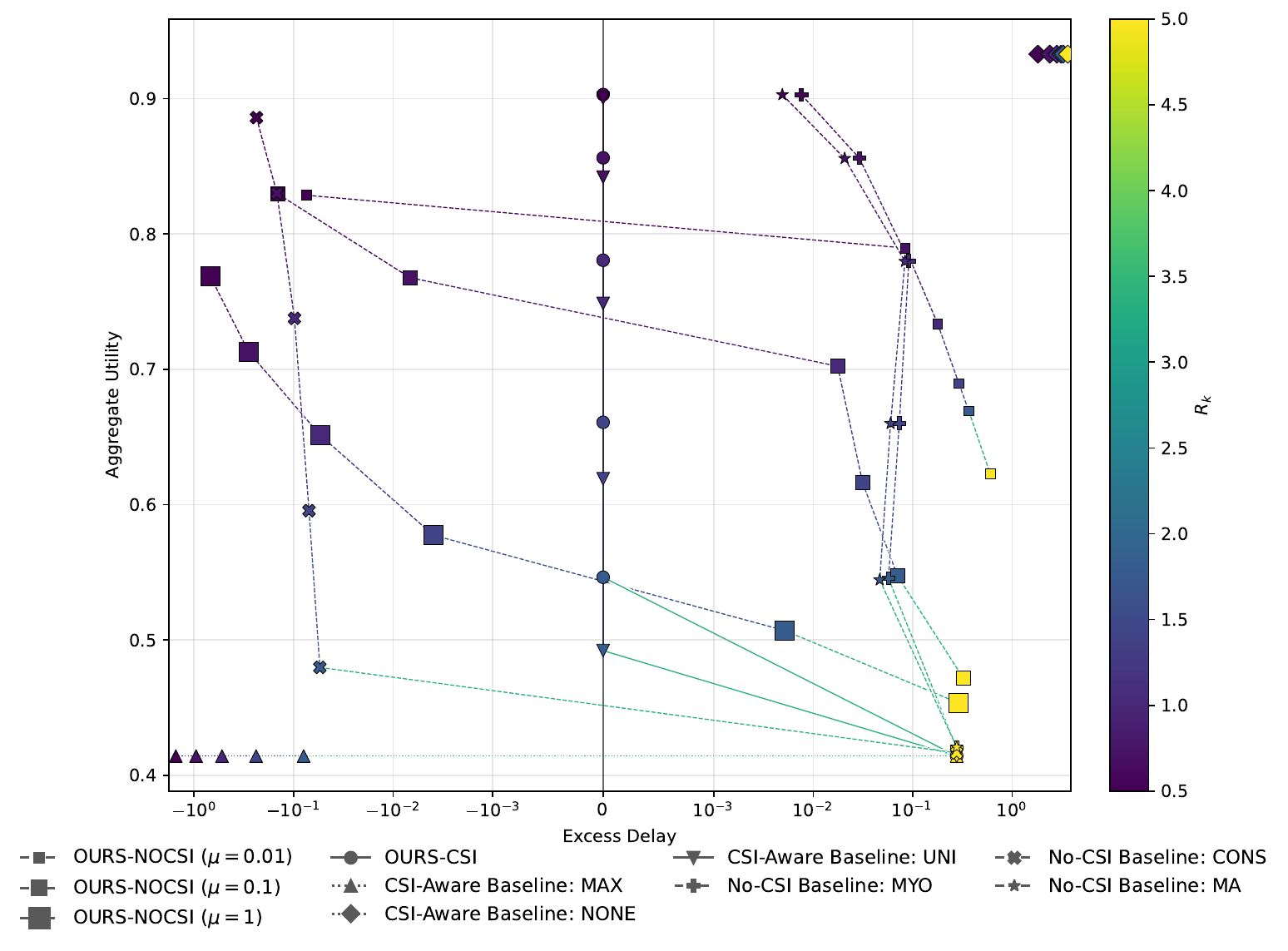}
    \caption{Scenario \scenVisionJetsonS, \taskVisionTwo single task results (Stein's, Top-$k$)}
    \label{fig:topk-resnet-jetson-stein}
\end{figure*}


\begin{figure*}[t]
    \centering
    \includegraphics[width=0.9\textwidth]{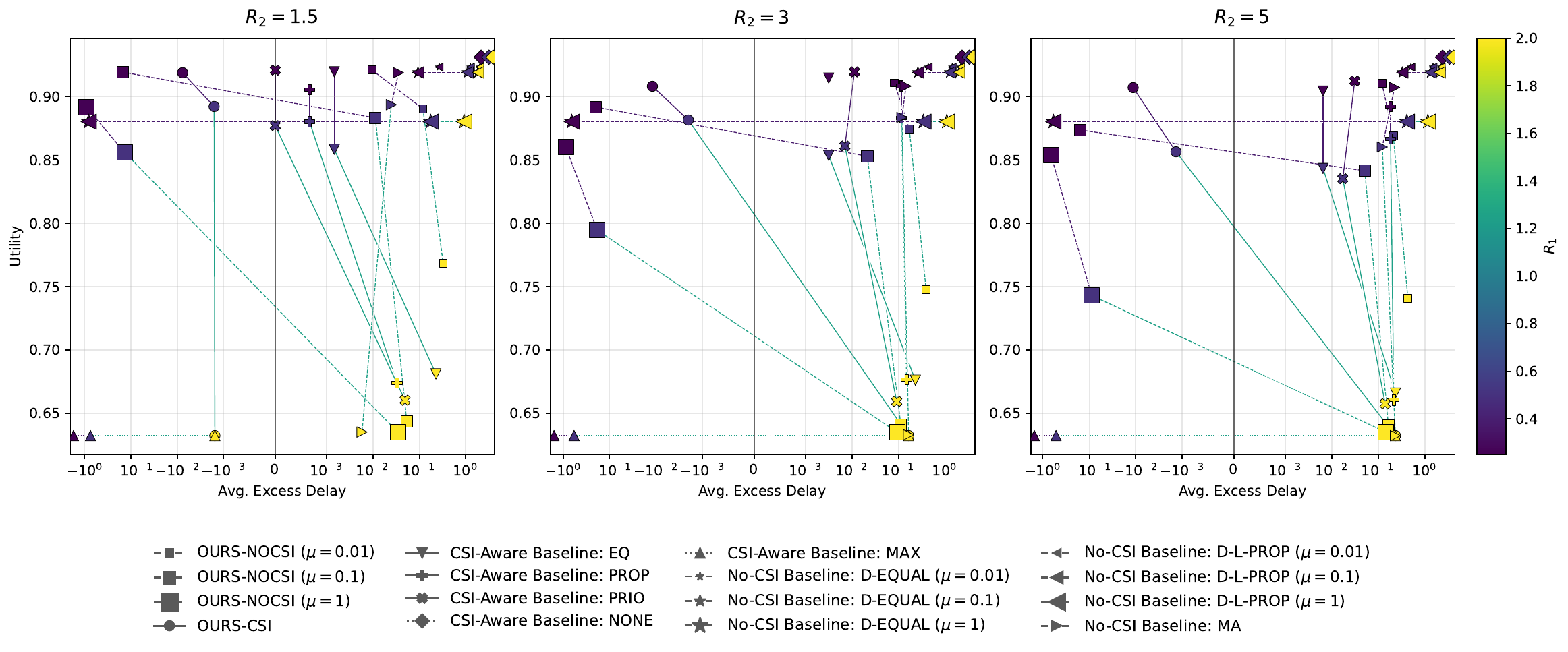}
    \caption{Top-$k$ multi-task tradeoff under the fitting-based estimator. Each panel fixes the Task-2 target rate \(R_2\) and sweeps the Task-1 target rate \(R_1\). The x-axis shows the overall average excess delay, and the y-axis shows the overall utility.}
    \label{fig:multi-topk-fitting-task1-fixed-r2}
\end{figure*}

\begin{figure*}[t]
    \centering
    \includegraphics[width=0.9\textwidth]{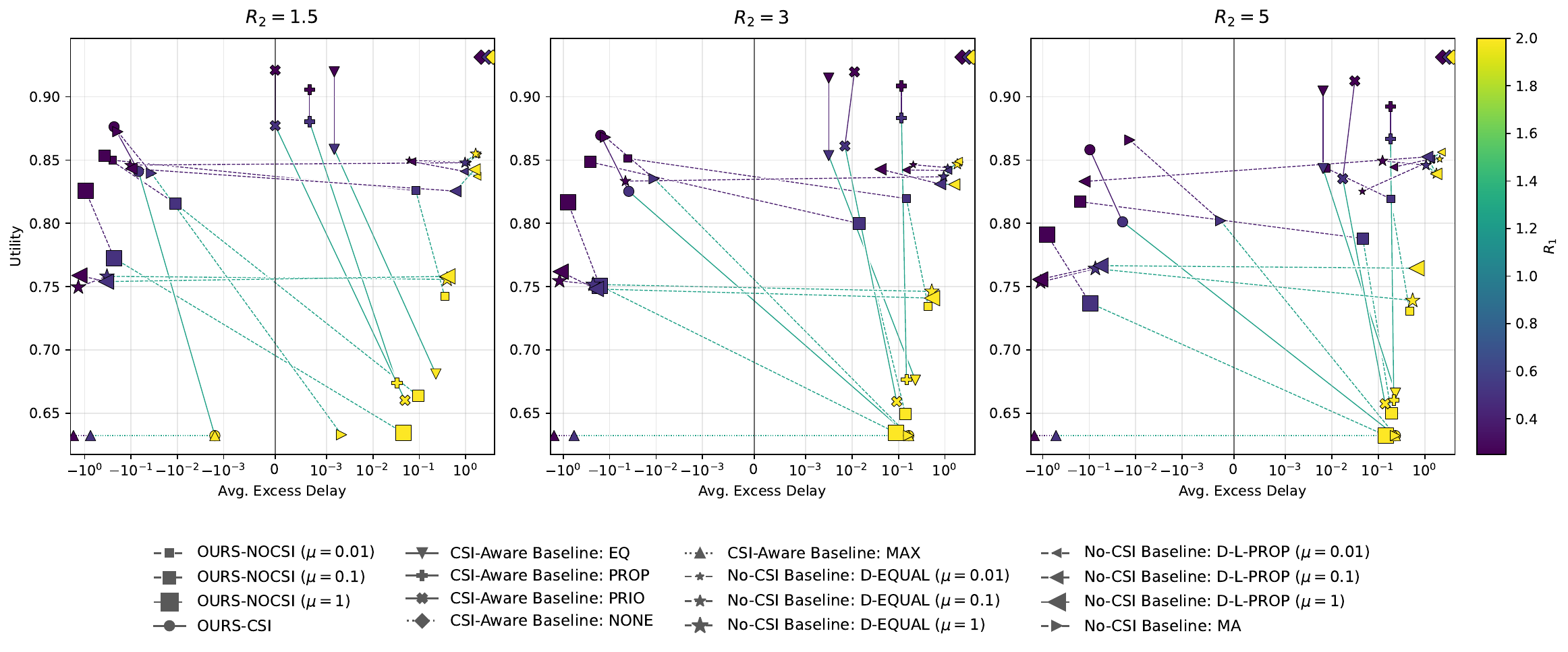}
    \caption{Top-$k$ multi-task tradeoff under Stein's estimator. Each panel fixes the Task-2 target rate \(R_2\) and sweeps the Task-1 target rate \(R_1\). The x-axis shows the overall average excess delay, and the y-axis shows the overall utility.}
    \label{fig:multi-topk-stein-task1-fixed-r2}
\end{figure*}

\subsection{Accuracy-Delay Tradeoff Trajectory under Different $R_k$}

We sweep both $R_k$ and $\mu$, and further compare two accuracy-estimation methods, \emph{fitting} and \emph{kernel-smoothing} (Stein’s method) under scenario \scenVisionJetsonS. For the \emph{Top-$k$} compressor, the resulting tradeoff curves are shown in Figure~\ref{fig:topk-resnet-jetson} and~\ref{fig:topk-resnet-jetson-stein}. We also evaluate the \emph{quantization} and \emph{LLM.int8}.  For these compressors, however, the accuracy decreases only mildly as $\eta$ becomes smaller, and thus no tradeoff plots are included here. These two figures illustrate the single-task accuracy--delay tradeoff. The experiments compare multiple decision policies across six target-rate $R_k$ settings. The x-axis shows the average excess delay, and the y-axis shows the average true accuracy. Each curve corresponds to one policy, and the color of each point indicates the target-rate setting.

First, we observe that our algorithm, especially the No-CSI, performs well across different values of $R_k$. Compared with the baselines, it is more likely to satisfy the target-rate constraint while maintaining higher accuracy. At the same time, there is a clear upper limit on $R_k$: once the target rate exceeds the capability of the compression mechanism, the algorithm fails to satisfy the constraint. This limitation also applies to the baselines.

Then, the comparison between \emph{fitting} and Stein's method suggests that, under our No-CSI algorithm, the fitting-based estimator tends to perform better relative to the baselines. In particular, it often achieves higher accuracy while still satisfying the target delay constraint. A likely reason is that, to reduce the computational overhead of Stein-based estimation, we use relatively small values of $\bm{N}$ and \datasetGradFit{}, which can lead to lower estimation quality for Stein's method.

By examining the curves under different values of $\mu$, we also observe that a larger $\mu$ leads to more aggressive compression, typically at the cost of some accuracy. This is consistent with our theoretical analysis, and can help guide the choice of $\mu$ to balance accuracy against the requirement of satisfying the target rate. If stricter control of target-rate violations is desired, a larger $\mu$ can be used; if a higher degree of violation is acceptable, a smaller $\mu$ can be chosen to achieve higher accuracy.

We further evaluate the multi-task setting with task~1 \taskVisionTwo and task~2 \taskLangSeven under scenario \scenMTLVJM, sweeping both $R_k$ and $\mu$, and comparing \emph{fitting} with \emph{kernel-smoothing} (Stein's method). For \emph{Top-$k$}, the resulting tradeoff curves are shown in Figures~\ref{fig:multi-topk-fitting-task1-fixed-r2} and~\ref{fig:multi-topk-stein-task1-fixed-r2}. In each subfigure, the target rate $R_2$ of task~2 is fixed, while the trajectory is obtained by sweeping the target rate $R_1$ of task~1. The x-axis shows the aggregated excess delay across the two tasks weighted by $w_k$, and the y-axis shows the aggregated utility weighted by $w_k$.

The multi-task results are broadly similar with the single-task case. Our method, especially the No-CSI variant, achieves better overall utility under the delay constraint across a range of target-rate settings, while the fitting-based estimator remains more effective than Stein's method in this testbed setting. The curves under different values of $\mu$ show the expected tradeoff: a larger $\mu$ leads to more conservative behavior and lower violation, while a smaller $\mu$ yields higher utility at the cost of looser delay control.

\begingroup\color{black}
\section{Model Extensions}\label{app:extensions}

Two modeling choices in Sec.~\ref{sec:meth} were deliberately kept simple: the target throughput $R_k(t)$ is a fixed requirement, and the transmitted payload is abstracted as $a_{i,k}\eta_{i,k}(t)$. We show here that both can be generalized without sacrificing the convex structure our solution methods rely on.

\subsection{Target Throughput as a Variable}\label{app:var-throughput}

In Prob.~\ref{prob:multi-task}, every active task $k \in \mathcal{K}(t)$ must meet a fixed target throughput $R_k(t)$. We now treat the served rate itself as a decision variable, which turns the per-slot problem into joint admission control, resource allocation, and compression.

For every task $k \in \mathcal{K}(t)$, let $r_k(t) \in [0, R_k^{\max}(t)]$ denote the admitted throughput allocated to task $k$ during slot $t$, where $R_k^{\max}(t)$ is an upper bound (e.g., the original target rate or the arrival rate). The value $r_k(t)=0$ means that task $k$ is not served during slot $t$, while $r_k(t)>0$ means that task $k$ is served at rate $r_k(t)$.

Replacing $R_k(t)$ by $r_k(t)$ in \eqref{cons:het-comp-lat} and \eqref{cons:het-comm-lat} gives
\begin{align}
\frac{\tau_{i,k}}{s^{\cmp}_{i,k}(t)}
&\le \frac{1}{r_k(t)},
\quad \forall k \in \mathcal{K}(t),~ i \in [L_k],
\label{eq:rate_var_cmp_orig}
\\
\frac{a_{i,k}\eta_{i,k}(t)}{s^{\com}_{i,k}(t)c_{i,k}(t)}
&\le \frac{1}{r_k(t)},
\quad \forall k \in \mathcal{K}(t),~ i \in [L_k-1],
\label{eq:rate_var_com_orig}
\end{align}
or, equivalently,
\begin{align}
\tau_{i,k}\, r_k(t)
&\le s^{\cmp}_{i,k}(t),
\quad \forall k \in \mathcal{K}(t),~ i \in [L_k],
\label{eq:rate_var_cmp_lin}
\\
a_{i,k}\eta_{i,k}(t)\, r_k(t)
&\le s^{\com}_{i,k}(t)c_{i,k}(t),
\quad \forall k \in \mathcal{K}(t),~ i \in [L_k-1].
\label{eq:rate_var_com_bilin}
\end{align}
The compute constraint \eqref{eq:rate_var_cmp_lin} is affine, but \eqref{eq:rate_var_com_bilin} is bilinear due to the product $\eta_{i,k}(t)r_k(t)$: rate and compression are no longer separable, since serving a task faster and transmitting more of its activation compete for the same link budget. Introducing the variables
\begin{equation}
z_{i,k}(t) \triangleq r_k(t)\,\eta_{i,k}(t),
\quad \forall k \in \mathcal{K}(t),~ i \in [L_k-1],
\label{eq:z_def}
\end{equation}
removes the bilinearity. Since $\eta_{i,k}(t) \in [\eta_{i,k}^{\min},1]$, the definition \eqref{eq:z_def} implies
\begin{equation}
\eta_{i,k}^{\min}\, r_k(t)
\le z_{i,k}(t)
\le r_k(t),
\quad \forall k \in \mathcal{K}(t),~ i \in [L_k-1],
\label{eq:z_bounds}
\end{equation}
and \eqref{eq:rate_var_com_bilin} becomes affine:
\begin{equation}
a_{i,k}\, z_{i,k}(t)
\le s^{\com}_{i,k}(t)c_{i,k}(t),
\quad \forall k \in \mathcal{K}(t),~ i \in [L_k-1].
\label{eq:z_comm_lin}
\end{equation}

To account for both admitted throughput and inference quality, we maximize the aggregate admitted utility
\begin{equation}
\textstyle\sum_{k \in \mathcal{K}(t)}
w_k\, r_k(t)\,
A_k\!\left(\frac{\mathbf{z}_k(t)}{r_k(t)}\right),
\label{eq:perspective_obj}
\end{equation}
where $\mathbf{z}_k(t) = (z_{i,k}(t))_{1 \le i \le L_k-1}$. For $r_k(t)>0$, the ratio $\mathbf{z}_k(t)/r_k(t)$ recovers the compression vector $\boldsymbol{\eta}_k(t)$, so \eqref{eq:perspective_obj} weighs each task's accuracy by the rate at which it is actually served. Crucially, the mapping $r_k(t)A_k(\mathbf{z}_k(t)/r_k(t))$ is the \emph{perspective} of $A_k$ and is therefore jointly concave in $(r_k(t),\mathbf{z}_k(t))$ whenever $A_k$ is concave, with the standard convention that it is extended by its closure at $r_k(t)=0$ (an unserved task contributes zero). We thus arrive at the following optimization problem:
\begin{problem}[CSI-Aware Multi-Task, Variable Throughput]\label{prob:var-throughput}
\begin{subequations}\label{prob:var-rate-main}
\begin{align}
\mathop{\text{Maximize:}}_{\substack{\mathbf{r}(t),\,\mathbf{z}(t),\\ \boldsymbol{s}^\cmp(t),\,\boldsymbol{s}^\com(t)}}
&~ \textstyle\sum_{k \in \mathcal{K}(t)}
w_k\, r_k(t)\,
A_k\!\left(\frac{\mathbf{z}_k(t)}{r_k(t)}\right)
\label{eq:var_rate_prob_obj}
\\[0.5ex]
\text{subj.~to:}~~
& \text{Constraints \eqref{eq:comp_cap} and \eqref{eq:comm_cap}}, \nonumber\\
& \tau_{i,k}\, r_k(t)
\le s^{\cmp}_{i,k}(t),
~ \forall k \in \mathcal{K}(t),\, i \in [L_k],
\label{eq:var_rate_prob_cmp}
\\
& a_{i,k}\, z_{i,k}(t)
\le s^{\com}_{i,k}(t)c_{i,k}(t),
\enspace \forall k \!\in\! \mathcal{K}(t),\, i \!\in\! [L_k\!-\!1],
\label{eq:var_rate_prob_com}
\\
& \eta_{i,k}^{\min}\, r_k(t)
\le z_{i,k}(t)
\le r_k(t),
\enspace \forall k \!\in\! \mathcal{K}(t),\, i \!\in\! [L_k\!-\!1],
\label{eq:var_rate_prob_z}
\\
& 0 \le r_k(t) \le R_k^{\max}(t),
~ \forall k \in \mathcal{K}(t),
\label{eq:var_rate_prob_r}
\\
& s^{\cmp}_{i,k}(t) \ge 0, \quad
s^{\com}_{i,k}(t) \ge 0.
\label{eq:var_rate_prob_nonneg}
\end{align}
\end{subequations}
\end{problem}
All constraints of Prob.~\ref{prob:var-throughput} are affine and, by the perspective argument above, the objective is concave; the program is therefore convex and the machinery of Sec.~\ref{sec:meth} applies unchanged. For any optimal solution with $r_k(t)>0$, the compression factors are recovered as
\begin{equation}
\eta_{i,k}(t)=\frac{z_{i,k}(t)}{r_k(t)},
\quad \forall k \in \mathcal{K}(t),~ i \in [L_k-1].
\label{eq:recover_eta}
\end{equation}
When $r_k(t)=0$, task $k$ is not served in slot $t$ and $\boldsymbol{\eta}_k(t)$ may be assigned arbitrarily within its feasible box, as it affects neither the objective nor the constraints.

Two properties are worth noting. First, fixing $r_k(t)=R_k(t)$ for all $k$ recovers Prob.~\ref{prob:multi-task}, so the variable-throughput formulation is a strict relaxation. Second, whereas Prob.~\ref{prob:multi-task} is infeasible whenever the fixed rates exceed what the instantaneous capacities can support (Lem.~\ref{lem:multi-feasibility}), Prob.~\ref{prob:var-throughput} is always feasible: $\mathbf{r}(t)=\mathbf{0}$, $\mathbf{z}(t)=\mathbf{0}$, $\boldsymbol{s}^\cmp(t)=\boldsymbol{s}^\com(t)=\mathbf{0}$ satisfies every constraint. Infeasibility is thus replaced by graceful degradation: the optimizer sheds admitted rate from the tasks whose weighted utility is least sensitive to it, rather than declaring the slot unservable. This is one mechanism mentioned in Sec.~\ref{sec:conc} for the regime in which no compression level suffices.

\subsection{Overhead-Aware Communication Model}\label{app:overhead}

Modeling the transmitted payload as $a_{i,k}\eta_{i,k}(t)$ is a deliberate abstraction: it exposes a continuous communication--accuracy tradeoff and renders the per-slot problem amenable to convex optimization, and identifying and exploiting this structure is one of our contributions. Real compressors, however, also transmit index metadata, are restricted to discrete bit widths, and consume compute to encode and decode. We show that each of these overheads can be folded into the formulation without destroying its structure.

\begin{table}[t]
  \centering
  \color{black}
  \caption{Overhead-aware payload of a single activation transmission with uncompressed size $a$ bits over $n$ coordinates. The metadata term is a bit-packed index mask and is independent of $\eta$, so it enters the communication constraint additively and preserves convexity.}
  \label{tab:overhead-model}
  \small
  \setlength{\tabcolsep}{4pt}
  \begin{tabular}{@{}lccc@{}}
    \toprule
    \textbf{Method} & \textbf{Values} & \textbf{Metadata} & \textbf{Total} \\
    \midrule
    Uncompressed            & $a$      & $0$   & $a$          \\
    \emph{Top-$k$}          & $a\eta$  & $n$   & $a\eta + n$  \\
    Uniform quantization    & $a\eta$  & $0$   & $a\eta$      \\
    \emph{LLM.int8}         & $a\eta$  & $n$   & $a\eta + n$  \\
    \bottomrule
  \end{tabular}
\end{table}

Consider one transmission on hop $i$ of task $k$, carrying an activation of $n_{i,k}$ coordinates whose uncompressed payload is $a_{i,k}$ bits. \emph{Top-$k$} and \emph{LLM.int8} must additionally convey \emph{which} coordinates are retained, respectively treated as outliers; our implementation sends a bit-packed mask of one bit per coordinate, costing $n_{i,k}$ bits (Appendix~\ref{app:compressors}). Uniform quantization requires no mask. Table~\ref{tab:overhead-model} summarizes the resulting payloads. The key structural observation is that the mask size is fixed by the tensor shape and is \emph{independent} of $\eta_{i,k}(t)$, so the communication constraint \eqref{cons:het-comm-lat} generalizes to
\begin{equation}
\frac{a_{i,k}\eta_{i,k}(t)+n_{i,k}}{s^{\com}_{i,k}(t)c_{i,k}(t)}
\le \frac{1}{R_k(t)},
\label{eq:overhead-comm}
\end{equation}
which is still jointly affine in $(\eta_{i,k}(t),s^{\com}_{i,k}(t))$. Convexity is therefore preserved, and every structural result carries over with $a_{i,k}\eta_{i,k}(t)$ replaced by $a_{i,k}\eta_{i,k}(t)+n_{i,k}$: the feasibility condition of Lem.~\ref{lem:multi-feasibility} becomes $\sum_{k \in \mathcal{K}(t)} \sum_{i:\,e_{i,k}=e} (a_{i,k}\eta_{i,k}^{\min}+n_{i,k})R_k(t)/c_{i,k}(t) \le 1$, and the single-task closed form of Thm.~\ref{thm:single-convex} becomes $\eta_i^\star(t)=\min\{1,(c_i(t)/R(t)-n_i)/a_i\}$. Intuitively, the mask consumes a fixed share of the link budget, leaving the remainder to be traded against accuracy exactly as before.

Finally, compression and decompression are not free on edge devices. We measure them and include them in the profiled stage time $\tau_{i,k}$, which is consequently compressor-specific; this is why $\boldsymbol{\tau}$ varies across compressors in our Jetson profiles while $\boldsymbol{a}$, being determined by tensor shape alone, does not (Appendix~\ref{app:scenarios}).
\endgroup

}{}

\end{document}